\documentclass[12pt]{iopart}

\usepackage{iopams}

\expandafter\let\csname equation*\endcsname\relax

\expandafter\let\csname endequation*\endcsname\relax

\usepackage{amsmath}
\usepackage{amsthm}
\usepackage[utf8]{inputenc}
\usepackage{amssymb}
\usepackage{array}
\usepackage{mathtools}
\usepackage{enumitem}
\usepackage{accents}

\newcommand{\utilde}[1]{\underaccent{\tilde}{#1}}

\usepackage{hyperref}

\numberwithin{equation}{section}
\numberwithin{figure}{section}

\usepackage{tikz-cd}
\usetikzlibrary{babel}

\usepackage [english]{babel}
\usepackage [autostyle, english = american]{csquotes}
\MakeOuterQuote{"}

\usepackage[bottom]{footmisc}
\usepackage[justification=centering]{caption}

\usepackage{etoolbox}

\makeatletter
\newcommand{\mainmatter}{%
  \setcounter{footnote}{0}%
  \patchcmd{\@makefntext}{\fnsymbol}{\arabic}{}{}%
  \patchcmd{\@thefnmark}{\fnsymbol}{\arabic}{}{}%
  \def\@makefnmark{\textsuperscript{\arabic{footnote}}}%
}
\makeatother

\theoremstyle{plain}%
\newtheorem{theorem}{Theorem}[section]
\newtheorem{proposition}[theorem]{Proposition}%
\newtheorem{lemma}[theorem]{Lemma}

\theoremstyle{remark}%
\newtheorem{remark}[theorem]{Remark}%

\theoremstyle{definition}%
\newtheorem{definition}[theorem]{Definition}%
\newtheorem{example}[theorem]{Example}%

\begin{document}

\title[Massive single-particle bundles]{Bundle Theoretic Descriptions of Massive Single-Particle State Spaces; With a view toward Relativistic Quantum Information Theory}

\author{Heon Lee}

\address{Department of Mathematical Sciences, Seoul National University, 1, Gwanak-ro, Gwanak-gu, Seoul, 08826, Republic of Korea}
\ead{heoney93@gmail.com}
\vspace{10pt}
\begin{indented}
\item[]March 2022
\end{indented}

\begin{abstract}
Relativistic Quantum Information Theory (RQI) is a flourishing research area of physics, yet, there has been no systematic mathematical treatment of the field. In this paper, we suggest bundle theoretic descriptions of massive single-particle state spaces, which are basic building blocks of RQI. In the language of bundle theory, one can construct a vector bundle over the set of all possible motion states of a massive particle, in whose fibers the moving particle's internal quantum state as perceived by a fixed inertial observer is encoded. A link between the usual Hilbert space description is provided by a generalized induced representation construction on the $L^2$-section space of the bundle. The aim of this paper is two-fold. One is to communicate the basic ideas of RQI to mathematicians and the other is to suggest an improved formalism for single-particle state spaces that encompasses all known massive particles including those which have never been dealt with in the RQI literature. Some of the theoretical implications of the formalism will be explored at the end of the paper.
\end{abstract}

%
%
\submitto{\jpa}
%
%
%

\mainmatter

\section{Introduction}
\label{sec:1}

Special Relativity (SR) is a principle by which two inertial observers' perceptions of the laws of physics must agree. When incorporated into Quantum Mechanics (QM), this principle manifests itself as a quantum symmetry, which is expressed as a unitary representation of the spacetime symmetry group $G:= \mathbb{R}^4 \ltimes SL(2, \mathbb{C})$ on a quantum Hilbert space. Given a quantum Hilbert space possessing this symmetry, one can describe how one inertial observer's perception of the quantum system is related to another inertial observer's perception of it.

This symmetry principle gave birth to the notion "Single-particle state spaces" \cite{wigner}, which are just the irreducible unitary representation spaces of the group $G$ and classified by two numerical invariants called \textit{mass} and \textit{spin} for the massive particle case (cf. Sect.~\ref{sec:4}). So, one might say that the single-particle state spaces are the smallest possible quantum systems in which comparisons between different observers' perceptions of one reality are possible.

However, when incorporated into Quantum Information Theory (QIT) scenarios, this principle, which is indispensable for a complete physical theory, caused some unexpected phenomena. For example, \cite{peres2002} observed that the spin entropy of a massive particle with spin-1/2 is an observer-dependent quantity (it can be zero in one frame and at the same time does not vanish in another frame) and \cite{gingrich2002} showed that the same conclusion holds with the spin entanglement between two massive spin-1/2 particles, which had immediate consequences on the relativistic Einstein-Podolsky-Rosen type experiment where one deals with two particles maximally entangled in the spin degree of freedom (\cite{caban2005, caban2006, gingrich2002}). See Sect.~\ref{sec:3.1} for a brief account of these observations.

An implication that these entailed was that when one wants to use the spin of a massive spin-1/2 particle (such as electron) as an information carrier (i.e., a qubit carrier), the concepts of entropy, entanglement, and correlation of the spin may require a reassessment \cite{peres2004}, which are important informational resources in QIT.

The work \cite{peres2002} has generated an intense study (\cite{peres2002, gingrich2002, terno2003, gingrich2003, caban2003, peres2004, caban2005, caban2006, he2007, caban2008a, caban2008b, caban2012, alsing2012, caban2013, saldanha2012a, saldanha2012b, derbarba2012, avelar2013, aghaee2017, bittencourt2018, bittencourt2019, caban2018, caban2019, ondra2021, lee2022}) that is still going on today.\footnote{However, the author strongly believes that the perplexity posed in \cite{peres2002} has finally reached a definitive clarification in \cite{lee2022}. In fact, the publication of \cite{lee2022} was one of the main motivations for the conception of the theory developed in this paper. See Sect.~\ref{sec:3}.} A great portion of these works deal with the concepts of relativistic entanglement and correlation of the discrete degrees of freedom (such as spin) between relativistic particles in various settings, using various measures of correlation. However, as far as we know, a relativistically invariant definition of entanglement between the spins of a multi-particle system is still missing \cite{alsing2012}.

Given these conceptual profundities and prospects, it is a curious fact that there has never appeared a systematic mathematical treatment of this field. One reason for this might be because there is already a nice treatment of the single-particle state space in the physics literature (e.g., \cite{weinberg}). But, the recent paper \cite{lee2022} claimed that the above-mentioned issues arise because there is an inherent conceptual problem in this standard treatment. So, we feel that it is the right time to suggest a new mathematical framework for the single-particle state space that is more suitable for RQI investigations.

The problem with the standard treatment can never be seen clearly when one uses the usual language of Hilbert spaces and operators to describe single-particle state space. However, there is a bundle theoretic way to view single-particle state space, in which the stated conceptual problem becomes transparently visible and is easily resolved. This point of view was first introduced in \cite{lee2022} for the massive spin-1/2 particle case and thoroughly exploited to give a definitive mathematical clarification of the perplexity posed in \cite{peres2002}.

In this picture, there is a vector bundle responsible for the description of a massive particle with spin-1/2, which is an assembly of two-level quantum systems corresponding to possible motion states of the particle, whose fibers are arranged in a way that reflects the perception\footnote{The precise meaning of "perception" used in this paper is given in Sect.~\ref{sec:2.5}.} of a fixed inertial observer who has prepared the bundle in the first place for the description of the state of the particle. So, given a motion state $p$, the fiber over $p$ is what the fixed observer perceives as the spin quantum system of the particle in the motion state $p$. Therefore, this viewpoint also provides us with the precise mathematical description of moving qubit systems as perceived by a fixed inertial observer.

Moreover, there is a naturally defined $G$-action on this bundle, which makes it a $G$-vector bundle. An action of the element $(a, \Lambda) \in G$ on the bundle amounts to a frame change by the transformation $(a, \Lambda)$. That is, the bundle description of an inertial observer who is $(a, \Lambda)$-transformed with respect to a fixed inertial observer is obtained by applying the $(a, \Lambda)$-action on the fixed observer's bundle description.

In this sense, the vector bundle description is similar to the classical coordinate system in which one records a particle's motion in the spacetime by four numerical values and for which a definite transformation law from one observer to another is given. The vector bundle description is just an extension of it which takes the particle's internal quantum states into account. We will see that the transformation law of the bundle description (i.e., the stated $G$-action) naturally extends that of the coordinate system.

It is the objective of the present paper to develop a mathematical theory that underlies this vector bundle point of view and generalize this point of view to all known massive particles (i.e., massive particles with arbitrary spin). After completing this job, we will explore some of the theoretical implications of this viewpoint. Specifically, we will see that the Dirac equation and the Proca equations are manifestations of a fixed inertial observer's perception of the internal quantum states of massive particles with spin-1/2 and 1, respectively.

This paper is organized as follows. In Sect.~\ref{sec:2}, we explain briefly how the ideas of SR come into play in the quantum realm, giving the definition of \textit{quantum system with Lorentz symmetry}, which is the right playground for testing special relativistic considerations in QM. In Sect.~\ref{sec:3}, after defining single-particle state space, we briefly survey the perplexities posed by some of the pioneering works of RQI and summarize the main result of the paper \cite{lee2022}, in which the problem with the standard approach of RQI, which is responsible for the mentioned perplexities, is clarified and resolved for the spin-1/2 case.

In Sect.~\ref{sec:4}, we embark on the job of extending the vector bundle point of view, which was first suggested in \cite{lee2022} for the spin-1/2 case, to arbitrary spin case. Specifically, we identify the massive single-particle state spaces and classify them by two numerical invariants called \textit{mass} and \textit{spin}. All the results of this section is well-known and included here for completeness. In Sect.~\ref{sec:5}, we develop a mathematical theory that underlies the bundle theoretic framework that this paper suggests.

In Sect.~\ref{sec:6} we present the promised vector bundle point of view for massive particles with arbitrary spin and show that the same problem as in the spin-1/2 case persists in the general spin case and is resolved in a similar manner. Sect.~\ref{sec:7} explores some of the theoretical implications of the present work. Concluding remarks and future research directions are given in Sect.~\ref{sec:8}.

\section{Special Relativity in Quantum Mechanics}\label{sec:2}

In this section, we briefly summarize the idea and formalism of Relativistic Quantum Mechanics that is used in the physics literature. The main references for this section are \cite{weinberg} pp.49--55 and \cite{folland2008} pp.39--40.

\subsection{Notations}\label{sec:2.1}
In this section, we summarize some notations and elementary facts that will be used throughout the paper.

Let $x \in \mathbb{R}^4$. We write its 4 components as $x= (x^0, x^1, x^2, x^3 ) $ where $x^0$ is the time component. When we want to deal with each component, we use Greek indices such as $x^\mu$  $(\mu = 0 , 1, 2 , 3 ) $ and if we need only spatial components, we use Latin indices such as $ x^j $  $ (j = 1, 2, 3)$. When we want to separate time and spatial components, we also use the convention $ x = (t, \mathbf{x} ) $. We set $\hbar = c = 1$.

When we encounter an expression with subscripts rather than superscripts as above, it must be understood as, for example, $ x_\mu : = \sum_{\nu = 0} ^3 \eta_{\mu \nu} x^\nu $, where $\eta= \text{diag} (1, -1 , -1 , -1)$.

We also use the Einstein summation convention. So that the above becomes $x_\mu = \eta_{\mu \nu} x^\nu$ and also, for example, $p_\mu p^\mu = \sum_{\mu, \nu = 0} ^ 3 p^\mu \eta_{\mu \nu} p ^\nu $ holds.

The Minkowski metric $\eta$ on $\mathbb{R}^4$ is also denoted as
\begin{equation}\label{eq:2.1}
\langle x , y \rangle = \eta_{\mu \nu} x^\mu y^\nu = x_\mu y^\mu = x^0 y^0 - x^1 y^1 - x^2 y^2 - x^3 y^3.
\end{equation}

Let's denote the Pauli matrices as
\begin{equation}\label{eq:2.2}
\tau^0 = I, \hspace{0.1cm} \tau^1 = \begin{pmatrix} 0 & 1 \\ 1 & 0 \end{pmatrix}, \hspace{0.1cm} \tau^2 = \begin{pmatrix} 0 & -i \\ i & 0 \end{pmatrix}, \hspace{0.1cm}  \tau^3 = \begin{pmatrix} 1 & 0 \\ 0 & -1 \end{pmatrix}.
\end{equation}

Denote
\begin{subequations}\label{eq:2.3}
\begin{eqnarray}
\tilde{x} &=& \begin{pmatrix} x^0 + x^3 & x^1 - i x^2 \\ x^1 + i x^2 & x^0 - x^3 \end{pmatrix} = x^0 \tau^0 + \mathbf{x} \cdot \boldsymbol{\tau}
 \\
\utilde{x} &=& \begin{pmatrix} x_0 + x_3 & x_1 - i x_2 \\ x_1 + i x_2 & x_0 - x_3  \end{pmatrix} = x_{\mu} \tau^{\mu}.
\end{eqnarray}
\end{subequations}

Note that the maps $\tilde{ (\cdot) }, \utilde{(\cdot )} : \mathbb{R}^4 \rightarrow H_2 $ are $\mathbb{R}$-linear isomorphisms from $\mathbb{R}^4$ onto the space of $2 \times 2 $ Hermitian matrices $H_2$.\footnote{These notations were borrowed from the book \cite{bleecker} with a slight modification.}

A direct calculation would show
\begin{subequations}\label{eq:2.4}
\begin{equation}\label{eq:2.4a}
\utilde{x}\tilde{y} + \utilde{y}\tilde{x} = 2 \langle x,  y \rangle I_2 = \tilde{x} \utilde{y} + \tilde{y} \utilde{x}
\end{equation}
and hence
\begin{equation}\label{eq:2.4b}
\utilde{x}\tilde{x} = \langle x, x \rangle I_2 = \tilde{x} \utilde{x}.
\end{equation}
\end{subequations}

\subsection{Physical Symmetry}\label{sec:2.2}

Let $(\mathcal{H} , \langle \cdot,\cdot \rangle)$ be a Hilbert space associated with a quantum system. From the axioms of QM (cf. \cite{hall}), we know that two state vectors $\phi, \psi \in \mathcal{H}$ represent the same physical state if and only if $\phi = \lambda \psi$ for some $\lambda \in \mathbb{C} \setminus \{0 \}$. So, denoting this equivalence relation as $\sim$, the "quantum states" are in fact elements of $\mathbb{P} (\mathcal{H}) := (\mathcal{H} \setminus \{0\}) / \sim$, the projectivization of $\mathcal{H}$.
\begin{definition}\label{definition:2.1}
There is a well-defined map $\mathbb{P} (\mathcal{H}) \times \mathbb{P} (\mathcal{H}) \rightarrow [0,1]$ defined as
\begin{equation}\label{eq:2.5}
([u]_{\sim} , [v]_{\sim} ) \mapsto \left(\frac{ | \langle u,v \rangle | } {\|u\| \|v \|} \right)^2 ,
\end{equation}
called the \textit{transition probability between $[u]$ and $[v]$}, and denoted as $([u],[v])$ (I will omit the $\sim$ signs from now on).
\end{definition}

If a system is in a state represented by $[u] \in \mathbb{P}(\mathcal{H})$, the probability of finding it in the state represented by $[v] \in \mathbb{P}(\mathcal{H})$ (using a certain measurement which has $v$ as an eigenstate) is $([u],[v])$ (cf. \cite{weinberg}).
\begin{definition}\label{definition:2.2}
A bijective map $T : \mathbb{P} (\mathcal{H}) \rightarrow \mathbb{P} (\mathcal{H})$ that preserves transition probability (i.e., $(T[u], T[v]) = ([u],[v])$ for all $u,v \in \mathcal{H}$) is called a \textit{physical symmetry}.
\end{definition}

\begin{example}\label{example:2.3}
Let $\mathcal{H} := L^2 (\mathbb{R}^3) \otimes L^2 (\mathbb{R}^3)$. Define $\sigma : \mathcal{H} \rightarrow \mathcal{H}$ as $(\sigma f) (x_1,x_2) = f(x_2,x_1)$ for $f \in \mathcal{H}$. Then, $\sigma$ is a unitary transformation and hence induces a well-defined map $\overline{\sigma} : \mathbb{P} (\mathcal{H}) \rightarrow \mathbb{P} (\mathcal{H})$ such that the following diagram commutes.
    \begin{equation*}
    \begin{tikzcd}[baseline=(current  bounding  box.center)]
\mathcal{H}  \setminus \{ 0 \} \arrow[r, "\sigma"] \arrow[d ]
& \mathcal{H} \setminus \{ 0 \} \arrow[d]  \\
\mathbb{P} (\mathcal{H})  \arrow[r, "\overline{\sigma}"]
& \mathbb{P} (\mathcal{H}) 
\end{tikzcd}
\end{equation*}

Since $\sigma$ was unitary, $\overline{\sigma}$ is a bijection and preserves transition probability. So, $\overline{\sigma}$ is a physical symmetry.

What is the significance of this permutation symmetry? Let $ \Phi := \phi \otimes \psi \in \mathcal{H} $ be a state describing two independent particles in $\mathbb{R}^3$ whose states are represented by $\phi$ and $\psi$, respectively. Now, a simple calculation shows $\sigma (\Phi) = \psi \otimes \phi$. So, the original state $\Phi (x_1,x_2) = \phi (x_1) \psi (x_2)$ becomes the transformed state $[\sigma \Phi] (x_1,x_2) = \phi (x_2) \psi (x_1)$. So, we see that acting the physical symmetry $\overline{\sigma}$ on the state $[\Phi] \in \mathbb{P} (\mathcal{H})$ amounts to "labeling in a different manner" the two particles.
\end{example}

\begin{remark}\label{remark:2.4}
This example gives a general insight into how we should interpret physical symmetries. Given a physical symmetry $T : \mathbb{P} (\mathcal{H}) \rightarrow \mathbb{P} (\mathcal{H})$, we hypothesize two observers $A, A'$ whose observations on the same quantum Hilbert space $\mathcal{H}$ are related by $T$ so that if a system is in the state $[\psi] \in \mathbb{P} (\mathcal{H})$ in $A$'s frame, then the system is in the state $T[\psi] \in \mathbb{P} (\mathcal{H})$ in $A'$'s frame.

In this respect, a physical symmetry is not an operation that we can perform on the physical system, but, rather, it gives us information about how two observers' observations are related (\cite{weinberg}, pp.50--51).
\end{remark}

\begin{example}\label{example:2.5}
The same analysis of Example~\ref{example:2.3} can be applied to any unitary or antiunitary\footnote{$U: \mathcal{H} \rightarrow \mathcal{H}$ is called antiunitary if $U$ is conjugate linear and $\langle U \phi , U \psi \rangle = \langle \psi, \phi \rangle$ for ${}^\forall \phi, \psi \in \mathcal{H}$.} transformation $U : \mathcal{H} \rightarrow \mathcal{H}$, yielding a physical symmetry
\begin{equation*}
    \begin{tikzcd}[baseline=(current  bounding  box.center)]
\mathcal{H} \arrow[r, "U"] \arrow[d ]
& \mathcal{H} \arrow[d]  \\
\mathbb{P} (\mathcal{H})  \arrow[r, "\overline{U}"]
& \mathbb{P} (\mathcal{H}) 
\end{tikzcd}.
\end{equation*}
\end{example}

It is Wigner's famous theorem that asserts that the converse of Example~\ref{example:2.5} is also true.
\begin{theorem}[Wigner]\label{theorem:2.6}
Given a Hilbert space $\mathcal{H}$ and a physical symmetry $T : \mathbb{P} (\mathcal{H}) \rightarrow \mathbb{P} (\mathcal{H})$, there exists a map $U(T) : \mathcal{H} \rightarrow \mathcal{H}$ that is either unitary or antiunitary such that the following diagram commutes.
\begin{equation}\label{eq:2.6}
\begin{tikzcd}[baseline=(current  bounding  box.center)]
\mathcal{H} \arrow[r, "U(T)"] \arrow[d ]
& \mathcal{H} \arrow[d]  \\
\mathbb{P} (\mathcal{H})  \arrow[r, "T"]
& \mathbb{P} (\mathcal{H}) 
\end{tikzcd}
\end{equation}
\end{theorem}
\begin{proof} For a proof, see \cite{weinberg} p.91.  
\end{proof}

If we denote the set of all antiunitary maps on $\mathcal{H}$ as $U^* (H)$, the set $U(\mathcal{H}) \amalg U^* (\mathcal{H})$ becomes a topological group with strong operator topology and contains $\mathbb{T}:= \{ \lambda I_{\mathcal{H}} : \lambda \in \mathbb{C}, |\lambda | = 1\}$ as a closed normal subgroup. Note that $U(\mathcal{H})$ is the identity component of this group since it is connected in the strong operator topology (cf. \cite{rudin3}).

According to Winger's theorem, the set of all physical symmetries on $\mathcal{H}$ is precisely
\begin{equation}\label{eq:2.7}
S(\mathcal{H}) = \left( U(\mathcal{H}) \amalg U^* (\mathcal{H}) \right)/\mathbb{T}.
\end{equation}

Note that the image of $U(\mathcal{H})$ in the quotient space, the projective unitary group $PU(\mathcal{H}) := U(\mathcal{H})/\mathbb{T}$, is the identity component of the quotient topological group $S(\mathcal{H})$.

\subsection{Lorentz Symmetry}
\label{sec:2.3}

SR is most elegantly described as the "Geometry of the Minkowski spacetime $(\mathbb{R}^4, \eta )$" where $\eta = \text{diag} (1, -1, -1 , -1)$. Throughout, let $M$ denote this pseudo-Riemannian manifold.

\begin{definition}\label{definition:2.7}
An \textit{inertial frame of reference} is a global orthonormal coordinate chart $\varphi = (x^\mu) : M \rightarrow \mathbb{R}^4 $ on which the coordinate representation of the metric $\eta$ is given by $\eta_{\mu \nu} dx^\mu dx^\nu$ (i.e., $(\varphi^{-1})^* \eta = \eta_{\mu \nu} dx^\mu dx^\nu$ in $\Lambda^2 (T^* \mathbb{R}^4)$).
\end{definition}

Choosing an inertial frame of reference means setting up a coordinate system $\varphi = (t,x,y,z)$ in which each point $E \in M$ (called an event) is recorded as the $4$ numerical values $\varphi(E)$.

Suppose Alice and Bob have chosen inertial frames of reference $\varphi_A = (t_A,x_A,y_A,z_A) $ and $\varphi_B = (t_B,x_B,y_B,z_B)$ respectively. Then, by definition, $\varphi_B \circ \varphi_A ^{-1} : \mathbb{R}^4 \rightarrow \mathbb{R}^4 $ is an element of the isometry group of $(\mathbb{R}^4 , \eta $), which is called the \textit{Poincar\'e group}. If we denote the linear isometry group of $(\mathbb{R}^4, \eta)$ (called the \textit{Lorentz group}) as
\begin{equation}\label{eq:2.8}
O(1,3) := \left \{ \Lambda \in GL_4 (\mathbb{R}) : \Lambda^\intercal \eta \Lambda = \eta \right\},
\end{equation}
then the Poincar\'e group is given by
\begin{equation}\label{eq:2.9}
P(4) := \mathbb{R}^4 \ltimes O(1,3)
\end{equation}
where the semidirect product is taken with respect to the natural action of $O(1,3)$ on the abelian group $\mathbb{R}^4$.\footnote{The group multiplication is thus given by $(a, \Lambda) (a' , \Lambda') = (a + \Lambda a' , \Lambda \Lambda')$. Obviously, the identity is $(0, I_4)$ and $(a, \Lambda)^{-1} = (-\Lambda^{-1} a , \Lambda^{-1})$.}

The study of the group $P(4)$ is essential in Relativistic Quantum Mechanics, which is succinctly summarized in Arthur Jaffe's note \cite{jaffe}. I have taken the following result from the note, which will be needed throughout.

\begin{theorem}\label{theorem:2.8}
Let $SO^\uparrow (1,3)$ be the connected component of $O(1,3)$. Then, $O(1,3)$ has four connected components given by
\begin{equation}\label{eq:2.10}
    O(1,3) = SO^{\uparrow} (1,3) \hspace{0.05cm} \amalg \hspace{0.05cm} \mathcal{P} \hspace{0.05cm} SO^{\uparrow} (1,3) \hspace{0.05cm} \amalg \hspace{0.05cm} \mathcal{T} \hspace{0.05cm} SO^{\uparrow} (1,3) \hspace{0.05cm} \amalg \hspace{0.05cm} \mathcal{P T} \hspace{0.05cm} SO^{\uparrow} (1,3)
\end{equation}
where $\mathcal{P} := \textup{diag} (1,-1,-1,-1)$ and $\mathcal{T} := \textup{diag} (-1,1,1,1)$ are called \textbf{parity inversion} and \textbf{time reversal}, respectively. 
\end{theorem}

Usually, one restricts attention to the connected component $\mathbb{R}^4 \ltimes SO^\uparrow (1,3)$ of $P(4)$ (there are physical reasons for this. Cf., \cite{weinberg} p.75) and we will follow this practice in this paper.

So, the two observers' records of an arbitrary event $E \in M$, namely, $\varphi_A (E) = (t_A ,x_A ,y_A ,z_A )$ and $\varphi_B (E) = (t_B,x_B,y_B,z_B)$, are related by an element $ (a, \Lambda):= \varphi_B \circ \varphi_A ^{-1} \in \mathbb{R}^4 \ltimes SO^\uparrow (1,3)$ such that
\begin{equation}\label{eq:2.11}
\begin{pmatrix}t_B \\ x_B \\ y_B \\ z_B \end{pmatrix} = a + \Lambda \begin{pmatrix} t_A \\ x_A \\ y_A \\ z_A \end{pmatrix}
\end{equation}
which may be expressed as $x_B ^\mu = a^\mu + \Lambda^\mu _\nu x_A ^\nu$ (cf. Sect.~\ref{sec:2.1}).

The postulates of SR require that every physical law and entity has an invariant meaning under this kind of coordinate transformation. Mathematically, this just means that physical entities should be objects living in the manifold $M$ and physical laws should be equations defined on the manifold $M$ which are independent of the choice of inertial frames of reference. So that, for example, the vacuum Maxwell's equations of Electrodynamics can be written as
\begin{eqnarray*}
dF = 0 \\
d \star F = 0
\end{eqnarray*}
for some $2$-form $F$ (called the \textit{electromagnetic tensor}) on the manifold $M$ where $d$ is the differential on $M$ and $\star$ is the Hodge star operator on $M$ (cf. \cite{bleecker}).

But, this requirement is usually expressed in coordinate representations in physics textbooks. For example, suppose a tangent vector $v \in TM$ has component representations $(v_A ^\mu)$ and $(v_B ^{ \mu})$ in Alice's frame and Bob's frame, respectively. Then, the transformation law Eq.~(\ref{eq:2.11}) gives the relation
\begin{equation}\label{eq:2.12}
v_B ^{\mu} = \Lambda_\nu ^\mu v_A ^\nu.
\end{equation}

Conversely, if any two inertial frames, $(x_A ^\mu)$ and $(x_B ^{\mu})$, connected by a transformation $(a,\Lambda) \in P(4)$ as in Eq.~(\ref{eq:2.11}) observed a vector quantity (e.g. a velocity of a particle) as $(v_A ^\mu)$ and $(v_B ^{\mu})$ in their respective frames and found the relation $v_B ^{\mu} = \Lambda_\nu ^\mu v_A ^\nu$, then they would conclude that the vector quantities are manifestations of an object living in $TM$ in their respective coordinate systems, i.e., it has meaning independent of the choice of inertial frames.

Any tensorial quantity transforming in this fashion from one inertial frame to another (e.g., $F_B ^{\mu \nu} = \Lambda_{\alpha} ^{\mu} \Lambda_{\beta} ^{\nu} F_A ^{\alpha \beta}$ for the electromagnetic tensor) is called a \textit{Lorentz covariant} tensor and can be regarded as elements of a tensor bundle on $M$. This is the usual way that physicists take to express the fact that a quantity has an invariant meaning in all inertial frames of reference.

How does this principle affect the description of QM? In special relativistic scenarios, one is interested in two inertial observers' perceptions of one physical reality. Accordingly, consider two inertial observers Alice and Bob, whose classical observations are related by Eq.~(\ref{eq:2.11}), who are now interested in the investigation of a quantum system described by the states in the Hilbert space $\mathcal{H}$. We naturally expect that there is a certain transformation $\overline{U}(a, \Lambda)$ on $\mathbb{P} (\mathcal{H})$ which depends on the Lorentz transformation $(a,\Lambda) := \varphi_B \circ \varphi_A ^{-1} $ such that whenever Alice perceives a quantum state $[\psi] \in \mathbb{P}(\mathcal{H})$, Bob would perceive it as $\overline{U} (a, \Lambda) [\psi] \in \mathbb{P} (\mathcal{H})$.

Recalling Definition~\ref{definition:2.1}, the principle of SR naturally requires that the two inertial observers should obtain the same transition probability. In view of Definition~\ref{definition:2.2}, this means that $\overline{U} (a, \Lambda)$ is a physical symmetry. I.e., \textit{the Lorentz transformations act as physical symmetries on a quantum system}.

By Theorem~\ref{theorem:2.6}, this implies the existence of a map $\overline{U}:\mathbb{R}^4 \ltimes SO^{\uparrow} (1,3) \rightarrow S(\mathcal{H})$ (cf. Eq.~(\ref{eq:2.7})). A moment's thought suggests that it is natural to require that $\overline{U}$ be a continuous group homomorphism.\footnote{E.g., the change of reference frame from Alice to Bob and then again to Alice should be the identity transformation and that "nearly same" inertial reference frames should observe "nearly same" quantum states, etc. See \cite{weinberg} pp.50--52 for details.} Therefore, the range of $\overline{U}$ is entirely contained in the identity component $PU(\mathcal{H})$ since the group $\mathbb{R}^4 \ltimes SO^{\uparrow} (1,3)$ is connected. In short, the principle of SR gives us a projective representation $\overline{U}:\mathbb{R}^4 \ltimes SO^{\uparrow} (1,3) \rightarrow PU(\mathcal{H})$.

We can ask a question at this point. As in Wigner's theorem, can we lift the representation to $\mathbb{R}^4 \ltimes SO^{\uparrow} (1,3) \rightarrow U(\mathcal{H})$ so that the following diagram holds?
\begin{equation*}
    \begin{tikzcd}[baseline=(current  bounding  box.center)]
{}
& U(\mathcal{H}) \arrow[d]  \\
\mathbb{R}^4 \ltimes SO^{\uparrow} (1,3)  \arrow[ru] \arrow[r]
& PU (\mathcal{H}) 
\end{tikzcd}
\end{equation*}

The answer is "No" in general. But, Bargman's theorem (\cite{folland2008}, p.40) asserts that passing to the universal cover $\mathbb{R}^4 \ltimes SL(2,\mathbb{C})$ of $\mathbb{R}^4 \ltimes SO^{\uparrow} (1,3)$, we always get a unitary representation. To see this, we must first identify a covering map $SL(2, \mathbb{C}) \rightarrow SO^\uparrow(1,3)$.

Given $\Lambda \in SL(2, \mathbb{C})$, let $\kappa (\Lambda) \in SO^\uparrow(1,3)$ be the matrix which acts on a four-vector $x \in \mathbb{R}^4$ as
\begin{subequations}\label{eq:2.13}
\begin{eqnarray}
\left(\kappa (\Lambda) x \right)^\sim &=& \Lambda \tilde{x} \Lambda^\dagger  \\
\left( \kappa (\Lambda) x \right)_\sim &=& \Lambda^{\dagger -1} \utilde{x} \Lambda^{-1}
\end{eqnarray}
\end{subequations}
(cf. Eq.~(\ref{eq:2.3})) where the RHS are ordinary products of matrices in $M_2 (\mathbb{C})$ and $(\cdot)^\dagger$ denotes the Hermitian conjugation of a complex matrix. Then, the map
\begin{equation}\label{eq:2.14}
\kappa : SL ( 2, \mathbb{C} ) \rightarrow SO^{\uparrow} (1,3)
\end{equation}
is a double covering homomorphism (cf. \cite{folland2008}) and since $SL(2, \mathbb{C})$ is simply connected, it is a universal covering homomorphism. Via $\kappa$, we obtain an action of $SL(2, \mathbb{C})$ on $\mathbb{R}^4$. We will often suppress $\kappa$ when we denote an action of an element $SL(2, \mathbb{C})$ on an element of $\mathbb{R}^4$. (e.g., $\kappa(\Lambda)x = \Lambda x$ and so on.) We can form a semi-direct product $\mathbb{R}^4 \ltimes SL(2, \mathbb{C})$ using this action. The map $\mathbb{R}^4 \ltimes SL(2, \mathbb{C}) \rightarrow \mathbb{R}^4 \ltimes SO^\uparrow(1, 3)$ defined by $(a, \Lambda) \mapsto (a, \kappa (\Lambda))$ is also a universal covering homomorphism. The following is a consequence of Bargman's theorem.
\begin{theorem}\label{theorem:2.9}
Given a projective unitary representation $\mathbb{R}^4 \ltimes SO^{\uparrow} (1,3)  \rightarrow PU(\mathcal{H})$ which is continuous with respect to the quotient strong operator topology on $PU(\mathcal{H})$, there is a (continuous) unitary representation $\mathbb{R}^4 \ltimes SL(2,\mathbb{C}) \rightarrow U(\mathcal{H})$ such that the following diagram commutes.
\begin{equation*}
    \begin{tikzcd}[baseline=(current  bounding  box.center)]
\mathbb{R}^4 \ltimes SL(2,\mathbb{C}) \arrow[r] \arrow{d}[swap]{1 \times \kappa}
& U(\mathcal{H}) \arrow[d]  \\
\mathbb{R}^4 \ltimes SO^{\uparrow} (1,3)  \arrow[r]
& PU (\mathcal{H}) 
\end{tikzcd}
\end{equation*}
\end{theorem}

Accordingly, we make the following definition.

\begin{definition}\label{definition:2.10}
A pair $(U, \mathcal{H})$ is called a \textit{quantum system with Lorentz symmetry} if $\mathcal{H}$ is a Hilbert space and $U: \mathbb{R}^4 \ltimes SL(2, \mathbb{C}) \rightarrow U(\mathcal{H})$ is a unitary representation. 
\end{definition}

Quantum systems possessing Lorentz symmetry are the right playground for testing relativistic considerations in QM.

\begin{remark}\label{remark:2.11}
Let $(U,\mathcal{H})$ be a quantum system with Lorentz symmetry. Here comes how we should interpret relativistic scenarios using this system. Suppose two inertial observers Alice and Bob are related by a Lorentz transformation $(a, \Lambda) \in \mathbb{R}^4 \ltimes SO^\uparrow (1,3)$ as in Eq.~(\ref{eq:2.11}) and the two observers' perceptions of one quantum state are given by $\psi_A \in \mathcal{H}$ and $\psi_B \in \mathcal{H}$, respectively. Then, by the above discussions, we require
\begin{equation}\label{eq:2.15}
[\psi_B] = [ U(a, \Lambda') \psi_A]
\end{equation}
where $\Lambda' \in SL(2, \mathbb{C})$ is a lift of $\Lambda$ via the covering map $\kappa$.

This transformation formula is the quantum analogue of the classical transformation formula Eq.~(\ref{eq:2.11}). If $U(-I_2) = \lambda I_\mathcal{H}$ for some scalar $\lambda \in \mathbb{C}$ so that it descends to a projective representation as in Theorem~\ref{theorem:2.9}, then this transformation does not depend on the choice of the lift $\Lambda'$. We will see that this is true in all the cases that we will be looking at (cf. Theorems~\ref{theorem:4.10}).

We should always have in mind the two rules Eqs.~(\ref{eq:2.11}) and (\ref{eq:2.15}) when dealing with a relativistic scenario in which more than one observer is involved.
\end{remark}

\subsection{A Standard Choice of Lorentz Boostings}\label{sec:2.4}

Fix an inertial frame of reference (call this frame Alice) and consider a massive particle moving with respect to the frame. If the particle has some internal states (such as spin), one may want to know how it is observed in a "particle-rest frame". But, there is an ambiguity in this concept. Namely, if one fixes a particle-rest frame, then any other frame transformed by a rotation (that is, an element in $SU(2)$) from this frame is also a particle-rest frame. So, to speak of internal states of a moving particle, it would be convenient for Alice to set up a choice of Lorentz transformation for each possible motion state of the particle.

We make one standard choice in this section. This will be important in later discussions and would serve as a good opportunity to get familiar with the notations of Sect.~\ref{sec:2.1} and Eq.~(\ref{eq:2.13}).

Let $m>0$ be the mass of the particle and denote $p_m = (m, 0, 0, 0)$. Then, the set of possible momentums that the particle can assume is given by

\begin{equation}\label{eq:2.16}
X := \{ p \in \mathbb{R}^4 : p_\mu p^\mu = m^2 , \hspace{0.1cm} p^0 >0 \} \quad \text{(cf. \cite{zee})}.
\end{equation}

For each $p \in X$, $\utilde{p}\tilde{p} = m^2 I = \tilde{p} \utilde{p}$ by Eq.~(\ref{eq:2.4b}). Also, these two matrices are positive matrices with the square roots given by
\begin{subequations}\label{eq:2.17}
\begin{eqnarray}
\sqrt{\utilde{p} } = \frac{1}{\sqrt{ 2 (m + p_0 ) }} \left( \utilde{p} + m I \right) \\ 
\sqrt{\tilde{p} } = \frac{1}{\sqrt{ 2 (m + p_0 ) }} \left( \tilde{p} + m I \right).
\end{eqnarray}
\end{subequations}

It is easy to see that $\sqrt{\utilde{p}} \sqrt{\tilde{p}} = mI = \sqrt{\tilde{p}}\sqrt{\utilde{p}}$ holds, which may be expressed as
\begin{equation}\label{eq:2.18}
\left( \sqrt{\frac{\utilde{p}}{m}} \right)^{-1} = \left( \sqrt{\frac{\tilde{p}}{m}} \right).
\end{equation}

If we observe $\tilde{p} = m \sqrt{\frac{\tilde{p}}{m}} \sqrt{\frac{\tilde{p}}{m}} = \sqrt{\frac{\tilde{p}}{m}} (p_m)^\sim \sqrt{\frac{\tilde{p}}{m}} $, we see that for the matrix
\begin{equation}\label{eq:2.19}
L (p) := \sqrt{\frac{\tilde{p}}{m}} \in SL(2, \mathbb{C}),
\end{equation}
we have $\kappa \big(L(p) \big) p_m = p$ by Eq.~(\ref{eq:2.13}). $L(p)$ is called the \textit{standard boosting sending $p_m$ to $p$}.

\begin{remark}\label{remark:2.12}
If $\hat{\mathbf{p}} = (\sin \theta \cos \phi, \sin \theta  \sin \phi , \cos \theta)$ with $0 \leq \phi \leq 2\pi$ and $0 \leq \theta \leq \pi$, then

\begin{equation*}
R(\hat{\mathbf{p}}) := \begin{pmatrix} e^{- i \frac{\phi}{2}} & 0 \\ 0 & e^{i \frac{\phi}{2}} \end{pmatrix} \begin{pmatrix} \cos \frac{\theta}{2} & - \sin \frac{\theta}{2} \\ \sin \frac{\theta}{2} & \cos \frac{\theta}{2} \end{pmatrix} \in SU(2)
\end{equation*}
is a rotation which takes $\hat{z}$ to $\hat{\mathbf{p}}$ and
\begin{equation*}
B_m ( |\mathbf{p}|) := \begin{pmatrix} \sqrt{  \frac{p^0 + | \mathbf{p}|}{m} }    & 0 \\ 0 & \sqrt{  \frac{p^0 - | \mathbf{p}|}{m} } \end{pmatrix} \in SL(2, \mathbb{C})
\end{equation*}
is the boosting along the $z$-axis which takes $p_m$ to $(p^0 , 0, 0 , | \mathbf{p}|)$ (cf. Eq.~(\ref{eq:2.13})).

One can easily see, using Eq.~(\ref{eq:2.13}), that
\begin{equation*}
\frac{\tilde{p}}{m}  = R( \hat{\mathbf{p}} ) B_m ( | \mathbf{p}|)^2 R(\hat{\mathbf{p}} )^{-1}
\end{equation*}
holds. Therefore,
\begin{equation}\label{eq:2.20}
L (p) = \sqrt{\frac{\tilde{p}}{m}} = R( \hat{\mathbf{p}} ) B_m ( | \mathbf{p}|) R(\hat{\mathbf{p}} )^{-1},
\end{equation}
which implies that the matrix Eq.~(\ref{eq:2.19}) is equal to the standard boosting that maps $p_m$ to $p$ used in the physics literature (cf. \cite{weinberg}, p.68).
\end{remark}

\subsection{Relativistic Perception}\label{sec:2.5}

In this section, we introduce the concept of "relativistic perception", which is the central topic of this paper. Let an inertial frame of reference be given (cf. Definition~\ref{definition:2.7}). Then, \textit{any tensorial quantity represented in the coordinate system of the frame that transforms covariantly under Lorentz transformations} is called \textit{"relativistic perception" of the frame}. Perhaps the best way to illustrate this concept is by giving examples and nonexamples.

Fix an inertial frame of reference (call this frame Alice) and suppose we are given a point particle with mass $m > 0$ whose relativistic momentum and angular momentum are represented as $p^\mu$ and $j_{\alpha \beta}$\footnote{There is a systematic way to promote non-relativistic, frame-dependent dynamical quantities (e.g. angular momentum) to corresponding relativistic concepts that have meaning in every inertial frame (see the discussions right below Eq.~(\ref{eq:2.11})). For example, the non-relativistic momentum $\mathbf{p} := m \mathbf{v}$ is an observer-dependent quantity, which is promoted to the four-momentum $p = (m \gamma , m \gamma \mathbf{v} )$ where $ \gamma = \frac{1}{ \sqrt{1 - | \mathbf{v}|^2}} $. Likewise, there is an antisymmetric 2-tensor $j_{\alpha \beta}$ called the \textit{relativistic angular momentum} which is promoted from the ordinary angular momentum 2-tensor $j_{kl}$. Usually, one uses a three-vector $\boldsymbol{j}$ defined by $j^1 = j_{23} $, $j^2 = j_{31}$, and $j^3 = j_{12}$ as the angular momentum three-vector. For more details, see \cite{zee} and \cite{anderson}.} in the frame, respectively. The two quantities are relativistic perceptions of Alice.

Now, suppose that the particle has non-zero spin (for the concept of spin in classical SR, see \cite{anderson}). Since spin is an internal angular momentum of the particle, we come to consider a particle-rest frame (Bob) in which the spin is represented as a three-vector $(0, \mathbf{s})$. Is the spin of the particle a relativistic perception of Alice? No, apart from the case when $p = (m,0,0,0)$. Rather, it is a relativistic perception of Bob.

\begin{example}[Pauli-Lubansky four-vector and Newton-Winger Spin]\label{example:2.13}
Naturally, we come to wonder how the spin of a particle is perceived by arbitrary inertial frames of reference with respect to which the particle might be moving (such as Alice). The quantity should be a Lorentz covariant vector (cf. Eq.~(\ref{eq:2.12})) and become a three-vector in any particle-rest frame. The \textit{Pauli-Lubansky four-vector} turns out to be the right object (see \cite{anderson} for an extended discussion of this vector). In Alice's frame, it is defined as
\begin{equation}\label{eq:2.21}
w^\mu = \frac{1}{2} \varepsilon^{\nu \alpha \beta \mu} p_{\nu} j_{\alpha \beta}
\end{equation}
where $\varepsilon^{\nu \alpha \beta \mu} $ is an alternating $4$-tensor which is $1$ when $(\nu, \alpha, \beta, \mu)$ is a cyclic permutation of $(0, 1, 2, 3)$. One can show that this is a Lorentz covariant vector (i.e., relativistic perception of Alice),
\begin{equation}\label{eq:2.22}
p_\mu w^\mu = 0,
\end{equation}
and when $p= (m,0,0,0)$ (that is, in a particle-rest frame),
\begin{equation}\label{eq:2.23}
w = (0, m \boldsymbol{j}),
\end{equation}
as it should be. So, using the concept of "relativistic perception" introduced in this section, one can say that \textit{the Pauli-Lubansky four-vector of a particle is the internal angular momentum (i.e., the spin) of the particle perceived by an observer who perceives the spin-carrying particle as moving with momentum $p$.}

Using the choice of boostings obtained in Sect.~\ref{sec:2.4}, we can obtain another important object, which will be relevant in later discussions.

Observe that for all $p \in X$,
\begin{equation}\label{eq:2.24}
L (p)^{-1} w = (0, m \mathbf{s})
\end{equation}
holds (see Eqs.~(37)--(38) of \cite{lee2022}), where
\begin{equation}\label{eq:2.25}
\boldsymbol{s} = \frac{1}{m} \left( \mathbf{w} - \frac{ w^0 \mathbf{p}}{m +p^0} \right)
\end{equation}
is called the \textit{Newton-Wigner spin three-vector}. The Newton-Wigner spin is just the Pauli-Lubansky four-vector perceived by an $L (p)^{-1}$-transformed inertial observer, with respect to whom the particle is at rest consequently.

Note that while the Pauli-Lubansky vector transforms covariantly under Lorentz transformations, the Newton-Wigner spin does not. Therefore, the Pauli-Lubansky vector is relativistic perception whereas the Newton-Wigner spin is not. The relation between these two vector descriptions of the internal angular momentum of a particle (cf. Eq.~(\ref{eq:2.23})) will be a recurrent theme throughout the paper (cf. Sect.~\ref{sec:3.3.2}).
\end{example}

\subsubsection{A scheme by which inertial observers can obtain their relativistic perception of the spin of a moving particle}\label{sec:2.5.1}

\hfill

In a relativistic scenario where several inertial observers are interested in the spin of a particle, it is desirable for each observer to record the spin information in the form of relativistic perception, i.e., as the Pauli-Lubansky four-vector in each frame since it is the information that has meaning in every frame (see the discussion surrounding Eq.~(\ref{eq:2.12})).

So, let's consider an inertial observer Bob who is trying to calculate the Pauli-Lubansky four-vector $w$ of a moving particle. Classically, Bob could, in principle, measure the momentum $p$ of the particle, conceive of a frame change to a particle-rest frame using the transformation $L(p)$, measure the spin three-vector $\mathbf{s}$ in that frame using spin-magnetic field interaction (cf. Ch.~7, pp.248--253 of \cite{anderson}), which is precisely the Newton-Wigenr spin of the particle, and recover the Pauli-Lubansky four-vector $w$ by using Eq.~(\ref{eq:2.24}).

\begin{remark}[The quantum particle case]\label{remark:2.14}

However, if the particle under investigation is a quantum particle, the quest of determining the Pauli-Lubansky four-vector $w$ of the particle becomes subtle. Since the motion state of a quantum particle (see Sect.~\ref{sec:3} for the definition) is given by a superposition of possible motion states, there is no such thing as a "particle-rest frame" in which the value of $\mathbf{s}$ gets meaningful, from which one can calculate $w$. In fact, there is no consensus among researchers about the precise definition of the relativistic spin operator in Relativistic Quantum Mechanics and consequently on how to measure the Pauli-Lubansky four-vector of a moving quantum particle (see \cite{bauke2014, bauke2014b, terno2016, DERIGLAZOV2016} on this issue).

One solution to this subtlety is to consider the wave functions representing the states of the quantum particle as fields of spin states corresponding to all possible motion states\footnote{This is where the language of bundle theory naturally comes in.} apply the above classical scheme to each spin-motion state to make it contain information of the Pauli-Lubansky four-vector, i.e., relativistic perception. This will expose a critical flaw of the standard Hilbert space description of single-particle state spaces and suggest a way to fix it. These statements will be illustrated in Sects.~\ref{sec:3.2}--\ref{sec:3.3} for massive particles with spin-1/2.

Arranging the internal quantum states in this way not only has the conceptual advantage as explained in this subsection (i.e., it is a faithful reflection of the reality demanded by the principle of SR), but also has observable consequences as we will see in Sect.~\ref{sec:7}.
\end{remark}

\section{The RQI of massive particles with spin-1/2}\label{sec:3}

In this section, we define the single-particle state space and use one particular example of them (namely, the particle with mass $m>0$ and spin-1/2) to briefly survey the fundamental perplexities of RQI first observed by the two pioneering papers \cite{peres2002, gingrich2002} and how these perplexities have finally reached a definitive clarification in the recent work \cite{lee2022}. Those who are interested in other aspects of RQI as well are referred to \cite{peres2004, mann2012} and references therein.

\begin{definition}\label{definition:3.1}
The irreducible unitary representation spaces of the group $G:=\mathbb{R}^4 \ltimes SL(2,\mathbb{C})$ are called \textit{single-particle state spaces}.
\end{definition}

That is, single-particle state spaces are \textit{the smallest possible quantum systems which possess Lorentz symmetry}. This definition is due to Wigner \cite{wigner}. We will see in Sect.~\ref{sec:6} that the Hilbert space
\begin{subequations}\label{eq:3.1}
\begin{equation}\label{eq:3.1a}
\mathcal{H}_{L,1/2} \cong L^2 \left(\mathbb{R}^3 , \frac{d^3 \mathbf{p}}{\sqrt{m^2 + |\mathbf{p}|^2}} \right) \otimes \mathbb{C}^{2},
\end{equation}
on which the representation $U_{L, 1/2}$ acts as
\begin{equation}\label{eq:3.1b}
[U_{L,1/2} (a, \Lambda) \psi] (p)= e^{-i \langle p , a \rangle}  W_L(\Lambda, \Lambda^{-1} p ) \psi (\Lambda^{-1} p)
\end{equation}
\end{subequations}
upon identifying $\mathbf{p} \cong ( \sqrt{m^2 + | \mathbf{p}|^2 } , \mathbf{p}) = p \in \mathbb{R}^4$, is a single-particle state space, which can be called the single-particle state space for particles of mass $m>0$ and spin-1/2. (Here $W_L ( \Lambda , p ) := L(\Lambda p)^{-1} \Lambda L(p) \in SU(2)$ for $\Lambda \in SL(2, \mathbb{C})$ and $p \in X$ is the \textit{Wigner rotation matrix}.) Many elementary particles including electron and quarks, and also very important non-elementary particles such as proton and neutron can be described by this representation. This case has been the most intensely studied class of particles in the context of RQI and Eq.~(\ref{eq:3.1}) has been the standard representation that has been used to describe the particles of this type.\footnote{Note that this representation is different from the one used in the textbook \cite{weinberg} by the normalization factor $\sqrt{ \big( \Lambda^{-1} p \big)^0 / p^0 }$. This is because this factor has been subsumed into the measure $\frac{d^3 \mathbf{p}}{\sqrt{m^2 + |\mathbf{p}|^2}}$ in the definition of the $L^2$-space in Eq.~(\ref{eq:3.1a}).} Throughout this paper, except in Sect.~\ref{sec:5}, $G$ will always denote the group $\mathbb{R}^4 \ltimes SL(2, \mathbb{C})$.

\subsection{The pioneering works}\label{sec:3.1}

Here comes a brief mathematical analysis of the two pioneering works \cite{peres2002} and \cite{gingrich2002}. Throughout, the identification $\mathbb{R}^3 \cong X$ (cf. Eq.~(\ref{eq:2.16})) given by $\mathbf{p} \mapsto (\sqrt{m^2 + |\mathbf{p}|^2 |}, \mathbf{p})$ will be assumed and we will freely identify $\mathbf{p} \in \mathbb{R}^3$ with $p = (\sqrt{m^2 + |\mathbf{p}|^2 |}, \mathbf{p}) \in X$ (i.e., $p^0 = \sqrt{m^2 + |\mathbf{p}|^2 }$). In this identification, the measure $\frac{d^3 \mathbf{p}}{\sqrt{m^2 + |\mathbf{p}|^2}}$ will be denoted as $d \mu (p)$.

\paragraph{The work of Peres, Scudo, and Terno}

\hfill

In the seminal paper \cite{peres2002}, the authors considered a massive spin-1/2 single-particle state space Eq.~(\ref{eq:3.1}). So, an element $\psi \in \mathcal{H}_{L,1/2}$ can be written as
\begin{equation*}
\psi = \begin{pmatrix} a_1  \\ a_2  \end{pmatrix}, \hspace{0.2cm} \Big(a_1, a_2 \in L^2 (\mathbb{R}^3, \mu) \Big).
\end{equation*}

On this space, they formed the \textit{density matrix corresponding to a unit vector $\psi$} (i.e., the projection onto the one dimensional space spanned by $\psi$)
\begin{equation*}
\rho = |\psi \rangle \langle \psi | \in L^2( \mathbb{R}^3 \times \mathbb{R}^3, \mu \times \mu ) \otimes M_2 (\mathbb{C}),
\end{equation*}
and defined the \textit{reduced density matrix for spin of $\psi$} by taking the partial trace (cf. Ch.~19 of \cite{hall}) with respect to the $L^2 (\mathbb{R}^3 , \mu)$-component of the Hilbert space, i.e.,
\begin{equation}\label{eq:3.2}
\tau := \text{Tr}_{L^2} ( \rho) \in M_2 (\mathbb{C}).
\end{equation}

Since $\mathcal{H}_{L, 1/2}\cong L^2 (\mathbb{R}^3 , \mu) \otimes \mathbb{C}^2$ is a tensor product system, the $\mathbb{C}^2$-factor of which contains the spin information of the single-particle states, the reduced $2 \times 2$-matrix $\tau$ is supposed to give the "spin information stored in $\mathbb{C}^2$ of the single-particle state $\psi$ independent of the momentum variable $L^2 (\mathbb{R}^3 , \mu)$" according to the usual treatment of composite systems in QIT (cf. \cite{holevo}). Naturally, the authors of the paper defined the \textit{spin entropy of the state $\psi$} as
\begin{equation}\label{eq:3.3}
S(\tau) = - \text{Tr} (\tau \log \tau ),
\end{equation}
which is (supposedly) a quantitative measure of the spin information contained in the state $\psi$.

Consider a scenario where two parties communicate with each other by using massive particles with spin-1/2 as qubit carriers. One party encodes one bit of information in the spin of a massive spin-1/2 particle and transmits the particle to another party. The receiving party is only interested in the spin information of the particle independent of its momentum. So, the reduced matrix Eq.~(\ref{eq:3.2}) is expected to function as an information resource that can be manipulated as in the usual treatment of QIT.

However, the authors of the paper found a certain perplexity which was against this innocuous expectation. They examined a situation in which one inertial observer (Alice) prepares a (supposed) spin-up state
\begin{equation*}
\psi_A  = \begin{pmatrix} a_1  \\ 0 \end{pmatrix}
\end{equation*}
where $a_1 \in L^2 (\mathbb{R}^3, \mu)$ is a normalized Gaussian distribution function, while the other observer (Bob), moving along the $x$-axis of Alice's coordinate system with constant velocity (denote this Lorentz transformation as $(0,\Lambda) \in G$), is trying to measure the spin $z$-component of the prerpared state. Let
\begin{equation*}
\psi_B = \begin{pmatrix} b_1 \\ b_2 \end{pmatrix}
\end{equation*}
be the above state in Bob's reference frame. According to Remark~\ref{remark:2.11}, the two states are related by
\begin{equation*}
\psi_B = U_{L, 1/2} (0, \Lambda) \psi_A.
\end{equation*}

Thus, if Bob is able to carry out a momentum-independent spin measurement, what he would get is the quantum informational property of the reduced density matrix given by
\begin{align}\label{eq:3.4}
\tau_B = \text{Tr}_{L^2} ( | \psi_B \rangle \langle \psi_B | ) = \text{Tr}_{L^2} \left[ U_{L,1/2}(0, \Lambda) |\psi_A \rangle \langle \psi_A | U_{L,1/2} (0, \Lambda) ^{-1} \right] \nonumber \\
= \int_{X_m} W_L(\Lambda, \Lambda^{-1} p ) \begin{pmatrix} |a_1 (\Lambda^{-1} p) |^2 & 0 \\ 0 & 0 \end{pmatrix} W_L(\Lambda, \Lambda^{-1} p ) ^{-1} d \mu (p)\footnotemark
\end{align}
\footnotetext{The partial trace with respect to an $L^2$-Hilbert space can be calculated by means of an integration over its underlying measure space in certain circumstances including the case at hand. See \cite{duflo}.}
while the spin information that Alice has prepared is contained in the matrix
\begin{equation}\label{eq:3.5}
\tau_A = \text{Tr}_{L^2} ( |\psi_A \rangle \langle \psi_A| ) = \text{Tr}_{L^2} \left[ |a_1  |^2 \otimes \begin{pmatrix} 1 & 0 \\ 0 & 0 \end{pmatrix} \right] = \begin{pmatrix} 1 & 0 \\ 0 & 0 \end{pmatrix}.
\end{equation}

The authors calculated the spin entropies of Eqs.~(\ref{eq:3.4})--(\ref{eq:3.5}) and showed that, while $S(\tau_A) =0$ always, $S(\tau_B)$ is in general non-zero depending on $\Lambda$, showing that the spin entropy of the particle has no relativistically invariant meaning. From this, they concluded that there is no definite transformation law between $\tau_A$ and $\tau_B$ which depends only on $\Lambda$ and thus, the notion "spin state of a particle" is meaningless unless one does not specify its complete state, including the momentum variables.

\paragraph{The work of Gingrich and Adami}

\hfill

Shortly, the paper \cite{gingrich2002} considered a similar scenario, but with two massive spin-1/2 particles. This time the quantum Hilbert space is given by
\begin{equation}\label{eq:3.6}
\mathcal{H}_{L,1/2} \otimes \mathcal{H}_{L,1/2}  \cong L^2( \mathbb{R}^3 \times \mathbb{R}^3 , \mu \times \mu ) \otimes \big( \mathbb{C}^2 \otimes \mathbb{C}^2 \big).
\end{equation}

Following \cite{peres2002}, the authors considered a two-particle state $\psi \in \mathcal{H}_{L, 1/2} \otimes \mathcal{H}_{L, 1/2}$, formed the density matrix corresponding to $\psi$ as
\begin{equation*}
\rho := | \psi \rangle \langle \psi | \in L^2 (\mathbb{R}^3 \times \mathbb{R}^3 , \mu \times \mu) ^{\otimes 2} \otimes \big( M_2 (\mathbb{C} ) \otimes M_2 (\mathbb{C}) \big),
\end{equation*}
and took the partial trace with respect to the momentum variable to obtain a two-qubit bipartite state
\begin{equation}\label{eq:3.7}
\tau := \text{Tr}_{L^2} (\rho) \in M_2 (\mathbb{C} ) \otimes M_2 (\mathbb{C}),
\end{equation}
where each component $M_2 (\mathbb{C})$ represents the spin quantum system of each particle, respectively. The entanglement of this bipartite state
\begin{equation}\label{eq:3.8}
E(\tau) := S( \text{Tr}_{\mathbb{C}^2} (\tau) )
\end{equation}
is called the \textit{spin entanglement of the two-particle state $\psi$}. Here, the trace in the RHS is done with respect to any one $\mathbb{C}^2$-component of the tensor product space $\mathbb{C}^2 \otimes \mathbb{C}^2$. The result does not depend on the choice (cf. \cite{holevo}).

With these notions at hand, they considered a situation where Alice has prepared a maximal Bell state with Gaussian momentum distribution
\begin{equation*}
\psi_A (p,q) = \frac{1}{\sqrt{2}} f(p) f(q) \phi^+ \in \mathcal{H}_{L,1/2} \otimes \mathcal{H}_{L,1/2}
\end{equation*}
where
\begin{equation*}
\phi^+ := \left( \begin{pmatrix}1 \\ 0 \end{pmatrix} \otimes \begin{pmatrix}1 \\ 0 \end{pmatrix} + \begin{pmatrix} 0 \\ 1\end{pmatrix} \otimes \begin{pmatrix}0 \\ 1 \end{pmatrix} \right) \in \mathbb{C}^2 \otimes \mathbb{C}^2
\end{equation*}
and $f(p)$ is a normalized Gaussian distribution function.

Bob is now moving with constant velocity along the $z$-axis with respect to Alice. Denote this Lorentz transformation as $(0, \Lambda) \in G$. Then, using Eq.~(\ref{eq:3.1b}) on the tensor product system, the state
\begin{equation*}
\psi_B  =  [U_{L,1/2} (0, \Lambda) \otimes U_{L, 1/2} (0, \Lambda)] \psi_A
\end{equation*}
is what the inertial observer Bob perceives as the Bell state that Alice has prepared (cf. Remark~\ref{remark:2.11}).

Now, following the above procedure, they form qubit bipartite states $\tau_A$ and $\tau_B$, measure the spin entanglements $E(\tau_A)$ and $E(\tau_B)$ in their respective frames, and compare. Since
\begin{equation*}
\tau_A := \text{Tr}_{L^2} ( | \psi_A \rangle \langle \psi_A |) = | \phi^+ \rangle \langle \phi^+ |,
\end{equation*}
we see that $\tau_A$ is a maximally entangled state and hence $E(\tau_A) = 1$ always. However, the authors found that $E(\tau_B)$ might even vanish depending on $\Lambda$\footnote{In the paper, however, the authors used \textit{Wootter's concurrence} $C(\tau)$ instead of our entanglement $E(\tau)$. Nevertheless, for two-qubit systems, the relations $E(\tau) = 1 \Leftrightarrow C(\tau)=1 $ and $E(\tau)=0 \Leftrightarrow C(\tau)=0$ hold. Therefore, the two notions can be equivalently used for the same purpose of showing whether a pure state is separable or not. See \cite{wootter1998} for details.}, which happens precisely when the bipartite state $\tau_B$ is separable (cf. \cite{holevo}). So, in particular, a two-particle state that is maximally entangled in spin with respect to one inertial frame may be perceived as a state that is completely unrelated in spin with respect to another inertial observer, a striking perplexity.

Therefore, just like the spin entropy in \cite{peres2002}, the spin entanglement of a two-particle state is also an observer-dependent quantity, showing its inadequacy as an informational resource in the context of RQI.

As written in the introduction, an implication that these two works entailed was that when one wants to use the spin of massive particles with spin-1/2 (such as electron) as information carriers (i.e., qubit carriers), the concepts of entropy, entanglement, and correlation of the spins, which are important informational resources in QIT, may require a reassessment \cite{peres2004}.

\subsection{A problem with the standard description Eq.~(\ref{eq:3.1})}\label{sec:3.2}

Some have questioned the meaning of the reduced density matrix Eq.~(\ref{eq:3.2}). For example, \cite{saldanha2012a} claimed, on the basis of a physical consideration, that a momentum-independent measurement of spin is impossible, and hence Eq.~(\ref{eq:3.2}) is in fact meaningless. A related issue is the search for a right definition of \textit{relativistic spin operator}, which still has no universally agreed upon definition (cf. \cite{alsing2012, caban2013, bauke2014, bauke2014b, terno2016, DERIGLAZOV2016, avelar2021}) whereas the paper \cite{peres2002} assumed it as the operator $1 \otimes \frac{1}{2} \boldsymbol{\tau}$ defined on the space Eq.~(\ref{eq:3.1a}).

Therefore, in effect, these have questioned the validity of the interpretation that the $\mathbb{C}^2$-component in Eq.~(\ref{eq:3.1a}) should give the momentum-independent spin states of the particle. As we shall see, the $\mathbb{C}^2$-component in Eq.~(\ref{eq:3.1a}) is indeed inherently momentum dependent. This is most clearly seen if we look at the bundle underlying the Hilbert space $\mathcal{H}_{L,1/2}$ instead of the space itself.

In \cite{lee2022}, it was pointed out that $\mathcal{H}_{L,1/2}$ can be viewed as the $L^2$-section space of the trivial bundle $X_m ^+ \times \mathbb{C}^2$ with inner product
\begin{equation*}
(f,g) \mapsto \int_{X} f(p)^\dagger g(p) d \mu (p).
\end{equation*}

Denote this bundle as $E_{L, 1/2}$. This bundle is an assembly of the two-level quantum systems $(E_{L,1/2})_p \cong \mathbb{C}^2$ corresponding to each motion state (momentum) $p \in X$, and each wave function $\psi \in \mathcal{H}_{L,1/2}$ becomes a field of qubits. The so-called momentum-spin eigenstate $|p, \chi \rangle, (p \in X, \chi \in \mathbb{C}^2)$ used in the physics literature can be identified with the point $ (p, \chi ) \in (E_{L,1/2})_p$ in this formalism.

Since the full information of each quantum state of a massive particle with spin-1/2 can be recorded in the bundle $E_{L,1/2}$, being an $L^2$-section on the bundle\footnote{This mathematical fact has nothing to do with physical measurement.\label{footnote:15}}, each inertial observer can use the bundle instead of the space $\mathcal{H}_{L,1/2}$ for the description of a massive particle with spin-$1/2$. How is this bundle description related among different inertial observers? Suppose two inertial observers Alice and Bob are related by a Lorentz transformation $(a, \Lambda) \in G$ as in Eq.~(\ref{eq:2.11}). If Alice has prepared a particle in the state $ \psi \in \mathcal{H}_{L,1/2}$ in her frame, then Bob would perceive this particle as in the state $U_{L,1/2 } (a, \Lambda) \psi \in \mathcal{H}_{L,s}$ according to Sect.~\ref{sec:2.3} (cf. Remark~\ref{remark:2.11}).

For these transformation laws for wave functions to be true, Alice's bundle $E_{L,1/2} ^A$ and Bob's bundle $E_{L,1/2} ^B$ should be related by the bundle isomorphism

\begin{gather}
\lambda_{L,1/2} (a, \Lambda ) : E_{L,1/2} ^A \rightarrow E_{L,1/2} ^B \nonumber \\
(p,v)^A \mapsto \left(  \Lambda p, e^{- i (\Lambda p)_\mu a^\mu} \sigma_{1/2} \left(W_L ( \Lambda, p) \right) v \right)^B \label{eq:3.9}
\end{gather}
so that the transformation law for the sections
\begin{equation}\label{eq:3.10}
\psi^A \mapsto \psi^B = \lambda_{L,1/2} (a, \Lambda) \circ \psi^A \circ \Lambda^{-1}
\end{equation}
becomes $U_{L,1/2} (a, \Lambda)$. The following commutative diagram is useful in visualizing the transformation law.
\begin{equation}\label{eq:3.11}
\begin{tikzcd}[baseline=(current  bounding  box.center)]
E_{L,1/2} ^A \arrow[r, "{\lambda_{L,1/2} ( \Lambda, a)}"] \arrow[d ]
& E_{L,1/2} ^B \arrow[d]  \\
X ^A  \arrow[r, "\Lambda"]
& X ^B
\end{tikzcd}.
\end{equation}

Now, suppose two inertial observers Alice and Bob are related by a Lorentz transformation $(0, L(p)) \in G$ as in Eq.~(\ref{eq:2.11}). Since the $E_{L,1/2}$-bundle description of the two observers are related by Eq.~(\ref{eq:3.9}), we have
\begin{equation}\label{eq:3.12}
\lambda_{L,1/2} (0, L (p) ) ( p_m , \chi )^A = ( p, \chi )^B
\end{equation}
for ${}^\forall (p_m,\chi) \in (E_{L,1/2}  ^A)_{p_m}$.

To see what physical implications that this equation entails, we need a brief digression into the quantum mechanics of the two-level system.

If $\chi \in \mathbb{C}^2$ is a qubit, that is, $\|\chi \|^2 = 1$, then, it is the definite spin-up state along the direction $\mathbf{n}:= \chi^\dagger \boldsymbol{\tau} \chi = (\chi^\dagger \tau^1 \chi, \chi^\dagger \tau^2 \chi, \chi^\dagger \tau^3 \chi) \in \mathbb{R}^3$, which means $(\frac{1}{2} \boldsymbol{\tau}\cdot \mathbf{n}) \chi = \frac{1}{2} \chi$. In fact,
\begin{equation}\label{eq:3.13}
\boldsymbol{\tau}\cdot \mathbf{n} = \chi \chi^\dagger - (I_2 - \chi \chi^\dagger) = 2 \chi \chi^\dagger - I_2
\end{equation}
since a state orthogonal to $\chi$ is the spin-down state along the direction $\mathbf{n}$. So, we see that the state $\chi$ is completely characterized by the three-vector $\mathbf{n}$ and we may call it the \textit{spin direction} of $\chi$.

Now, suppose an inertial observer Bob is interested in the relativistic perception of the spin of a quantum particle whose state is given by $\psi \in \mathcal{H}_{L,1/2}$. Since the particle does not have definite momentum, Bob cannot naively apply the classical scheme given in Sect.~\ref{sec:2.5.1} to obtain the relativistic perception. So, he resorts to the strategy outlined in Remark~\ref{remark:2.14}.

As argued above, the state $\psi$ of the particle can be represented as a section of the $L^2$-section space of the bundle $E_{L, 1/2} ^B$. So, the full information of the state can be expressed as $\{ (p, \psi(p)) \in E_{L, 1/2} ^B : p \in X \}$. Now, treating each motion-spin state $(p, \psi(p)) \in ( E_{ L, 1/2 }^ B )_p$ as a moving qubit with momentum $p \in X$ and spin state $\psi(p) \in \mathbb{C}^2$, we try to apply the scheme of Sect.~\ref{sec:2.5.1}.

First, Bob transforms his frame by $L(p)^{-1}$, getting an inertial frame Alice in whose frame the qubit is at rest, calculate the spin direction of the qubit in that frame, and use Eq.~(\ref{eq:3.12}) to obtain his relativistic perception of the spin. If, by calculating the spin direction $\mathbf{n}$, Alice finds that the qubit is $(p_m, \chi) \in (E_{L, 1/2} ^A)_{p_m}$, then Bob would conclude that his relativistic perception of the qubit is $(p, \chi) \in (E_{L,1/2} ^B)_p$ according to Eq.~(\ref{eq:3.12}), whose spin direction is again $\mathbf{n}$ by Eq.~(\ref{eq:3.13}).

However, the relativistic perception of $\mathbf{n}$ in Bob's frame should be $L (p) (0, \mathbf{n})$ according to Eq.~(\ref{eq:2.12}), which is not equal to $\mathbf{n}$ in general (cf. \cite{derbarba2012}). Therefore, we conclude that, without recourse to the frame change $L (p)$, the three-vector $\mathbf{n}$, and hence the qubit $\chi$ itself, does not reflect Bob's perception of the spin state.

So, Eq.~(\ref{eq:3.12}) tells us that the qubits in $(E_{L,1/2} ^B)_p$ are not Bob's perception (in the sense of Sect.~\ref{sec:2.5}) of the spin state if the qubits in $(E_{L,1/2} ^A)_{p_m}$ are the perception of the $L (p)^{-1}$-transformed observer Alice\footnote{Throughout the paper, we assume (for obvious reason) that the elements in the fiber over the stationary state $p_m$ correctly reflect the relativistic perception.}. In other words, the vectors contained in $(E_{L,1/2} ^B)_p$ themselves don't have meaning in Bob's reference frame. They become useful only if Bob is also provided with the knowledge of $L (p)$. So, in particular, a state of the form $\begin{pmatrix} a_1  \\ 0 \end{pmatrix} \in \mathcal{H}_{L, 1/2}$ in Bob's frame cannot be called "the spin-up state along the $\hat{z}$-axis of Bob". It only tells us that if the particle happens to have momentum $p$, then the $L (p)^{-1}$-transformed observer Alice would see that the particle is in the spin-up state along the $\hat{z}$-axis in her frame.

Because of this fact, when we have no access to the momentum variable, the mere information of the $\mathbb{C}^2 $-component in Eq.~(\ref{eq:3.1a}) does not give us any information about the spin at all, let alone the reduced density matrix Eq.~(\ref{eq:3.2}) which is formed by summing over these pieces of information.

To see this last assertion explicitly, let's look more closely at the anatomy of the reduced density matrix for spin (Eq.~(\ref{eq:3.2})). Let $\psi = f \chi \in \mathcal{H}_{L,1/2}$ be a state where $f \in L^2(X, \mu)$ is a continuous $L^2$-normalized function and $\chi :X \rightarrow E_{L,1/2}$ is a continuous section such that $h_{L,1/2} (\chi(p), \chi(p)) = \| \chi(p) \|^2 = 1$ for all $p \in X$, i.e. $\chi(p)$ is a qubit in $(E_{L,1/2})_p$ for each $p \in X$. By denoting the spin direction of each qubit $\chi(p)$ as $\mathbf{n}(p)$ (i.e., $\mathbf{n}(p) = \chi(p)^\dagger \boldsymbol{\tau} \chi(p)$), we have
\begin{equation}\label{eq:3.14}
\psi(p) \psi (p)^\dagger = \frac{|f(p)|^2}{2} ( \boldsymbol{\tau} \cdot \mathbf{n} (p) + I_2 ) = \frac{|f(p)|^2}{2} \Big( (0, \mathbf{n}(p))^\sim + I_2 \Big)
\end{equation}
by Eq.~(\ref{eq:3.13}) (for the last expression, see Eq.~(\ref{eq:2.3})). So, the spin reduced density matrix Eq.~(\ref{eq:3.2}) becomes
\begin{equation}\label{eq:3.15}
\tau:= \int_{X_m}  \psi(p) \psi(p)^\dagger d \mu (p) = \frac{1}{2} + \frac{1}{2} \int_{X_m} |f(p)|^2 (0,\mathbf{n}(p))^\sim d \mu(p)
\end{equation}
which is just a weighted average of the spin direction $\mathbf{n}(p)$ of the qubits $\chi(p)$. However, since each three-vector $\mathbf{n}(p)$ gets its meaning only with respect to the $L(p)^{-1}$-transformed frame (as shown above), Eq.~(\ref{eq:3.15}) is a summation of vectors living in a whole lot of different coordinate systems. So, we see that this average value really has no meaning at all.\footnote{This mathematical proof is taken from \cite{lee2022}. For a physical reasoning for the meaninglessness of Eq.~(\ref{eq:3.2}), see \cite{saldanha2012a}.}

Although we will not give as detailed analysis as for Eq.~(\ref{eq:3.2}), we remark that the matrix Eq.~(\ref{eq:3.7}) is meaningless also for the same reason. That is, since the fibers $(E_{L, 1/2})_p \otimes (E_{L,1/2})_q$ don't reflect the perception of a fixed inertial observer who is using the bundle $E_{L, 1/2} \boxtimes E_{L,1/2} \rightarrow \mathbb{R}^3 \times \mathbb{R}^3$ for the description of two particles, the mere information of the $\mathbb{C}^2 \otimes \mathbb{C}^2$-component of Eq.~(\ref{eq:3.6}) does not give the observer any information of the spin unless the observer has access to the momentum variable. Therefore, the matrix Eq.~(\ref{eq:3.7}), which is just the sum of these pieces of information, has no real meaning.

This definitive clarification of the perplexities that we explored in Sect.~\ref{sec:3.1} is due to the recent work \cite{lee2022}. After giving this proof, the paper went further to remark that "any anticipation that this (the matrix Eq.~(\ref{eq:3.2})) would give the spin information independent of the momentum variable is an illusion caused by the form of the standard representation space Eq.~(\ref{eq:3.1a}) as a tensor product system."

Having seen a problem with the representation space Eq.~(\ref{eq:3.1}) regarding the perception of a fixed inertial observer, which was the root of the fundamental perplexities posed by the pioneering works, we now proceed to resolve this difficulty.

\subsection{The perception and boosting bundle descriptions for massive spin-1/2 particles}\label{sec:3.3}

The central idea of the paper \cite{lee2022} was that by introducing an Hermitian metric on the bundle $E_{L,1/2}$, we can construct another bundle $E_{1/2}$, called \textit{perception bundle}, whose fibers correctly reflect each inertial observer's "relativistic perception" introduced in Sect.~\ref{sec:2.5}. So, the problem of the $E_{L,1/2}$-bundle description as noted in the preceding subsection is resolved in the perception bundle description. These statements will be made clear in this subsection.

Before we begin, we note that $L (p)$ in Eq.~(\ref{eq:2.20}) is not the only boosting which maps $p_m$ to $p$. In fact, $L' (p) = L (p) B(p)$ for any $B(p) \in SU(2)$ does the same job and all the preceding arguments hold just as well. Of course, if one uses a different definition of $L(p)$, then the representation Eq.~(\ref{eq:3.1b}) is changed along with the physical meaning of the $\mathbb{C}^2$-component in Eq.~(\ref{eq:3.1a}).

So, there is certain arbitrariness in the $E_{L,1/2}$-bundle description Eq.~(\ref{eq:3.1}), which reveals additional superiority of the perception bundle description since it is completely free from such choices. But, we will not pursue this $L (p)$-arbitrariness issue any further in this paper for simplicity (see \cite{avelar2013, lee2022}). Nevertheless, for this reason, the bundle $E_{L,1/2}$ will be called the \textit{boosting bundle for massive particle with spin-1/2}, signifying its dependence of the boosting $L$.

\subsubsection{The perception bundle description}\label{sec:3.3.1}

\hfill

On the bundle $X \times \mathbb{C}^2$, we introduce the following Hermitian metric
\begin{equation}\label{eq:3.16}
h_{1/2} \Big( (p,v) , (p,w) \Big) := \left( L(p)^{-1} v \right) \cdot \left( L(p)^{-1} w \right) = v^\dagger \frac{ \utilde{p}}{m} w.
\end{equation}

The Hermitian bundle $(X \times \mathbb{C}^2 , h_{1/2})$ will be denoted as $E_{1/2}$ and called the \textit{perception bundle for massive particle with spin-1/2}.

Note that the map
\begin{equation}\label{eq:3.17}
\begin{tikzcd}[baseline=(current  bounding  box.center), column sep=1.5em]
    \alpha_{1/2} : E_{1/2} \arrow ["{(p,v) \mapsto (p , L(p)^{-1} v)}"]{rrrrrr} \arrow{drrr} 
 & & & & & & E_{L, 1/2} \arrow{dlll}
\\
  & & &X& & &
\end{tikzcd}
\end{equation}
is an Hermitian bundle isomorphism. Via $\alpha_{1/2}$, the $G$-action Eq.~(\ref{eq:3.9}) is pulled back to the following $G$-action on $E_{1/2}$.

\begin{gather}
\lambda_{1/2} (a, \Lambda ) : E_{1/2} ^A \rightarrow E_{1/2} ^B \nonumber \\
(p,v)^A \mapsto \left(  \Lambda p, e^{- i (\Lambda p)_\mu a^\mu} \Lambda v \right)^B \label{eq:3.18}
\end{gather}

Let $(p, \chi) \in (E_{1/2})_p$ be a qubit, i.e., $ \| \chi \|_{(E_{1/2})_p} = 1$, which is equivalent to $ \| L (p)^{-1} \chi \|_{(E_{L,1/2})_p} = \| L(p)^{-1} \chi \| = 1$. Denote the spin direction of the qubit $L(p)^{-1} \chi \in \mathbb{C}^2$ as $\mathbf{n} \in \mathbb{R}^3 $. Then, by Eqs.~(\ref{eq:2.20}), (\ref{eq:3.13}), and (\ref{eq:2.13}),
\begin{align}\label{eq:3.19}
m \chi \chi^\dagger - \frac{ \tilde{p}}{2} = m L(p) L(p)^{-1} \big( \chi \chi^\dagger - \frac{\tilde{p}}{2m} \big) L(p)^{-1} L(p) \nonumber \\ =  m L (p) \left( \frac{1}{2} \boldsymbol{\tau}\cdot \mathbf{n} \right) L (p)  = \frac{m}{2} \Big(L (p) (0, \mathbf{n}) \Big)^\sim .
\end{align}

So, there is $w \in \mathbb{R}^4$ such that
\begin{equation}\label{eq:3.20}
\tilde{w} = m \chi \chi^\dagger - \frac{ \tilde{p}}{2}
\end{equation}
which will be called the \textit{Pauli-Lubansky four-vector of the qubit $(p, \chi) \in (E_{1/2})_p$} (cf. Example~\ref{example:2.13}). Note that this definition is completely free from any reference to the boostings $L (p)$. $L (p)$ came into the picture for the sole purpose of showing that the RHS of Eq.~(\ref{eq:3.20}) can be represented by an element of $\mathbb{R}^4$ in the form of Eq.~(\ref{eq:2.3}).

Let Alice and Bob be inertial observers related by a Lorentz transformation $(a, \Lambda) \in G$ as in Eq.~(\ref{eq:2.11}). Then, via the isomorphism Eq.~(\ref{eq:3.17}), we see that the perception bundle descriptions of the two observers are related by Eq.~(\ref{eq:3.18}). Substituting $(a, \Lambda) = (0, L (p))$ into the transformation law, we obtain
\begin{equation}\label{eq:3.21}
\lambda_{1/2} ( 0, L (p)) (p_m , \chi)^A = (p, L (p) \chi)^B .
\end{equation}

As in Sect.~\ref{sec:3.2}, Alice prepares a qubit $(p_m, \chi) \in (E_{1/2}^A)_{p_m}$ in her rest frame. By Eq~(\ref{eq:3.21}), the qubit looks like $(p, L (p) \chi)^B \in (E_{1/2} ^B)_p$ in Bob's frame. He forms the Pauli-Lubansky vector for the qubit $(p, L (p)\chi) \in (E_{1/2} ^B)_p$ according to Eq.~(\ref{eq:3.20}) to find that
\begin{equation}\label{eq:3.22}
\tilde{w} = \Big( L (p) ( 0, \frac{\mathbf{n}}{2}) \Big)^\sim
\end{equation}
where $\mathbf{n}$ is the spin direction of the qubit $(p_m, \chi)^A \in (E_{1/2} ^A)_{p_m}$ in Alice's frame. This four-vector $w = L(p) (0, \frac{\mathbf{n}}{2}) \in \mathbb{R}^4$ is exactly the information content of the qubit $(p_m, \chi)^A$ as perceived in Bob's frame (see the paragraph following Eq.~(\ref{eq:3.13})). So, we see that each fiber $(E_{1/2} ^B)_p$ correctly reflects Bob's perception (in the sense of Sect.~\ref{sec:2.5}) of the particle's spin state when the particle is moving with momentum $p$ (hence the name perception bundle). In this regard, choosing the perception bundle description instead of the more standard boosting bundle description seems more sensible in addressing relativistic questions. Also, see \cite{lee2022} for more features of the perception bundle description.

\subsubsection{A relation between the two descriptions; the bundles}\label{sec:3.3.2}

\hfill

The relation between the two bundle descriptions $E_{1/2}$ and $E_{L,1/2}$ is the quantum analogue of the relation between the Pauli-Lubansky four-vector and the Newton-Wigner spin in classical SR (cf. Example~\ref{example:2.13}). More precisely, a qubit $(p, \chi) \in (E_{1/2})_p$ has information of the Pauli-Lubansky four-vector of the particle (cf. Eq.~(\ref{eq:3.20})) and the corresponding vector in the boosting bundle, $\alpha_{1/2} (p, \chi) = (p, L (p)^{-1} \chi) \in (E_{L,1/2} )_p$, has information of the Newton-Wigner spin. To see this, form the spin three-vector $\mathbf{n}$ of $\alpha_{1/2} (p, \chi) = (L (p)^{-1} \chi) \in (E_{L,1/2})_p$ as in Eq.~(\ref{eq:3.13}). Then, by Eqs.~(\ref{eq:3.20}) and (\ref{eq:2.24}),

\begin{align}\label{eq:3.23}
\boldsymbol{\tau} \cdot \left(\frac{1}{2}  \mathbf{n}\right) &=  L (p)^{-1} \chi \chi^\dagger L (p)^{-1} - \frac{1}{2} I_2 \nonumber \\
&=  L (p)^{-1} \left( \frac{\tilde{w}}{m} + \frac{\tilde{p}}{2m}  \right) L (p)^{-1} - \frac{1}{2} I_2 \nonumber \\
& = \frac{1}{m} \left(L (p)^{-1} w \right)^\sim = ((0,\mathbf{s}))^\sim = \boldsymbol{\tau} \cdot \mathbf{s}
\end{align}
where $\mathbf{s}$ is the Newton Wigner spin given by Eq.~(\ref{eq:2.25}) in terms of the Pauli-Lubansky four-vector $w$ of the qubit $(p, \chi)$ given by Eq.~(\ref{eq:3.20}).

So, we conclude that the information contained in the qubits of the bundle $E_{L,1/2}$ in relation to those of $E_{1/2}$ via the bundle isomorphism Eq.~(\ref{eq:3.17}) is precisely the Newton-Wigner spin of the particle. The qubits in the perception bundle are "relativistic perception" just like the Pauli-Lubansky vector is (cf. Example~\ref{example:2.13}), whereas those in the boosting bundle are not, just like the Newton-Wigner spin vector.

Now, we can give a classical analogue for each bundle description. Let $M$ be the Minkowski spacetime. Fix an inertial frame of reference $\varphi = (x^\mu)$ (cf. Definition~\ref{definition:2.7}) and suppose there is a spinning particle with momentum $p^\mu$ and Pauli-Lubansky vector $w^\mu$ with respect to the frame. These information of the particle can be recorded as a point in the tangent bundle $TM$ and expressed as $(p^\mu, w^\mu)$ in the coordinate representation of the chosen frame.

The perception bundle is the faithful quantum analogue of this coordinate representation for moving quantum systems as we have just seen in this subsection. However, the boosting bundle is the quantum version of the altered trivialization $\Big(p^\mu, (L(p)^{-1}w)^\mu \Big)$ of the tangent bundle $TM$, which moreover depends on the choice of the boostings $L$. This is an utter artificiality given the fact that there is certain arbitrariness in the choice of $L$ (see the introduction to Sect.~\ref{sec:3.3} and references therein).

One should note, however, that the boosting bundle description $E_{L,1/2}$ has been the standard approach to the problems in RQI.

\subsubsection{A relation between the two descriptions; the representations}\label{sec:3.3.3}

\hfill

Note that as the action $\lambda_{L,1/2}$ gave rise to the representation $U_{L,1/2}$ (cf. Eq.~(\ref{eq:3.10})), the action $\lambda_{1/2}$ also gives rise to a representation of $G$ by the formula
\begin{equation}\label{eq:3.24}
[U_{1/2} (a, \Lambda) \psi] (p) := \lambda_{1/2} (a, \Lambda) \circ \psi \circ \Lambda^{-1}
\end{equation}
on the Hilbert space
\begin{equation}\label{eq:3.25}
\mathcal{H}_{1/2} := L^2 \left( X , E_{1/2} ; \mu, h_{1/2} \right).
\end{equation}

This representation is equivalent to Eq.~(\ref{eq:3.1}) via the isomorphism Eq.~(\ref{eq:3.17}) and hence can be used to describe massive particles with spin-1/2 (cf. Definition~\ref{definition:3.1}). One may wonder whether the relation between the two bundles $E_{1/2}$ and $E_{L,1/2}$ manifests itself on the level of the two representations $U_{1/2}$ and $U_{L,1/2}$. Later, we will see that this is indeed the case. But, we are forced to defer the discussion until Sect.~\ref{sec:6.2.2} since we need to introduce several quantum operators before we can precisely state in what sense this is true.

\subsection{A preview of the main results of the paper}\label{sec:3.4}

We have seen that the fundamental perplexities posed by the two pioneering papers of RQI (\cite{peres2002, gingrich2002}) have arisen because the standard representation that has been predominantly used in the RQI literature (Eq.~(\ref{eq:3.1})) was constructed from a "wrong" bundle (the boosting bundle) and that by using a "right" bundle, the fibers of which correctly reflect relativistic perception (cf. Sect.~\ref{sec:2.5}) of inertial frames (the perception bundle), the perplexities are resolved.

A natural question that immediately comes to one's mind would be that whether the same kind of bundle theoretic descriptions are possible for all kinds of particles, not just the massive spin-1/2 ones. In this paper, we are going to show that this is indeed possible for massive particles with arbitrary spin\footnote{We leave out the massless case to a sequel paper.}, i.e., we are going to construct bundles whose fibers correctly reflect relativistic perception of inertial frames. Also, we will explore some of the theoretical implications of this bundle theoretic description. Specifically, we will see that some of the fundamental equations of Quantum Field Theory (QFT) are just manifestations of relativistic perception of inertial observers.

\subsection{Other approaches to RQI and the scope of the paper}\label{sec:3.5}

Before we begin, we want to mention other existing apporaches to RQI that are not covered in this paper and how the results of the present paper are related to them.

First, localized quantum systems that are relevant to quantum informational tasks such as moving cavities, point-like detectors, and localized wave packets have been discussed in the literature based on the language of QFT (cf. \cite{alsing2012}). Since single-particle state spaces are basic building blocks of QFT (cf. \cite{weinberg}), the results of the present paper are closely related to this approach (cf. \cite{caban2006}).\footnote{In fact, we expect that the results of this paper will give a new insight into the QFT approach.} However, we do not need to use the language of QFT in this paper since we restrict our attention to how the principle of SR affects our perception of the quantum reality (of which the single-particle state spaces are the simplest examples) and leave out applications of the theory introduced in this paper to actual quantum informational scenarios, which might require QFT, to a future research direction (see the concluding remarks in Sect.~\ref{sec:8}). Those who are interested in the QFT approach are referred to \cite{alsing2012} and references therein.

Also, in order to apply the results of the present paper to actual problems of RQI (which we leave out as a future research), one needs to know the theory of relativistic quantum measurement. One can find a good treatment in Ch.~11 of \cite{breuer}. We will give a link between this theory and some of the results of the present paper in Sect.~\ref{sec:7.4}.

\section{Single-particle state spaces}\label{sec:4}
In this section, we identify and classify the single-particle state spaces that are called "massive particles". The main technical tool that is needed to obtain a classification of single-particle state spaces is "Mackey machine". Let's set the stage for the main technical theorem. All the discussions until Theorem \ref{theorem:4.2} can be found in \cite{folland2015}. First, we define induced representations.
\begin{definition}[Induced representation]\label{definition:4.1}
Let $G$ be a locally compact group and $H \leq G$ be a closed subgroup such that there is a $G$-invariant measure $\mu$ on $G/H$. Given a unitary representation $\sigma$ of $H$ on the Hilbert space $\big(\mathcal{H}_\sigma , \langle \cdot, \cdot \rangle_\sigma \big)$, define
\begin{subequations}\label{eq:4.1}
\begin{equation}\label{eq:4.1a}
\mathcal{F}_0 := \left\{ f \in \mathcal{B}(G, \mathcal{H}_\sigma) :  f(gh) = \sigma(h)^{-1} f(g) \small{\text{ (${}^\forall h \in H$),}} \int_{G/H} \|f(x)\|_\sigma ^2 d \mu (\dot{x}) < \infty  \right\}
\end{equation}
where $\dot{x} =xH$, $\| \xi \|_\sigma ^2 := \langle \xi , \xi \rangle_\sigma$ for $\xi \in \mathcal{H}$, and $\mathcal{B} (G, \mathcal{H}_\sigma)$ is the set of all Borel functions from $G$ into $\mathcal{H}_\sigma$. \footnote{A discussion about the measurability of Banach space-valued functions can be found, for example, in Appendix~B of \cite{williams}. But, in all the cases that we will be looking at in this paper, $\mathcal{H}_\sigma$ is finite dimensional, to which the elementary measure theory as presented in \cite{rudin2} can be applied.} $\mathcal{F}_0$ is a vector space, of which $N = \{ f \in \mathcal{F}_0 : \int_{G/H} \| f(x) \|_\sigma ^2 d \mu ( \dot{x}) = 0\}$ is a subspace. Let
\begin{equation}\label{eq:4.1b}
\mathcal{F} := \mathcal{F}_0 / N.
\end{equation}
\end{subequations}

Then, the map $(f,g) \mapsto \int_{G/H} \langle f(x),g(x) \rangle _\sigma d \mu (\dot{x})$ is a well-defined inner product on the vector space $\mathcal{F}$, with respect to which $\mathcal{F}$ becomes a Hilbert space.

On this Hilbert space, we have a (continuous) unitary representation $\text{Ind}_H ^G (\sigma) : G \rightarrow U(\mathcal{F})$ defined by
\begin{equation}\label{eq:4.2}
\left[ \text{Ind}_H ^G (\sigma) (x) f \right] (y) = f( x^{-1} y).\footnote{Induced representations can be defined for all closed subgroups $H \leq G$, including those which have no $G$-invariant measure on $G/H$. However, we will not need this generality in this paper. Notice that Theorem~\ref{theorem:4.2} is stated without reference to invariant measures.}
\end{equation}  \qed
\end{definition}

Now, let $G$ be a second countable locally compact group and $N$ be a closed abelian normal subgroup of $G$ such that $G = N \ltimes H$ for some closed subgroup $H \leq G$, which means that the map $N \times H \rightarrow G$ given by $(n,h) \mapsto nh$ is a homeomorphism.

Then, $G$ has a natural left action on $\hat{N}$ (the dual of $N$) given by $x \cdot \nu (n) = \nu ( x^{-1} n x) $ for $x \in G$ and $ \nu \in \hat{N}$. Let $G_\nu$ be the isotropy group for $\nu \in \hat{N}$. We call $H_{\nu} = H \cap G_\nu$ the \textit{little group} for $\nu$. Note that $G_\nu = N \ltimes H_\nu$.

Let $\nu \in \hat{N}$ and $\sigma : H_{\nu} \rightarrow U(\mathcal{H}_\sigma)$ be an irreducible representation of $H_\nu$. Then, the map $\nu \sigma : G_{\nu} \rightarrow U(\mathcal{H}_\sigma)$ defined by
\begin{equation}\label{eq:4.3}
\nu \sigma (n h) = \nu (n) \sigma (h) \quad \text{ $(n \in N, h \in H_\nu)$} 
\end{equation}
is a well-defined irreducible representation of $G_\nu$.

We say that the action of $G$ on $\hat{N}$ is \textit{regular} if the natural bijections $G/G_\nu \rightarrow G \cdot \nu$ are homeomorphisms for all $\nu \in \hat{N}$ when $G\cdot \nu$ is endowed with the subspace topology. Now, let's state the main technical theorem.

\begin{theorem}[\cite{folland2015}, Theorem~6.43]\label{theorem:4.2}
Suppose $G=N \ltimes H$, where $N$ is a closed abelian normal subgroup, $H$ a closed subgroup, and the second countable group $G$ acts regularly on $\hat{N}$. Then, the following conclusions hold.

\begin{enumerate}
\item If $\nu \in \hat{N}$ and $\sigma$ is an irreducible representation of $H_\nu$, then $\textup{Ind}_{G_\nu} ^G (\nu \sigma )$ is an irreducible representation of $G$.

\item Every irreducible representation of $G$ is equivalent to one of this form.

\item $\textup{Ind}_{G_\nu} ^G (\nu \sigma) $ and $\textup{Ind}_{G_{\nu'}} ^G (\nu' \sigma') $ are equivalent if and only if $\nu$ and $\nu'$ belong to the same orbit, say $\nu' = x \nu$, and $h \mapsto \sigma(h)$ and $h \mapsto \sigma'(xhx^{-1})$ are equivalent representations of $H_\nu$. 
\end{enumerate}
\end{theorem}

Let's apply this theorem to the group $G := \mathbb{R}^4 \ltimes SL(2,\mathbb{C})$. From here on, we will follow the approach of \cite{folland2008}. Observe that $\mathbb{R}^4$ is a closed abelian normal subgroup. To show that the action is regular, we find all the orbits $G \cdot \nu$ in $\hat{\mathbb{R}}^4$.

First, the map $p \mapsto e^{-i \langle p, \cdot \rangle}$ is a topological group isomorphism from $\mathbb{R}^4$ onto $\hat{\mathbb{R}}^4$ (\cite{folland2015}, p.98). Via this isomorphism, the natural action of $G$ on $\hat{\mathbb{R}}^4$ as defined in this section translates into a $G$-action on $\mathbb{R}^4$ given by $((a,\Lambda) , p) \mapsto \kappa(\Lambda) p$ since
\begin{align*}
 \left[ (a,\Lambda) \cdot e^{-i \langle p , \cdot \rangle} \right] (b , I) = e^{-i \langle p, \cdot \rangle} \left(  (a, \Lambda)^{-1 } (b, I ) (a, \Lambda) \right)  
= e^{-i \langle p, \cdot \rangle} (\kappa (\Lambda)^{-1} b , I) \\ = e^{-i \langle p, \kappa(\Lambda)^{-1} b \rangle} = e^{-i \langle \kappa(\Lambda)p , b \rangle } = \left[ e^{-i \langle \kappa(\Lambda) p , \cdot \rangle} \right] ( b, I)
\end{align*} 
where we used the fact that $\kappa(\Lambda) \in SO^\uparrow(1,3)$ preserves the scalar product $\langle \cdot , \cdot \rangle$. For the rest of the paper, we shall remove $\kappa$ from all expressions involving an action of $\Lambda \in SL(2, \mathbb{C})$ on $ x \in \mathbb{R}^4$ via $\kappa$, i.e., we write $\Lambda x$ for $\kappa (\Lambda) x $. So, the action of $G$ on $\hat{\mathbb{R}}^4 \cong \mathbb{R}^4$ becomes
\begin{equation}\label{eq:4.4}
(a,\Lambda) \cdot p = \Lambda p,
\end{equation}
from which we see that the $G$-orbits of this action are exactly the $SO^{\uparrow}(1,3)$ orbits of its canonical action on $\mathbb{R}^4$.
\begin{proposition}\label{proposition:4.3}
The $G$-orbits in $\mathbb{R}^4$ are exactly the $SO^{\uparrow}(1,3)$-orbits in $\mathbb{R}^4$, which consist of
\begin{equation}\label{eq:4.5}
X_m ^+ = \{p : p_\mu p^\mu = m^2, p^0 > 0 \} \text{, } X_m ^- = \{p : p_\mu p^\mu = m^2, p^0 < 0 \}
\end{equation}
for $0 \leq m < \infty$,
\begin{equation}\label{eq:4.6}
Y_m = \{ p : p_\mu p^\mu = - m^2 \}
\end{equation}
for $0 < m < \infty$, and
\begin{equation*}
    \{0\}.
\end{equation*}

The following can be used as representatives (elements $p \in \mathbb{R}^4$ for $G \cdot p$) for these orbits:

For $X_m ^\pm$,
\begin{equation}\label{eq:4.7}
    p_m ^{\pm} : = (\pm m, 0, 0, 0).
\end{equation}

For $X_0 ^{\pm}$,
\begin{equation}\label{eq:4.8}
    p_0 ^{\pm} : = (\pm 1, 0, 0, \pm 1).
\end{equation}

For $Y_m$,
\begin{equation*}
    q_m = (0, 0, m, 0).
\end{equation*}

For $\{0\}$, $0$.
\end{proposition}
\begin{proof}
The proof is easy once one notices that each subset listed above is $SO^\uparrow(1,3)$-invariant.
\end{proof}

Note that $X$ and $p_m$ used in Sect.~\ref{sec:3} are equal to $X_m ^+$ and $p_m ^+$, respectively.

These orbits are all embedded submanifolds of $\mathbb{R}^4$ and the bijections $G/G_\nu \rightarrow G \cdot \nu$ are $G$-equivariant smooth maps between transitive $G$-manifolds. So, these bijections are of constant rank and hence diffeomorphisms when $G \cdot \nu$ are endowed with the subspace topologies (cf. \cite{lee}), which implies that the action of $G$ on $\mathbb{R}^4 \cong \hat{\mathbb{R}}^4$ is regular. Therefore, we can apply Theorem \ref{theorem:4.2} to $G$.

\begin{remark}\label{remark:4.4}
So, if $p \in \mathbb{R}^4$ and $H_p$ is the corresponding little group (that is, the isotropy subgroup of $H = SL(2,\mathbb{C})$), then every irreducible representation $\sigma : H_p \rightarrow U(\mathcal{H}_\sigma )$ induces an irreducible representation
\begin{equation*}
\rho_{p, \sigma} (a,\Lambda)  = \exp(-i p_\mu a^\mu ) \sigma(\Lambda) = e^{-i \langle p , a \rangle} \sigma(\Lambda)
\end{equation*}
of $G_p$ (cf. Eq.~(\ref{eq:4.3})), which in turn induces an irreducible representation
\begin{equation*}
\pi_{p, \sigma} = \text{Ind}_{G_p} ^G (\rho_{p, \sigma})
\end{equation*}
of $G$ by Theorem~\ref{theorem:4.2}.1. Moreover, Theorem~\ref{theorem:4.2}.2 asserts that every irreducible representation of $G$ arises in this way and Theorem~\ref{theorem:4.2}.3 tells us that if we restrict the choice of $p \in \mathbb{R}^4$ to the chosen representatives listed in Proposition~\ref{proposition:4.3}, the resulting representations are all distinct.

So, the classification of single-particle state spaces will be completed once we calculate the little group $H_p$ for each representative $p$ listed in Proposition~\ref{proposition:4.3}, find all irreducible representations $\sigma$ of this little group, and calculate $\pi_{p, \sigma}$. In this paper, we will only consider the representations associated with the orbits $X_m ^\pm$ for $ m>0$, which correspond to \textit{massive particles}. 
\end{remark}

Let's investigate the physical meaning of the constant $m$ which was used to classify the orbits as in Proposition~\ref{proposition:4.3}. Let $q \in \mathbb{R}^4$ be any element and $ \sigma:H_{q} \rightarrow U(\mathcal{H}_\sigma)$ be an irreducible representation of the little group $H_{q}$ for $q$. By unraveling Definition~\ref{definition:4.1}, the induced irreducible representation $\pi_{q, \sigma}:G \rightarrow U(\mathcal{F})$ satisfies, for $b \in \mathbb{R}^4$, $f \in \mathcal{F}$, and $(a, \Lambda) \in G$,
\begin{align}\label{eq:4.9}
    \left[ \pi_{q, \sigma} (b, I) f  \right] ((a, \Lambda )) = f( (-b, I) (a, \Lambda )) = f( (a, \Lambda) (-\Lambda^{-1} b , I)) \nonumber \\
= e^{-i \langle q, \Lambda^{-1}b \rangle} f ((a, \Lambda)) =e^{-i \langle \Lambda q , b \rangle} f( (a, \Lambda)).
\end{align}

Write $p = \Lambda q \in \mathbb{R}^4$. Since $\pi_{q, \sigma} (b, I)$ would represent spacetime translations (cf. Remark~\ref{remark:2.11}), we see that the \textit{four-momentum operators} $P^\mu$ on this representation space, which are by definition the infinitesimal generators of the spacetime translations (cf. \cite{hall}), are given by the following formulae
\begin{subequations}\label{eq:4.10}
\begin{align}
\left [P^0 f \right] (a,\Lambda) &:= \left[ i \frac{\partial}{\partial b^0} \pi_{q, \sigma} (b, I) f \right] (a,\Lambda) =  p_0 f(a,\Lambda) = p^0 f(a,\Lambda) \label{eq:4.10a} \\
\left [P^j f \right] (a,\Lambda) &:= \left[- i \frac{\partial}{\partial b^j} \pi_{q, \sigma} (b, I) f \right] (a,\Lambda) =  -p_j f(a,\Lambda) = p^j f(a,\Lambda). \label{eq:4.10b}
\end{align}
\end{subequations}
which are (unbounded) multiplication operators.

Since $p_\mu p^\mu = \langle \Lambda q, \Lambda q \rangle = q_\mu q^\mu$, the operator $P_\mu P^\mu = (P^0)^2 - (P^1)^2 - (P^2)^2 - (P^3)^2$ acts on the $\pi_{q, \sigma}$-representation space as $f \mapsto (q_\mu q^\mu) f$, the multiplication by the constant $q_\mu q^\mu$. So, we see that all the vectors in the representation space of $\pi_{q,\sigma}$ are eigenvectors of the operator $P_\mu P^\mu$ with the eigenvalue $q_\mu q^\mu$.
Inspired by the famous energy-momentum relation from SR (cf. \cite{zee}), we make the following definition.
\begin{definition}\label{definition:4.5}
The \textit{mass} of the single-particle states associated with the irreducible representation $\pi_{q, \sigma}$ is the constant $M = \sqrt{q_\mu q^\mu}$.
\end{definition}

For the orbits $X_m ^{\pm}$, we have $M = m$ and hence the nonnegative number $m$ represents the mass of the particles associated with the orbit $X_m ^{\pm}$. For this reason, the orbits $X_m ^{\pm}$ are called the \textit{mass shells}. But, for the orbits $Y_m$, $M$ is an imaginary number (and hence there is an ambiguity in the definition of $M$). In \cite{weinberg}, it is stated that there is no known interpretation, in terms of physical states, of the states associated with the orbits $Y_m$.

From now on, we will focus our attention on the representations associated with the orbits $X_m ^\pm$ with $m >0 $, the \textit{massive particles}, leaving the analysis of the orbit $X_0 ^\pm$, the \textit{massless particles}, to a sequel paper.

Let's embark on the job that was set in Remark \ref{remark:4.4} for $X_m ^\pm$ with $m>0$.
\begin{proposition}\label{proposition:4.6}
For $m>0$, the little group $H_{p_m ^\pm}$ for $p_m ^{\pm}$ is $SU(2)\leq SL(2, \mathbb{C})$.
\end{proposition}
\begin{proof}
$A \in H_{p_m ^\pm} $ if and only if $A p_m ^\pm = p_m ^\pm$, i.e. by Eq.~(\ref{eq:2.13}), if and only if $ \pm m A A^\dagger = \pm mI_2 $.  
\end{proof}

The irreducible representations of the group $SU(2)$ are well-known to both mathematicians and physicists. But, for later discussions, we need a concrete realization. The following arguments are adapted from \cite{varadarajan}.

Let $\mathcal{V} := \mathbb{C}^2$. Fix $s \in \frac{1}{2} \mathbb{N}_0$ and consider the following vector space
\begin{equation}\label{eq:4.11}
V_s := \Sigma^{2s} (\mathcal{V} )  = \mathcal{V}^{\otimes {2s}} / N^{2s}
\end{equation}
where
\begin{align}\label{eq:4.12}
N^{2s} = \text{span}_{\mathbb{C}} \{ x_1 \otimes \cdots \otimes x_s - x_{\tau(1)} \otimes \cdots \otimes x_{\tau(2s)} \in \mathcal{V}^{\otimes 2s} :\nonumber \\
x_1 , \cdots , x_s \in \mathcal{V}, \hspace{0.1cm} \tau \in S_{2s} \}.
\end{align}

We denote the image of $x_1 \otimes \cdots \otimes x_{2s}$ in the quotient space $V_s$ as $x_1 \cdots x_{2s}$. There is a natural embedding $\Sigma^{2s} (\mathcal{V}) \rightarrow \mathcal{V}^{\otimes 2s}$ given by
\begin{equation}\label{eq:4.13}
x_1 \cdots x_{2s} \mapsto \frac{1}{ (2s)!} \sum_{\tau \in S_{2s}} x_{\tau (1)} \otimes \cdots \otimes x_{\tau(2s)}.
\end{equation}

Let $\langle \cdot , \cdot \rangle$ be the Hermitian inner product on $\mathcal{V} = \mathbb{C}^2$. It extends to a unique inner product $\langle \cdot , \cdot \rangle$ on $\mathcal{V}^{\otimes 2s}$ satisfying
\begin{equation}\label{eq:4.14}
\langle x_1 \otimes \cdots \otimes x_{2s} , y_1 \otimes \cdots \otimes y_{2s} \rangle = \langle x_1 , y_1 \rangle \cdots \langle x_{2s} , y_{2s} \rangle.
\end{equation}

Via the embedding Eq.~(\ref{eq:4.13}), $V_s$ inherits this inner product to become an inner product space. Denote $u = \begin{pmatrix} 1 \\ 0 \end{pmatrix} , v = \begin{pmatrix} 0 \\ 1 \end{pmatrix} \in \mathbb{C}^2$. Then, $V_s$ is of dimension $2s + 1$ with an orthonormal basis give by
\begin{equation}\label{eq:4.15}
\mathcal{B} = \left\{\sqrt{\frac{(2s)!}{ k! (2s -k)! }} u^k v^{2s-k} : 0 \leq k \leq 2s \right\}.
\end{equation}

Given a linear map $T : \mathcal{V} \rightarrow \mathcal{V}$, the map $T^{\otimes 2s} : \mathcal{V}^{\otimes 2s} \rightarrow \mathcal{V}^{\otimes 2s}$ restricts to a well-defined linear map $\Sigma^{2s} (T) : V_s \rightarrow V_s$ defined on the basic elements $x_1 \cdots x_{2s}$ by
\begin{equation}\label{eq:4.16}
\Sigma^{2s} (T) \left( x_1 \cdots x_{2s} \right) = (T x_1 ) \cdots (T x_{2s}),
\end{equation}
which is unitary if $T$ is unitary.

So, the map $\sigma_s : SU(2) \rightarrow U(V_s)$ defined by
\begin{equation}\label{eq:4.17}
\sigma_s (A) = \Sigma^{2s}  (A) 
\end{equation}
is a unitary representation of $SU(2)$ on the ($2s +1$)-dimensional Hilbert space $V_s$, which has a natural extension $\Phi_s : SL(2, \mathbb{C}) \rightarrow GL( V_s)$ given by
\begin{equation}\label{eq:4.18}
\Phi_s (A) = \Sigma^{2s}  (A) .
\end{equation}

To show that the representations $\sigma_s$ are irreducible, we need the following well-known facts about the Lie algebras $\mathfrak{su}(2) \leq \mathfrak{sl}(2, \mathbb{C})$. Recalling the definitions of the Pauli matrices (Eq.~(\ref{eq:2.2})),
\begin{subequations}\label{eq:4.19}
\begin{eqnarray}
J^j = -\frac{i}{2} \tau^j \label{eq:4.19a} \in \mathfrak{su}(2) \\
K^j = \frac{1}{2} \tau^j \label{eq:4.19b} \in \mathfrak{sl}(2, \mathbb{C})
\end{eqnarray}
\end{subequations}
are respectively called the \textit{angular momentum} and the \textit{boosting} along the $j$-th axis. Thery are $\mathbb{R}$-linearly independent, and
\begin{subequations}\label{eq:4.20}
\begin{eqnarray}
\mathfrak{su}(2) &=& \text{span}_{\mathbb{R}} (J^1, J^2 , J^3 ) \label{eq:4.20a} \\
\mathfrak{sl}(2, \mathbb{C}) &=& \text{span}_{\mathbb{R}} (J^1, J^2 , J^3 , K^1 , K^2 , K^3) \label{eq:4.20b}.
\end{eqnarray}
\end{subequations}

Now, returning to $\sigma_s$ and $\Phi_s$, observe that for $B \in \mathfrak{sl}(2, \mathbb{C})$, one can prove, by carrying out a differentiation, that
\begin{align}\label{eq:4.21}
(\Phi_{s})_* (B) ( x_1 \cdots x_{2s} ) &= \nonumber \\
(Bx_1) x_2 & \cdots x_{2s} + x_1 (Bx_2) x_3 \cdots x_{2s} + \cdots x_1 \cdots x_{2s -1} (B x_{2s}).
\end{align}
 
Define $\hat{J}^k := i (\sigma_s)_* (J^k) = i (\Phi_s )_* (J^k)$. Using Eq.~(\ref{eq:4.21}), we see
\begin{subequations}\label{eq:4.22}
\begin{align}
\hat{J}^1 u^k v^{2s - k} &= \frac{k}{2} u^{k-1} v^{2s - k +1} + \frac{2s -k }{2} u^{k+1} v^{2s - k -1} \label{eq:4.22a} \\
\hat{J}^2 u^k v^{2s - k} &= \frac{i k}{2} u^{k-1} v^{2s - k +1} - i \frac{2s - k }{2} u^{k+1} v^{2s - k - 1} \label{eq:4.22b} \\
\hat{J}^3 u^k v^{2s - k} &= (k - s) u^{k} v^{2s - k} \label{eq:4.22c}
\end{align}
\end{subequations}
and hence
\begin{subequations}\label{eq:4.23}
\begin{align}
\left( \hat{J}^1 + i \hat{J}^2 \right) u^{2s} &= \left( \hat{J}^1 - i \hat{J}^2 \right) v^{2s} = 0 \\
\left( \hat{J}^1 + i \hat{J}^2 \right) u^{k} v^{2s - k} &= (2s - k) u^{k+1} v^{2s - k -1}, \quad 0 \leq k \leq 2s -1 \\
\left( \hat{J}^1 - i \hat{J}^2 \right) u^{k} v^{2s - k} &= k u^{k-1} v^{2s - k + 1} \quad 1 \leq k \leq 2s.
\end{align}
\end{subequations}

\begin{theorem}\label{theorem:4.7}
The representations $\{ \sigma_s : s \in \frac{1}{2} \mathbb{N}_0 \}$ are irreducible, distinct, and exhaust all irreducible representations of $SU(2)$. Note that the orthonormal basis Eq.~(\ref{eq:4.15}) consists of the eigenvectors of the operator $\hat{J}^3$ whose eigenvalues are given by $k-s$ for $0 \leq k \leq 2s$, respectively.
\end{theorem}
\begin{proof}
This follows from Eq.~(\ref{eq:4.23}) and a usual argument involving the "ladder operators" $\hat{J}^1 \pm \hat{J}^2$ (cf. \cite{hall}, pp.371--375 and Proposition~16.39). Note that since $SU(2)$ is compact, every irreducible representation of it is finite dimensional (cf. \cite{folland2015}). The statement about the orthonormal basis follows from Eq.~(\ref{eq:4.22c}).  
\end{proof}

So, if we define
\begin{equation}\label{eq:4.24}
\pi_{m,s} ^{\pm} := \pi_{p_m ^{\pm} , \sigma_s}
\end{equation}
following the procedure of Remark~\ref{remark:4.4}, we see from Theorem~\ref{theorem:4.2} that each $\pi_{m,s} ^{\pm}$ is distinct for each value of $m >0, s = 0 , \frac{1}{2}, 1, \frac{3}{2}, \cdots $, and the $\pm$ signs, and they exhaust all irreducible representations associated with the orbits $X_m ^{\pm}$.

\begin{remark}\label{remark:4.8}
We know from non-relativistic quantum mechanics that if a particle is described by the states in $L^2 (\mathbb{R}^3) \otimes V_s$ with $V_s$ carrying an irreducible $SU(2)$-representation given in Theorem~\ref{theorem:4.7}, then the number $s$ is called the \textit{spin} of the particle (cf. Ch.~17 of \cite{hall}). We will see a direct link between this tensor product space and the representation space of $\pi_{m, s} ^\pm$ later (see the rightmost column of Table~\ref{tab:2}).
\end{remark}

This remark suggests the following definition.

\begin{definition}\label{definition:4.9}
The value $s$ for the irreducible representation $\pi_{m, s} ^{\pm}$ is called the \textit{spin} of the single-particle states associated with this representation.
\end{definition}

The following is the conclusion of this section.

\begin{theorem}\label{theorem:4.10}
The irreducible representations associated with the orbits $X_m ^{\pm}$ with $m > 0$ are classified by \textbf{mass} and \textbf{spin}, i.e., they are precisely
\begin{equation}\label{eq:4.25}
\left\{\pi_{m,s} ^{\pm} : m >0 , s = 0 , \frac{1}{2},  1,  \frac{3}{2}, \cdots  \right\}
\end{equation}
and they descend to projective representations of the group $\mathbb{R}^4 \ltimes SO^\uparrow (1,3)$ as in Theorem~\ref{theorem:2.9} (cf. Remark~\ref{remark:2.11}). In fact, $\pi_{m, s} ^\pm (-I) = (-1)^{2s}$.
\end{theorem}
\begin{proof}
The first assertion is just the summation of the preceding discussions. For the second statement, observe that, by unravelling Definition~\ref{definition:4.1},
\begin{align*}
\left [ \pi_{m, s} ^\pm (0, -I) f \right] ((a,\Lambda)) =  f( (0, -I) (a,\Lambda) )   f (  (a,\Lambda) (0, -I) ) \\  = \sigma_s (-I) f(a,\Lambda) 
 = (-1 )^{2s} f(a, \Lambda),
\end{align*}
where, in the last equality, we used the identity
\begin{equation*}
\sigma_s ( -I) = \sigma_s ( e^{2 \pi J^3}) = e^{2 \pi (\sigma_s)_* (J^3) } = e^{- 2  \pi i \hat{J}^3} = (-1)^{2s}
\end{equation*}
which follows from Theorem~\ref{theorem:4.7}.  
\end{proof}

The representations $\pi_{m,s} ^-$ associated with $X_m ^-$ do not seem to represent realistic particles since the associated particles have negative energy (cf. Eq.~(\ref{eq:4.10a})). But, by introducing the concept of quantum fields, we can interpret them as representing \textit{antiparticle} states (cf. \cite{folland2008, weinberg}).

\section{A bundle theoretic description of induced representation}\label{sec:5}

In this section, we develop a relevant mathematical theory that will be needed in the following discussions on RQI. We assume that the readers are familiar with the basic notions of vector bundles such as sections, metrics, subbundles, and tensor products, etc. These materials can be found in \cite{lee} and \cite{tu}. All the pre-induced representations in this section will be assumed to be smooth and finite-dimensional, and all the bundles, sections, and bundle homomorphisms appearing in this section will mean smooth ones unless stated otherwise.

\paragraph{Section spaces}

\hfill

Let $E \xrightarrow{\xi} M$ be a complex vector bundle. We denote its smooth and continuous section spaces by $C^\infty (M , E )$ and $C (M, E)$, respectively. We define the \textit{Borel-section space of $E$} as
\begin{equation}\label{eq:5.1}
\mathcal{B} (M, E) := \{ \psi :M \rightarrow E | \text{$\psi$ is a Borel map and } \xi \circ \psi (x) = x, \hspace{0.1cm} \forall x \in M \}.
\end{equation}

It is an easy exercise to check that $\mathcal{B}(M,E)$ becomes a vector space with respect to the pointwise addition and scalar multiplication. In fact, it is a module over $\mathcal{B}(M)$, the ring of Borel functions on $M$.

Let $\mu$ be a positive Borel measure on $M$ and $g$ be an Hermitian metric on $E$. We define the \textit{$L^2$-section space of $E$} as
\begin{equation}\label{eq:5.2}
L^2 (M, E ; \mu , g) := \left\{\psi \in \mathcal{B} (M, E) : \int_M g( \psi, \psi ) \mu < \infty \right\}.
\end{equation}

\begin{proposition}\label{proposition:5.1}
Upon identifying almost everwhere equal functions, $L^2 (M, E ; \mu , g)$ becomes a Hilbert space with the inner product
\begin{equation}\label{eq:5.3}
\langle \psi , \phi \rangle = \int_M g(\psi, \phi) d\mu.
\end{equation}
\end{proposition}
\begin{proof}
We omit the proof.
\end{proof}

\paragraph{Hermitian \texorpdfstring{$G$}{TEXT}-bundle}

\hfill
\begin{definition}\label{definition:5.2}
Let $E \xrightarrow{\xi} M$ be a vector bundle. Let $G$ be a Lie group. Suppose there are $G$-actions $\lambda : G \times E \rightarrow E$ and $l : G \times M \rightarrow M$ such that for each $s \in G$, the following diagram commutes
\begin{equation}\label{eq:5.4}
\begin{tikzcd}[baseline=(current  bounding  box.center)]
E  \arrow[r, "\lambda(s)"] \arrow{d}[swap]{\xi}
&  E \arrow[d, "\xi"] \\
M \arrow[r, "l_s"]  & M
\end{tikzcd}.
\end{equation}

If $E$ is endowed with a metric $g$ with respect to which each $\lambda(s)$ becomes an isometric bundle isomorphism, then we call the triple $(\xi,  g, \lambda)$ an \textit{Hermitian $G$-bundle}. When the base space $M$ is understood, we often write it simply as $(E,  g, \lambda)$ and call it an \textit{Hermitian $G$-bundle over $M$}.

Given two Hermitian $G$-bundles $(E, g, \lambda)$ and $(E', g', \lambda')$ over $M$, $G$-equivariant isometric homomorphisms from $E$ into $E'$ over $M$ are called \textit{Hermitian $G$-bundle homomorphisms from $(E, g, \lambda)$ into $(E', g', \lambda')$ over $M$}. 
\end{definition}

Hermitian $G$-bundles are related to induced representation by the following construction.

\begin{definition}\label{definition:5.3}
Let $(M, \mu)$ be a (left) $G$-invariant measure space and $(E,g,\lambda)$ be an Hermitian $G$-bundle over $M$. Then, the map $U : G \rightarrow U\Big(L^2 (M,E;\mu, g)\Big)$ given by
\begin{equation}\label{eq:5.5}
U (s) f = \lambda(s) \circ f \circ (l_s)^{-1}
\end{equation}
is easily seen to be a (strongly continuous) unitary representation. This representation is called the \textit{induced representation associated with $(E,g,\lambda ; \mu)$}.
\end{definition}

As we shall see, these representations have a close relationship with the induced representation introduced in Definition~\ref{definition:4.1}.

\begin{proposition}\label{proposition:5.4}
Let $(M, \mu)$ be a $G$-invariant measure space and let $(E,g,\lambda) \xrightarrow{\alpha} (E',g', \lambda')$ be an Hermitian $G$-bundle isomorphism over $M$. Then, the map $\alpha: L^2( M, E;  \mu, g) \xrightarrow{\alpha \circ ( \hspace{0.05cm} \cdot \hspace{0.05cm})} L^2 (M, E'; \mu' , g')$ is a Hilbert space isomorphism and gives a unitary equivalence between the two induced representations. I.e., if we denote the associated induced representations by $U$ and $U'$ respectively, then
\begin{equation}\label{eq:5.6}
\alpha U (s) = U' (s) \alpha
\end{equation}
for all $s \in G$.
\end{proposition}
\begin{proof}
$\alpha$ is a Hilbert space isomorphism because it is isometric on the level of the bundles. Observe that, for $\psi \in L^2 (M, E ; \mu , g)$,
\begin{align*}
\alpha U (s) \psi = \alpha \circ (\lambda(s) \circ \psi \circ (l_s)^{-1} ) = \lambda'(s) \circ (\alpha \circ \psi) \circ (l_s)^{-1} \\
= U' (s)  \alpha \psi
\end{align*}
due to the $G$-equivariance of $\alpha$. 
\end{proof}

The following theorem is the main result of this section.

\begin{theorem}\label{theorem:5.5}
Let $G = N \ltimes H$ and fix $\nu \in \hat{N}$. Suppose $\sigma$ is a unitary representation of $H_\nu $ on the Hilbert space $( \mathcal{H}_\sigma , \langle \cdot , \cdot \rangle_\sigma)$ that extends to a representation $\Phi : H \rightarrow GL( \mathcal{H}_\sigma)$, $L : H/H_\nu \rightarrow H$ is a global section, and there is an $H$-invariant measure $\mu$ on $H/H_\nu$.

Define a group element
\begin{equation}\label{eq:5.7}
W_L (x, yH) := L(xyH)^{-1} x L(yH) \in H
\end{equation}
which will be called the \textbf{Wigner transformation} and consider the two Hermitian $G$-bundles in Table~\ref{tab:1} and their associated induced representations $U$ and $U_L$. Then,
\begin{equation}\label{eq:5.8}
\textup{Ind}_{G_\nu} ^G \nu \sigma \cong U \cong U_L
\end{equation}
and the map
\begin{equation}\label{eq:5.9}
\begin{tikzcd}[baseline=(current  bounding  box.center), column sep=1.5em]
    \alpha: E_{\sigma} \arrow ["{(xH_\nu,v) \mapsto (xH_\nu , \Phi \big(L(xH_\nu) ^{-1} \big)v)}"]{rrrrrr} \arrow{drrr} 
 & & & & & & E_{L, \sigma} \arrow{dlll}
\\
  & & &H/H_\nu& & &
\end{tikzcd}
\end{equation}
is an Hermitian $G$-bundle isomorphism that intertwines the structures listed in Table~\ref{tab:1}.

\begin{table}[h]

\caption{The structures of the perception bundle and boosting bundle}
\label{tab:1}

\centering
\begin{tabular}{|m{1.5cm}|m{5.2cm}|m{7cm}|}
\hline\noalign{\smallskip}
  & $E_\sigma$  (The perception bundle) & $E_{L, \sigma}$  (The boosting bundle) \\
\noalign{\smallskip}\hline\noalign{\smallskip}
Bundle  &  $H/H_\nu \times \mathcal{H}_\sigma$  &  $H/H_\nu \times \mathcal{H}_\sigma$ \\
\noalign{\smallskip}\hline\noalign{\smallskip}
Metric &   $h \Big( (xH_\nu , v) , (xH_\nu, w) \Big)$ \newline  $= \langle v , \Phi(x)^{\dagger -1} \Phi(x)^{-1} w \rangle_\sigma $  &   $h_L \Big( (xH_\nu , v) , (xH_\nu ,w) \Big)= \langle v , w \rangle_\sigma$ \\
\noalign{\smallskip}\hline\noalign{\smallskip}
Action &  $\lambda(nh) (yH_\nu , v)= $\newline$  \Big( h yH_\nu  , \nu \big( (hy)^{-1} n h y \big) \Phi (h) v \Big)$  &    $\lambda_{L} (nh) (yH_\nu , v)=  $ \newline $\Big(h yH_\nu , \nu \big( (hy)^{-1} nhy \big) \sigma \big( W_L(h, yH_\nu)\big) v \Big)$ \\
\noalign{\smallskip}\hline\noalign{\smallskip}
Space & $\mathcal{H}_{U} := L^2 \Big( H/H_\nu , E_\sigma ; \mu , h \Big)$ & $\mathcal{H}_{U_L} = L^2 ( H/H_\nu  ; \mu ) \otimes \mathcal{H}_\sigma$ \\
\noalign{\smallskip}\hline\noalign{\smallskip}
$\textup{Ind}_{G_\nu} ^G \nu \sigma$ & $ U(nh) \phi = \lambda(nh) \circ \phi \circ (l_h)^{-1}$ & $U_L (nh)\psi = \lambda_{L} (nh) \circ \psi \circ (l_h)^{-1}$ \\
\noalign{\smallskip}\hline
\end{tabular}

\end{table}
\end{theorem}

The Hermitian $G$-bundles $E_\sigma$ and $E_{L, \sigma}$ in Table~\ref{tab:1} will be called the \textit{perception bundle associated with $\sigma$} and the \textit{boosting bundle associated with $L$ and $ \sigma$}, respectively, for reasons that will become clear in Sect.~\ref{sec:6}. Accordingly, the representation spaces $\mathcal{H}_U$ and $\mathcal{H}_{U_L}$ in Table~\ref{tab:1} will be called the \textit{perception space} and the \textit{boosting space}, respectively.
\begin{proof}
Since the proof needs a long list of new definitions and lemmas, it has been exiled to \ref{sec:A}. Note that the action of $N$ on $H/H_\nu$ is trivial.
\end{proof}

As shown in Remark~\ref{remark:4.4}, all single-particle state spaces are of the form $\textup{Ind}_{G_\nu} ^G \nu \sigma$. Thus, we have just seen that the single-particle state spaces can be expressed in terms of induced representations associated with Hermitian $G$-bundles as defined in Definition~\ref{definition:5.3}. We have listed two relevant such descriptions in Table~\ref{tab:1}, comparisons of which lie at the heart of this paper.

\section{Bundle theoretic descriptions of massive particles}\label{sec:6}

In this section, we apply the mathematical framework developed in Sect.~\ref{sec:5} to massive particle state spaces listed in Theorem~\ref{theorem:4.10} and obtain bundle theoretic descriptions of massive particles with arbitrary spin, which was first suggested in \cite{lee2022} for spin-1/2 case (cf. Sect.~\ref{sec:3}). For the rest of the paper, $G$ will always denote the group $\mathbb{R}^4 \ltimes SL(2, \mathbb{C})$. Fix $m>0$ once and for all.

\paragraph{A $G$-invariant measure on the mass shell}

\hfill

First, it is necessary to identify a $G$-invariant measure on the orbit space $SL(2, \mathbb{C})/SU(2) \cong X_m ^\pm$ to apply the result of Sect.~\ref{sec:5}. Write $\omega_{\mathbf{p}} ^\pm : = \pm \sqrt {m^2 + |\mathbf{p}|^2}$. Then, the map $\mathbb{R}^3 \rightarrow X_m ^\pm $ given by
\begin{equation}\label{eq:6.1}
\mathbf{p} \mapsto (\omega_{\mathbf{p}} ^\pm , \mathbf{p})
\end{equation}
is a diffeomorphism, by which we always identify $\mathbb{R}^3$ with $X_m ^\pm$ and write $p = (\omega_{\mathbf{p}} ^\pm , \mathbf{p}) \in X_m ^\pm$. I.e., we set $p^0 = \omega_{\mathbf{p}} ^\pm$.

\begin{proposition}\label{proposition:6.1}
The following is a $G$-invariant measure on the orbit $X_m ^\pm \cong \mathbb{R}^3$.
\begin{equation}\label{eq:6.2}
    d\mu^\pm (p) \cong \frac{ d^3 \mathbf{p}}{ | \omega_{\mathbf{p}}  ^\pm |} =\frac{ d^3 \mathbf{p}}{ |p^0| }
\end{equation}
\end{proposition}
\begin{proof}
For a proof, see Ch.~1 of \cite{folland2008}.  
\end{proof}

From now until Sect.~\ref{sec:6.2}, we restrict our attention to the mass shell $X_m ^+$ and suppress all the $+$ superscripts throughout (e.g., $p_m ^+ = p_m$). The mass shell $X_m ^-$ will be taken up in Sect.~\ref{sec:7}.

\paragraph{The description table for massive particles}

\hfill

Note that $L : X_{m} \cong H/H_{p_m} \rightarrow H$ given by Eq.~(\ref{eq:2.19}) is a continuous global section. Let $s \in \frac{1}{2}\mathbb{N}_0$ and consider the irreducible representation $\sigma_{s} : SU(2) \rightarrow U(V_s)$ which extends to $\Phi_{s} : SL(2, \mathbb{C}) \rightarrow GL(V_s)$ (cf. Eqs.~(\ref{eq:4.17})--(\ref{eq:4.18})). Then, Theorem~\ref{theorem:5.5} gives us Table~\ref{tab:2} with an intertwining isomorphism

\begin{equation}\label{eq:6.3}
\begin{tikzcd}[baseline=(current  bounding  box.center), column sep=1.5em]
    \alpha_s : E_{s} \arrow ["{(p,v) \mapsto (p , \Phi_s \big(L(p) ^{-1} \big)v)}"]{rrrr} \arrow{drr} 
   & & & & E_{L, s} \arrow{dll}
\\
   & &X_m& &
\end{tikzcd}.
\end{equation}

Notice that $\pi_{m,s} := \textup{Ind}_{G_{p_m}} ^G p_m \sigma_s$ represents the single-particle of mass $m$ and spin $s$ (cf. Remark~\ref{remark:4.4} and Theorem~\ref{theorem:4.10}).

\begin{table}[h]
\caption{The perception and boosting bundles for a massive particle with spin-\texorpdfstring{$s$}{TEXT}}
\label{tab:2}
\centering
\begin{tabular}{|m{1.5cm}|m{5.2cm}|m{7cm}|}
\hline\noalign{\smallskip}
  &  $E_{s} \text{ (The perception bundle)}$ &  $E_{L , s} \text{ (The boosting bundle)}$ \\
\noalign{\smallskip}\hline\noalign{\smallskip}
Bundle  &  $X_m \times V_s$  &  $X_m \times V_s$ \\
\noalign{\smallskip}\hline\noalign{\smallskip}
Metric &   $h_s \Big( (p , v) , (p, w) \Big)$ \newline  $= v^\dagger \Phi_{s}(\frac{\utilde{p}}{m}) w $  &   $h_{L,s} \Big( (p , v) , (p ,w) \Big)= v^\dagger w $ \\
\noalign{\smallskip}\hline\noalign{\smallskip}
Action &  $\lambda_s (a,\Lambda) (p , v)= $\newline$  \Big( \Lambda p  , e^{-i \langle \Lambda p , a \rangle} \Phi_{s} (\Lambda) v \Big)$  &    $\lambda_{L , s} (a, \Lambda) (p , v)= $ \newline $ \Big(\Lambda p , e^{-i \langle \Lambda p , a \rangle} \sigma_s \big( W_L(\Lambda,p)\big) v \Big)$ \\
\noalign{\smallskip}\hline\noalign{\smallskip}
Space & $\mathcal{H}_s := L^2 \Big( X_m , E_{s} ; \mu , h \Big)$ & $\mathcal{H}_{L , s} := L^2 ( X_m  ; \mu ) \otimes V_s$ \\
\noalign{\smallskip}\hline\noalign{\smallskip}
$\pi_{m,s}$ & $ U_s (a, \Lambda) \phi = \lambda(a, \Lambda) \circ \phi \circ \Lambda^{-1}$ & $U_{L, s}(a, \Lambda)\psi = \lambda_{L, s} (a, \Lambda) \circ \psi \circ \Lambda^{-1}$ \\
\noalign{\smallskip}\hline
\end{tabular}
\end{table}

All the formulae listed in Table~\ref{tab:2} are straightforwardly computed from the definitions except the one for $h_s$. To obtain it, observe that if $p = \Lambda p_m \in X_m$, then since $\Phi_s$ preserves the adjoints (cf. Eqs.~(\ref{eq:4.14})--(\ref{eq:4.18})), we have

\begin{equation*}
\Phi_s(\Lambda)^{\dagger -1} \Phi_s (\Lambda)^{-1} = \Phi_s (\Lambda^{ \dagger -1} \Lambda) = \Phi_s \left(\Lambda^{\dagger -1} \left(\frac{1}{m} (p_m)_\sim \right) \Lambda^{-1} \right) = \Phi_s (\frac{\utilde{p}}{m})
\end{equation*}
by Eq.~(\ref{eq:2.13}). Notice that
\begin{equation}\label{eq:6.4}
W_L (\Lambda,p) := L(\Lambda p)^{-1} \Lambda L(p) \in SU(2)
\end{equation}
(cf. Eq.~(\ref{eq:5.7})) is indeed the \textit{Wigner rotation matrix} used in the physics literature (cf. \cite{weinberg}).

\subsection{The vector bundle point of view for massive particles}\label{sec:6.1}

In \cite{lee2022}, it was suggested that expressing some problems of RQI in terms of Hermitian $G$-bundles has several advantages. In this picture, the bundles $E_{s}$ and $E_{L , s}$ are assemblies of the $d= (2s +1)$-level quantum systems $(E_s)_p$ and $(E_{L,s})_p$ corresponding to each motion state (momentum) $p \in X_m$, and each wave function $\psi \in \mathcal{H}_s \hspace{0.2cm} \textup{or} \hspace{0.2cm} \mathcal{H}_{L,s}$ becomes a field of qudits. The so-called momentum-spin eigenstate $|p, \chi \rangle, (p \in X_m, \chi \in V_s)$ used in the physics literature can be identified with the point $ (p, \chi ) \in (E_{L,s})_p$ in this formalism.\footnote{However, each point in $E_{1/2}$ corresponds to a "relativistic chiral qubit" introduced in \cite{caban2019}.}

Since the single-particle state space for massive particle with spin-$s$ can be constructed from the bundles $E_s$ and $E_{L,s}$ according to Table~\ref{tab:2}, each inertial observer can use the bundles $E_{s}$ and $E_{L,s}$ instead of $\pi_{m,s}$ for the description of a massive particle with spin-$s$ in the sense that the full information of each quantum state that the particle can assume (which is an $L^2$-section of the bundles) can be recorded in the bundle.\footnote{This mathematical fact has nothing to do with physical measurement.\label{footnote:15}} How are these bundle descriptions related among different inertial observers? Suppose two inertial observers, Alice and Bob, are related by a Lorentz transformation $(a, \Lambda) \in G$ as in Eq.~(\ref{eq:2.11}). If Alice has prepared a particle in the state $\mathcal{H}_s \ni \phi \stackrel{\alpha_s}{=} \psi \in \mathcal{H}_{L,s}$ (cf. Eq.~(\ref{eq:6.3})) in her frame, then Bob would perceive this particle as in the state $\mathcal{H}_s \ni U_{s } (a, \Lambda) \phi \stackrel{\alpha_s}{=} U_{L,s} (a, \Lambda) \psi \in \mathcal{H}_{L,s}$ according to Sect.~\ref{sec:2.3} (cf. Remark~\ref{remark:2.11}).

For these transformation laws for wave functions to be true, Alice's bundles $E_s ^A$, $E_{L,s} ^A$ and Bob's bundles $E_s ^B$, $E_{L,s} ^B$ should be related by the Hermitian $G$-bundle isomorphisms

\begin{gather}
\lambda_s (a, \Lambda ) : E_s ^A \rightarrow E_s ^B \nonumber \\
(p,v)^A \mapsto \left(  \Lambda p, e^{- i (\Lambda p)_\mu a^\mu} \Phi_{s} \left(\Lambda \right) v \right)^B \label{eq:6.5}
\end{gather}
and
\begin{gather}
\lambda_{L,s} (a, \Lambda ) : E_{L,s} ^A \rightarrow E_{L,s} ^B \nonumber \\
(p,v)^A \mapsto \left(  \Lambda p, e^{- i (\Lambda p)_\mu a^\mu} \sigma_{s} \left(W_L ( \Lambda, p) \right) v \right)^B \label{eq:6.6},
\end{gather}
respectively, so that the transformation laws for the sections
\begin{equation}\label{eq:6.7}
\phi^A  \mapsto \phi^B = \lambda_s ( a, \Lambda ) \circ \phi^A \circ \Lambda^{-1}
\end{equation}
and
\begin{equation}\label{eq:6.8}
\psi^A \mapsto \psi^B = \lambda_{L,s} (a, \Lambda) \circ \psi^A \circ \Lambda^{-1}
\end{equation}
become $U_s (a, \Lambda)$ and $U_{L,s} (a, \Lambda)$, respectively. The following commutative diagrams are useful in visualizing the transformation laws.
\begin{equation}\label{eq:6.9}
\begin{tikzcd}[baseline=(current  bounding  box.center)]
E_s ^A \arrow[r, "{\lambda_s ( \Lambda, a)}"] \arrow[d ]
& E_s ^B \arrow[d] & & E_{L,s} ^A \arrow[r, "{\lambda_{L,s} ( \Lambda, a)}"] \arrow[d ] & E_{L,s} ^B \arrow{d}  \\
X_m ^A  \arrow[r, "\Lambda"]
& X_m ^B & & X_m ^A \arrow[r, "\Lambda"] & X_m ^B
\end{tikzcd}.
\end{equation}

Note that Eqs.~(\ref{eq:6.5}) and (\ref{eq:6.6}) are just the $G$-actions listed in Table~\ref{tab:2}.

\begin{remark}\label{remark:6.2}
This vector bundle viewpoint is similar to the setting up of a coordinate system in classical SR. In fact, as we see from the diagrams Eq.~(\ref{eq:6.9}), it is nothing more than a momentum coordinate system with the particle's internal quantum systems (corresponding to each possible motion state) taken into account, whose transformation law is governed by the Hermitian $G$-actions given in Table~\ref{tab:2}. Note that the two quantum transformation laws between inertial observers (Eqs.~(\ref{eq:6.5}) and (\ref{eq:6.6})) are extensions of the classical transformation law for the momentum observation expressed by the base space transformation $X_m \xrightarrow{\Lambda} X_m$.

Fix an inertial observer who is interested in describing a massive quantum particle with spin-$s$. The full information of the quantum states of the particle as perceived by the observer can be recorded either in the perception bundle $E_s$ or in the boosting bundle $E_{L,s}$ (see footnote~\ref{footnote:15}). The recurring theme of the analyses given in this paper is that while the fibers of the perception bundle $E_s$ reflect correctly the perception (in the sense of Sect.~\ref{sec:2.5}) of the fixed observer (hence the name), each fiber $(E_{L,s})_p$ of the boosting bundle $E_{L,s}$ is rather the perception of an $L(p)^{-1}$-boosted observer (with respect to the fixed one) for each $p \in X_m$ (hence the name) and hence is not directly accessible from the fixed observer, just as in the case $s=1/2$ which we observed in Sect.~\ref{sec:3}. So, in the context of RQI, the description provided by the bundle $E_s$ is conceptually more appropriate than $E_{L,s}$.

Theses interpretations of the two bundles will be in a sense proved in Sect.~\ref{sec:7}, where we will see that the Dirac equation and the Proca equations, which are fundamental equations of QFT obeyed by massive particles with spin-1/2 and 1, respectively, emerge as the defining equations of the respective perception bundles. Therefore, one may say that \textit{the Dirac equations and the Proca equations are nothing but manifestations of a fixed inertial observer's perception of the internal quantum states of massive particles with spin-1/2 and 1, respectively.}
\end{remark}

\subsection{The perception and boosting bundle descriptions for massive particles with arbitrary spin}\label{sec:6.2}

Let's take $\pi_{m,s} = \pi_{m,s} ^+$ where $s \in \frac{1}{2} \mathbb{N}$ and see if the discussions for the $s=1/2$ case (cf. Sect.~\ref{sec:3.3}) carry over to this case as well.

\subsubsection{A relation between the two descriptions; the bundles}\label{sec:6.2.1}

\hfill

In this case, it is not obvious how to find a relationship between the two descriptions since, for higher $s$, qudits in a $d = (2s+1)$-level quantum system are not characterized by three-vectors, unlike the $s=1/2$ case.

Therefore, instead of considering a general qudit $\chi \in V_s$, we argue as follows. Suppose Alice has prepared a qudit which is a spin eigenstate with eigenvalue $k-s, \hspace{0.2cm} (0 \leq k \leq 2s)$ along the $\hat{z}$-axis in her rest frame. According to Theorem~\ref{theorem:4.7}, this means that she has picked the state $\sqrt{\frac{ (2s)!} {k! (2s - k )! }} u^k v^{2s - k}$ from the basis Eq.~(\ref{eq:4.15}). From the $L (p)$-transformed Bob's frame, the qudit should be a spin eigenstate along the $L (p) \hat{z}$-axis with eigenvalue $k-s$. But, what are spin eigenstates along the $L (p) \hat{z}$-axis, a four-vector direction?

According to the discussion of Sects.~\ref{sec:3.2}--\ref{sec:3.3}, on the two-level system case, the vectors $L (p) u$ and $L (p) v$ might be called the spin eigenstates along the $L (p) \hat{z}$-axis with the eigenvalues $1/2$ and $-1/2$ respectively. Let $w = L(p) \hat{z} \in \mathbb{R}^4$. Observe that the traceless operator
\begin{equation}\label{eq:6.10}
\tilde{w} \frac{ \utilde{p}}{m} = L (p) \left(\frac{1}{2} \tau^3 \right) L (p)^{-1} \in \mathfrak{sl}(2, \mathbb{C})
\end{equation}
(cf. Eqs.~(\ref{eq:2.13}) and (\ref{eq:2.19})) is an Hermitian operator in the fiber $\big((E_{1/2} )_p , (h_{1/2})_p \big)$ and $L (p) u$, $L (p) v$ are two eigenvectors of this operator with eigenvalues $1/2$, $-1/2$, respectively. So, we see that Eq.~(\ref{eq:6.10}) is the observable for the spin along the $L (p) \hat{z}$-axis. (In fact, as can be seen from Eq.~(\ref{eq:7.26}), it is the third component of the Newton-Wigner spin operator restricted to the fiber $(E_{1/2})_p$.)

Generalizing this to the $(2s+1)$-level quantum system, we see that the operator
\begin{equation}\label{eq:6.11}
 (\Phi_s)_* (  \tilde{w} \frac{ \utilde{p}}{m}) = i (\sigma_s )_* (-i \tilde{w} \frac{ \utilde{p}}{m}),
\end{equation}
which is Hermitian on the fiber $\big((E_s )_p, (h_s)_p \big)$, is the spin observable along the $L (p) \hat{z}$-axis on this system. (Actually, this is the third component of the Newton-Wigner spin operator restricted to the fiber $(E_s)_p$. See the remark of Sect.~\ref{sec:6.2.2}.) Since the $L(p)$-transformed Bob should perceive the qudit prepared by Alice as a spin eigenstate with eigenvalue $k-s$ along the $L (p) \hat{z}$-axis, the qudit as perceived from Bob's frame should be an eigenstate of the operator $(\Phi_s)_* (\tilde{w} \frac{ \utilde{p}}{m})$ whose eigenvalue is $k-s$. By Eq.~(\ref{eq:4.21}), we see that this is precisely
\begin{equation}\label{eq:6.12}
\sqrt{\frac{ (2s)!}{ k! (2s -k)!}} (L (p) u )^k (L (p) v)^{2s-k} = \Phi_s (L (p) ) \left(  \sqrt{\frac{ (2s)!}{ k! (2s -k)!}} u^k v^{2s -k} \right).
\end{equation}

Since this holds for all the basis elements listed in Eq.~(\ref{eq:4.15}), we conclude that Bob's perception (in the sense of Sect.~\ref{sec:2.5}) of the qudit $(p_m, \chi)^A \in (E_{s})_{p_m} = (E_{L,s})_{p_m} $, which is prepared in Alice's rest frame, is
\begin{equation}\label{eq:6.13}
\Big(p, \Phi_s (L (p) ) \chi \Big)^B,
\end{equation}
which is precisely captured by the transformation law Eq.~(\ref{eq:6.5}). Hence as remarked in Sect.~\ref{sec:3.3}, the qudits in the bundle $E_s$ are "relativistic perception" of a fixed inertial observer. I.e., the fibers of the perception bundle $E_s$ correctly reflect the perception of the fixed inertial observer (hence the name).

Also as a consequence of this fact, the equation
\begin{equation}\label{eq:6.14}
\lambda_{L,s} (0, L(p)) (p_m , \chi)^A = (p, \chi)^B,
\end{equation}
which follows from the transformation law Eq.~(\ref{eq:6.6}), tells us that the qudits in $(E_{L,s})_p$ don't reflect the perception of the fixed inertial observer in whose frame the qudit-carrying particle is moving with momentum $p$. Rather, they are the perception of an $L(p)^{-1}$-boosted observer (hence the name).

We conclude that the interpretations and relations given in Sect.~\ref{sec:3.3} about the two descriptions for the spin-1/2 case hold in full generality, i.e., for all possible values of spin.

\subsubsection{A relation between the two descriptions; the representations}\label{sec:6.2.2}

\hfill

We need to check whether the description $(E_s,\mathcal{H}_s)$ is related to the Pauli-Lubansky four-vector and $(E_{L,s}, \mathcal{H}_{L,s} )$ is related to the Newton-Wigner spin in relation to the former just as in Sect.~\ref{sec:3.3}. On the bundle levels, this fact is not so obvious since the qudits in a higher-level system are in general not characterized by vectors in $\mathbb{R}^3$. But, moving into the level of wave functions and operators, we can obtain an analogous relation (see Sect.~\ref{sec:3.3.3}).

For the discussion on the level of Hilbert spaces and operators, we define the \textit{Pauli-Lubansky operators} which are elements of the universal enveloping algebra of $\mathfrak{g}_\mathbb{C}$ (here, $\mathfrak{g}_\mathbb{C}$ is the complexification of the Lie algebra $\mathfrak{g}$ of $G= \mathbb{R}^4 \ltimes SL(2, \mathbb{C})$) by
\begin{equation}\label{eq:6.15}
W^\mu = \frac{1}{2} \varepsilon^{\nu \alpha \beta \mu} P_\nu J_{\alpha \beta}
\end{equation}
where the relativistic angular momentum operators $J_{\alpha \beta} = -J_{\beta \alpha}$ are defined as $J_{23} = J^1$, $J_{31} = J^2$ , $J_{12} = J^3$ and $J_{j0}= K^j$, respectively (cf. Eq.~(\ref{eq:4.19})), and hence
\begin{subequations}\label{eq:6.16}
\begin{align}
W^0 &= \mathbf{P} \cdot \mathbf{J} \label{eq:6.16a} \\
\mathbf{W} &= P^0 \mathbf{J} - \mathbf{P} \times \mathbf{K}. \label{eq:6.16b}
\end{align}
\end{subequations}

Then, the \textit{Newton-Wigner spin operator} is defined as
\begin{equation}\label{eq:6.17}
\mathbf{S}_{NW} := \frac{1}{m} \left( \mathbf{W} - \frac{W^0 \mathbf{P} }{ m + P^0 } \right)
\end{equation}
which is also an element of the universal enveloping algebra of $\mathfrak{g}_\mathbb{C}$, which will then become operators on any quantum system with Lorentz symmetry via the given representation of the group $G$ (cf. Definition~\ref{definition:2.10}).

In \cite{lee2022}, it is proved\footnote{The paper only deals with the $s=1/2$ case. However, if we replace $\tau^j$ by $(\Phi_s)_* (\tau^j)$ in the proof given there, then we obtain Eq.~(\ref{eq:6.19}).} that the operator
\begin{equation}\label{eq:6.18}
\mathbf{S} := 1 \otimes (\Phi_s)_* ( \frac{1}{2} \boldsymbol{\tau}) = 1 \otimes i (\sigma_s )_* ( \mathbf{J}) = 1 \otimes \hat{\mathbf{J}}
\end{equation}
on the Hilbert space $\mathcal{H}_{L,s}$ becomes the Newton-Wigner spin operator on the Hilbert space $\mathcal{H}_s$ (cf. Eq.~(\ref{eq:6.17})), i.e.,
\begin{equation}\label{eq:6.19}
(U_s)_* (\mathbf{S}_{NW}) = \alpha_s ^{-1} \circ \mathbf{S} \circ \alpha_s
\end{equation}
where in general, given a unitary representation $\pi$ of a Lie group, $\pi_*$ denotes the induced $*$-representation of the universal enveloping algebra of the comlexified Lie algebra of the Lie group (cf. Ch.~0 of \cite{taylor}).

So, the $\mathbb{C}^{2s+1}$-component of the space $\mathcal{H}_{L,s}$ has the meaning of the Newton-Wigner spin in the perception space description $\mathcal{H}_{s}$. We conclude that for $s=1/2$, the relation between the two bundles (cf. Sect.~\ref{sec:3.3.2}) also holds on the level of Hilbert spaces and operators (see Sect.~\ref{sec:3.3.3}).

\subsubsection{The spin and Pauli-Lubansky reduced density matrices}\label{sec:6.2.3}

\hfill

Let $\psi \in \mathcal{H}_{L,s}$ be a state and $\rho := | \psi \rangle \langle \psi |$ be the density matrix corresponding to $\psi$. Just as in Sect.~\ref{sec:3.1}, we form the spin reduced density matrix for $\psi$ by
\begin{equation}\label{eq:6.20}
\tau = \text{Tr}_{L^2} \hspace{0.1cm} (\rho) ,
\end{equation}
which is a $(2s+1) \times (2s+1)$ density matrix.

We saw in Sect.~\ref{sec:3.2} that this matrix for the $s= 1/2$ case has no meaning at all. This was because, since the fibers of the bundle $E_{L,1/2}$ do not reflect the perception of a fixed inertial observer who is taking the partial trace, Eq.~(\ref{eq:6.20}) becomes a summation over the objects living in a whole lot of different reference frames, which is an absurdity unless the objects are first pulled back to the fixed inertial frame before the summation takes place.

In Sect.~\ref{sec:6.2.1}, we saw that the same problem resides in the general spin case as well. I.e., the fibers of the boosting bundle $E_{L,s}$ for general $s$ also do not reflect the perception of a fixed inertial observer who is using this bundle. Therefore, in particular, given a state $\psi = f \chi \in \mathcal{H}_{L,s}$ defined analogously as in Sect.~\ref{sec:3.2}, each qudit state $\chi(p) \chi(p)^{\dagger}$ gets meaningful only in an $L (p)^{-1}$-transformed inertial observer (cf. Eq.~(\ref{eq:6.14}) and the remark following it). So, we conclude that the spin reduced density matrix Eq.~(\ref{eq:6.20}), which is expressed as
\begin{equation}\label{eq:6.21}
\tau:= \int_{X_m} \psi (p) \psi(p)^\dagger d \mu (p) = \int_{X_m} |f(p)|^2 \chi(p) \chi(p)^\dagger d \mu (p)
\end{equation}
is meaningless either for all values of spin $0 \neq s \in \frac{1}{2} \mathbb{Z}$.

Even though it is of some interest to see whether the phenomenon observed in \cite{peres2002} (cf. Sect.~\ref{sec:3.1}) is still present in the general spin case by carrying out an analytic computation of Eq.~(\ref{eq:6.20}) using the formalism presented in this paper, we will not pursue that direction any further since we have just seen that Eq.~(\ref{eq:6.20}) is meaningless.

In \cite{lee2022}, moreover, it was suggested that the only way to modify Eq.~(\ref{eq:6.20}) in order for it to attain a substance is by first pulling back the integrand states $\chi (p) \chi(p)^\dagger$ to the fixed inertial frame and then carry out the integration. In Sect.~\ref{sec:6.2.1}, we saw that the pulled-back integrands are precisely $\Phi_s (L (p)) \chi (p) \chi(p)^\dagger \Phi_s( L (p))$ (cf. Eq.~(\ref{eq:6.13})). So, the modified reduced matrix is

\begin{align}\label{eq:6.22}
\sigma := \int_{X_m} |f(p)|^2 \Phi_s (L (p) ) \chi(p) \chi(p)^\dagger \Phi_s (L (p)) d\mu(p) \nonumber \\ = \int_{X_m} \Phi_s(L (p)) \psi(p) \psi(p)^\dagger \Phi_s(L (p)) \nonumber \\
=\int_{X_m} [\alpha^{-1} \psi ] (p) [ \alpha^{-1} \psi ] (p) d \mu(p),
\end{align}
which is just the operation
\begin{equation}\label{eq:6.23}
\sigma (\phi) = \int_{X_m} \phi (p) \phi(p)^\dagger d \mu (p)
\end{equation}
applied to $\alpha^{-1} \psi \in \mathcal{H}_s$. One should note that this operation cannot be defined for all elements in $\mathcal{H}_s$ due to the non-trivial Hermitian metric $h_s$ (cf. Table~\ref{tab:2}). However, this operation is well-defined at least on the Schwartz section space $\{ \phi \in \mathcal{H}_s : \phi_i \text{ is of Schwartz class for } i= 1, \cdots , 2s +1  \} \leq \mathcal{H}_s$ where $\phi_i$ is a component function of the section $\phi$ of the trivial bundle $E_s = X_m \times \mathbb{C}^{2s +1}$.

In \cite{lee2022}, Eq.~(\ref{eq:6.23}) was called the \textit{Pauli-Lubansky reduced matrix} since it has information about the average Pauli-Lubansky four-vector plus the average momentum in the $s=1/2$ case. It is very important to notice that Eq.~(\ref{eq:6.23}) is not a partial trace operation since, after all, $\mathcal{H}_s$ is not a tensor product system and second, $\phi(p) \phi(p)^\dagger$ is not the state corresponding to the qudit $\phi(p)$ due to the form of the inner product $h_s$ (cf. Table~(\ref{tab:2})).

Nevertheless, Eq.~(\ref{eq:6.23}) has some desirable features. It is positive and has a nonzero trace:
\begin{equation}\label{eq:6.24}
u^\dagger \sigma u = \int_X d \mu(p) \hspace{0.1cm} u^\dagger \phi(p) \phi(p)^\dagger u \geq 0 \quad \forall u \in \mathbb{C}^2
\end{equation}
and
\begin{equation}\label{eq:6.25}
\text{Tr} \hspace{0.1cm} \sigma = \int_X d \mu (p) \text{Tr} \rho = \int_X d \mu(p) \left( |\phi_1 (p)|^2 + |\phi_2 (p)|^2 \right) > 0.
\end{equation}

So, the matrix $\sigma$ can be normalized to yield a density matrix. Let's find its transformation law under a change of reference frame.

Suppose Alice has prepared a state $\phi \in \mathcal{H}_s$ and formed the matrix $\sigma_A = \sigma (\phi)$. Consider another observer Bob, in whose frame the state is $\pi_{m,s} ( a, \Lambda ) \phi \in \mathcal{H}_s$. Then, according to the transformation law for the perception bundle description (cf. Table~\ref{tab:2}), we have
\begin{equation}\label{eq:6.26}
\sigma_B = \int_X d\mu (p) \hspace{0.1cm} \Phi_s (\Lambda) \phi ( \Lambda^{-1} p) \phi( \Lambda^{-1} p)^\dagger \Phi_s(\Lambda ^{\dagger} ) = \Phi_s(\Lambda) \hspace{0.05cm} \sigma_A \hspace{0.05cm} \Phi_s (\Lambda^{\dagger} ).
\end{equation}

So, we see that Eq.~(\ref{eq:6.23}) is Lorentz covariant and hence has a relativistically invariant meaning, which may be interpreted as the average internal quantum state of the single-particle state $\phi \in \mathcal{H}_s$ as perceived by a fixed inertial observer. Investigating operational aspects of this matrix is beyond the scope of this paper. So, we leave it to researchers who are interested in exploring it.

\section{Theoretical implications}\label{sec:7}

In this section, we explore some of the theoretical implications of the perception bundle description. Specifically, we will see that the Dirac equation and the Proca equations (cf. \cite{folland2008}) are manifestations of a fixed inertial observer's perception of the internal quantum states of massive particles with spin-1/2 and 1, respectively.

\subsection{A modified framework}\label{sec:7.1}

Besides the single-particle state spaces dealt with in Sect.~\ref{sec:6}, there are some special forms of single-particle state spaces to which the formalism of Sect.~\ref{sec:5} cannot be directly applied. These include the Dirac bispinor representation for massive particles with spin-1/2 and the Minkowski space representation of massive particles with spin-1 (see below). To accommodate these into our formalism, we need the following modified version of Theorem~\ref{theorem:5.5}.

\begin{theorem}\label{theorem:7.1}
Let $G = N \ltimes H$ and fix $\nu \in \hat{N}$. Suppose $\sigma$ is a unitary representation of $H_\nu$ on the Hilbert space $(\mathcal{H}_\sigma , \langle \cdot , \cdot \rangle_\sigma)$ that extends to a representation $\Phi : H \rightarrow GL( \mathcal{K}_\Phi)$ where $\mathcal{K}_\Phi$ contains $\mathcal{H}_\sigma$ as a closed subspace, $L : H/H_\nu \rightarrow H$ is a global section, and there is an $H$-invariant measure $\mu$ on $H/H_\nu$.

Consider the two Hermitian $G$-bundles in Table~\ref{tab:1} and their associated induced representations $U$ and $U_L$ where the bundle $R_\sigma$ is given by the range of the bundle embedding
\begin{equation}\label{eq:7.1}
\begin{tikzcd}[baseline=(current  bounding  box.center), column sep=1.5em]
    \alpha^{-1} : E_{L, \sigma} \arrow ["{(xH_\nu,w) \mapsto (xH_\nu , \Phi \big(L(xH_\nu) \big)w)}"]{rrrrrr} \arrow{drrr} 
 & & & & & & H/H_\nu \times \mathcal{K}_\Phi \arrow{dlll}
\\
  & & &H/H_\nu& & &
\end{tikzcd}.
\end{equation}

Then, this map is an Hermitian $G$-bundle isomorphism that intertwines the structures listed in Table~\ref{tab:3} and

\begin{equation}\label{eq:7.2}
\textup{Ind}_{G_\nu} ^G \nu \sigma \cong U \cong U_L .
\end{equation}

\begin{table}[h]
\caption{The structures of the perception bundle and boosting bundle}
\label{tab:3}

\centering
\begin{tabular}{|m{1.5cm}|m{5.2cm}|m{7cm}|}
\hline\noalign{\smallskip}
  & $E_\sigma$  (The perception bundle) & $E_{L, \sigma}$  (The boosting bundle) \\
\noalign{\smallskip}\hline\noalign{\smallskip}
Bundle  &  $R_\sigma \leq H/H_\nu \times \mathcal{K}_\sigma$  &  $H/H_\nu \times \mathcal{H}_\sigma$ \\
\noalign{\smallskip}\hline\noalign{\smallskip}
Metric &   $h \Big( (xH_\nu , v) , (xH_\nu, w) \Big)$ \newline  $= \langle v , \Phi(x)^{\dagger -1} \Phi(x)^{-1} w \rangle_\sigma $  &   $h_L \Big( (xH_\nu , v) , (xH_\nu ,w) \Big)= \langle v , w \rangle_\sigma$ \\
\noalign{\smallskip}\hline\noalign{\smallskip}
Action &  $\lambda(nh) (yH_\nu , v)= $\newline$  \Big( h yH_\nu  , \nu \big( (hy)^{-1} n h y \big) \Phi (h) v \Big)$  &    $\lambda_{L} (nh) (yH_\nu , v)=  $ \newline $\Big(h yH_\nu , \nu \big( (hy)^{-1} nhy \big) \sigma \big( W(h, yH_\nu)\big) v \Big)$ \\
\noalign{\smallskip}\hline\noalign{\smallskip}
Space & $\mathcal{H}_{U} := L^2 \Big( H/H_\nu , E_\sigma ; \mu , h \Big)$ & $\mathcal{H}_{U_L} = L^2 ( H/H_\nu  ; \mu ) \otimes \mathcal{H}_\sigma$ \\
\noalign{\smallskip}\hline\noalign{\smallskip}
$\textup{Ind}_{G_\nu} ^G \nu \sigma$ & $ U(nh) \phi = \lambda(nh) \circ \phi \circ (l_h)^{-1}$ & $U_L (nh)\psi = \lambda_{L} (nh) \circ \psi \circ (l_h)^{-1}$ \\
\noalign{\smallskip}\hline
\end{tabular}
\end{table}
\end{theorem}

Analogously as in Sect.~\ref{sec:5}, we call Hermitian $G$-bundles $E_{\sigma}$ obtained in this way \textit{perception bundles}. Although we didn't need to modify the boosting bundle description, we have written it here for the sake of comparison.

\begin{proof}
The proof is a straightforward adaptation of the proof of Theorem~\ref{theorem:5.5}. One only needs to be careful of the fact that the perception bundle might be a proper subbundle of the trivial bundle $H/H_\nu \times \mathcal{K}_\Phi$. See \ref{sec:A}.
\end{proof}

\subsection{The Dirac bispinor representation of massive particles with spin-1/2}\label{sec:7.2}

In addition to the description given in Sect.~\ref{sec:3}, there is an equivalent way to describe massive particles with spin-1/2 called the Dirac bispinor representation. In the QFT literature this representation is of paramount importance (cf. \cite{weinberg}). In the context of RQI, this representation has been investigated, for example in \cite{bittencourt2018, bittencourt2019}. Particles/antiparticles described by this representation will be referred to as \textit{Dirac particles}.

Since antiparticle states will also be relevant to the discussion, we need to consider the representations associated with the two mass shells $X_m ^\pm$ simultaneously (cf. Proposition~\ref{proposition:4.3}). One must be careful in keeping track of the superscripts $\pm$. Note that $H_{p_m ^\pm} = SU(2)$ (cf. Proposition~\ref{proposition:4.6}).

Since we are dealing with two mass shells, we must have two choices of boostings (cf. Sect.~\ref{sec:2.4}). We choose
\begin{equation}\label{eq:7.3}
L^\pm (p) = \sqrt{ \frac{ \pm \tilde{p}} {m}} = \frac{1}{\sqrt{2m (m \pm p^0)}} ( \pm \tilde{p} + m I_2 ) \in SL(2, \mathbb{C}), \quad p \in X_m ^\pm.
\end{equation}

It is easy to check that $\kappa \big( L^\pm (p) \big) p_m ^\pm = p$ for $ p \in X_m ^\pm$. A remark similar to Remark~\ref{remark:2.12} also holds for $L^- (p)$ with obvious modifications. To avoid clutter, we will suppress the superscripts $\pm$ from $L^\pm$ when $L^\pm$ appears as a subscript for an object related to the boosting bundle description.

\paragraph{The perception bundle for Dirac particles}

\hfill

Instead of choosing $\sigma_{1/2}$ as in Sect.~\ref{sec:3}, we choose $\sigma_{1/2} \oplus \sigma_{1/2} : SU(2) \rightarrow GL(4, \mathbb{C})$ and its (non-unitary) extension $\Phi : SL(2, \mathbb{C}) \rightarrow GL(4, \mathbb{C})$ given by
\begin{equation}\label{eq:7.4}
\Phi(\Lambda) = \begin{pmatrix} \Lambda & 0 \\ 0 & \Lambda^{\dagger -1} \end{pmatrix}.
\end{equation}

Observe that the subspaces $V^\pm := \left\{ z \in \mathbb{C}^4 : (z_1 , z_2) = \pm (z_3 , z_4) \right\} \leq \mathbb{C}^4$ are 2-dimensional orthogonal invariant spaces with respect to $\sigma_{1/2} \oplus \sigma_{1/2}$. We write its corresponding subrepresentations as $\sigma^\pm$, i.e., we set $\sigma^\pm (\Lambda) := [\sigma_{1/2} \oplus \sigma_{1/2}] (\Lambda)|_{V^\pm}$ for $\Lambda \in SU(2)$.

The maps $u^\pm : \mathbb{C}^2 \rightarrow V^\pm$ given by $v \mapsto \frac{1}{\sqrt{2}}(v, \pm v)$ are unitary maps intertwining $\sigma_{1/2}$ and $\sigma^\pm$, respectively. So, we see $\sigma^\pm \cong \sigma_{1/2}$ and thus $\sigma^\pm$ are irreducible. With the understanding that $\mathcal{K}_\Phi = \mathbb{C}^4$, we apply Theorem~\ref{theorem:7.1}.

\begin{proposition}\label{proposition:7.4}
The range bundles $R^\pm \leq X_m^\pm \times \mathbb{C}^4$ of Table~\ref{tab:3} for Dirac particles are given by
\begin{equation}\label{eq:7.5}
R ^\pm = \left \{ (p, z ) \in X_m ^\pm \times \mathbb{C}^4 : p_\mu \gamma^\mu z = mz \right\} 
\end{equation}
where $\gamma^\mu := \begin{pmatrix} 0 & \tau_\mu \\ \tau^\mu & 0 \end{pmatrix}$ is the Weyl representation of the Dirac matrices (cf. \cite{folland2008}).

The Hermitian metrics $h^\pm$ on $R^\pm$ provided by Table~\ref{tab:3} become
\begin{equation}\label{eq:7.6}
h_p ^\pm (v, w ) = v^\dagger \Phi ( \pm \frac{ \utilde{p}} {m} ) w = \pm v^\dagger \gamma^0 w = \frac{m}{| p^0 |} v^\dagger w
\end{equation}
for $v, w \in \left( R ^\pm \right)_p$.
\end{proposition}
\begin{proof}
Throughout the proof, let's denote $\Lambda := L^\pm (p)$ for ${}^\forall p \in X_m ^\pm$ for simplicity of notation. Observe that for $p \in X_m ^\pm$, and $z \in V^\pm$,
\begin{align*}
p_\mu \gamma^\mu \Phi(\Lambda) z = \begin{pmatrix} 0 & \tilde{p} \\  \utilde{p} &0 \end{pmatrix} \Phi(\Lambda) z  = \begin{pmatrix} 0 & \Lambda (p_m ^\pm)^\sim \Lambda^{\dagger} \\ \Lambda^{\dagger -1} ( p_m ^\pm )_\sim \Lambda^{-1} & 0 \end{pmatrix}  \Phi(\Lambda) z \\
= \pm \Phi(\Lambda)  m \gamma^0 \Phi(\Lambda)^{-1} \Phi(\Lambda) z = m \Phi(\Lambda) (\pm  \gamma^0 z) = m \Phi(\Lambda) z
\end{align*}
and hence indeed the range of the bundle maps given by Eq.~(\ref{eq:7.1}) are contained in the subbundles $R^\pm \leq X_m ^\pm \times \mathbb{C}^4 $.

Since $\left( p_\mu \gamma^\mu \right)^2 = m^2 I_4 $ (cf. Eq.~(\ref{eq:2.4b})), the map $p_\mu \gamma^\mu : \mathbb{C}^4 \rightarrow \mathbb{C}^4$ decomposes $\mathbb{C}^4$ into two eigenspaces corresponding to the eigenvalues $\pm m$. Since $\text{det} \left( p_\mu \gamma^\mu \right) = (p^2)^2 = m^4$ and $p_\mu \gamma^\mu \neq \pm m I_4$, we see the multiplicities of $\pm m$ are both 2. So, $R^\pm$ are subbundles of the trivial bundle $X_m ^\pm \times \mathbb{C}^4$ with rank 2. Since the maps Eq.~(\ref{eq:7.1}) are injective at each fiber of the boosting bundle $X_m ^\pm \times V^\pm$, this implies that the range bundles are all of $R^\pm$.

The first equality of Eq.~(\ref{eq:7.6}) is proved by the same calculation presented right below Table~\ref{tab:2}. Observe
\begin{equation*}
\Phi( \frac{\utilde{p}}{m}) = \begin{pmatrix} \frac{\utilde{p}}{m} & 0 \\ 0 & \frac{\tilde{p}}{m} \end{pmatrix} = \frac{1}{m} \gamma^0 p_\mu \gamma^\mu.
\end{equation*}

Since $p_\mu \gamma^\mu = m I_4$ on $(R^\pm )_p$ for each $p \in X_m ^\pm$, we see the second equality in Eq.~(\ref{eq:7.6}) holds. A direct computation would show that $v^\dagger \gamma^j v = 0$ for $ v \in V^\pm$ and $\gamma^j$ for $j=1,2,3$. Now, observe
\begin{equation*}
v^\dagger \gamma^0 v = \frac{1}{p^0} v^\dagger (p_0 \gamma^0) v = \frac{1}{p^0} v^\dagger (p_\mu \gamma^\mu) v = \frac{m}{p^0} v^\dagger v
\end{equation*}
and use the polarization identity to see that the third identity also holds.
\end{proof}

\paragraph{The description table for Dirac particles}

\hfill

Table~\ref{tab:4} below, which is a consequence of Theorem~\ref{theorem:7.1} and Proposition~\ref{proposition:7.4}, is the description table for Dirac particles. Notice that since $\sigma^\pm \cong \sigma_{1/2}$, we have $\textup{Ind}_{G_{p_m ^\pm}} ^G (e^{-i \langle p_m ^\pm , \cdot \rangle} \sigma^\pm) \cong \pi_{m, 1/2} ^\pm$, which represent particles/antiparticles of mass $m>0$ and spin-1/2, respectively.

\begin{table}[h]
\caption{The perception and boosting bundles for Dirac particles}
\label{tab:4}
\centering
\begin{tabular}{|m{1.2cm}|m{7cm}|m{5.5cm}|}
\hline\noalign{\smallskip}
  &  $E^\pm \text{ (The perception bundle)}$ &  $E_L ^\pm \text{ (The boosting bundle)}$ \\
\noalign{\smallskip}\hline\noalign{\smallskip}
Bundle  &  $R ^\pm = \left \{ (p, z ) \in X_m ^\pm \times \mathbb{C}^4 : p_\mu \gamma^\mu z = mz \right\} $  &  $X_m ^\pm \times V^\pm$ \\
\noalign{\smallskip}\hline\noalign{\smallskip}
Metric &   $h^\pm \Big( (p , v) , (p, w) \Big)= \pm v^\dagger \gamma^0 w = \frac{m}{| p^0 |} v^\dagger w$  &   $h_L ^\pm \Big( (p , v) , (p ,w) \Big)= v^\dagger w $ \\
\noalign{\smallskip}\hline\noalign{\smallskip}
Action &  $\lambda^\pm (a,\Lambda) (p , v)=  \Big( \Lambda p  , e^{-i \langle \Lambda p , a \rangle} \Phi (\Lambda) v \Big)$  &    $\lambda_{L } ^\pm (a, \Lambda) (p , v)= $ \newline $ \Big(\Lambda p , e^{-i \langle \Lambda p , a \rangle} \sigma^\pm ( W_{L^\pm} (\Lambda,p)) v \Big)$ \\
\noalign{\smallskip}\hline\noalign{\smallskip}
Space & $\mathcal{H} ^\pm := L^2 \Big( X_m ^\pm , E^\pm ; \mu^\pm , h^\pm \Big)$ & $\mathcal{H}_{L } ^\pm := L^2 ( X_m ^\pm  ; \mu^\pm ) \otimes V^\pm$ \\
\noalign{\smallskip}\hline\noalign{\smallskip}
$\pi_{m,1/2} ^\pm$ & $ U^\pm (a, \Lambda) \phi = \lambda ^\pm (a, \Lambda) \circ \phi \circ \Lambda^{-1}$ & $U_{L} ^\pm (a, \Lambda)\psi = \lambda_{L} ^\pm (a, \Lambda) \circ \psi \circ \Lambda^{-1}$ \\
\noalign{\smallskip}\hline
\end{tabular}
\end{table}

The isomorphisms between the two descriptions (Eq.~(\ref{eq:7.1})) in this case are given by the Hermitian $G$-bundle isomorphisms
\begin{equation}\label{eq:7.7}
\begin{tikzcd}[baseline=(current  bounding  box.center), column sep=1.5em]
    \alpha^\pm : E^\pm \arrow ["{(p,v) \mapsto (p , \Phi \big(L^\pm (p) ^{-1} \big)v)}"]{rrrr} \arrow{drr} 
   & & & & E_{L} ^\pm \arrow{dll}
\\
   & &X_m ^\pm& &
\end{tikzcd}.
\end{equation}

\subsubsection{The vector bundle point of view for Dirac particles}\label{sec:7.2.1}

\hfill

As in Sect.~\ref{sec:6.1}, the description table Table~\ref{tab:4} tells us that if two inertial observers Alice and Bob, who are related by a Lorentz transformation $(a, \Lambda) \in G$ as in Eq.~(\ref{eq:2.11}), are using the two bundle descriptions for Dirac particles to describe a massive particle/antiparticle with spin-1/2,\footnote{For the precise meaning of this sentence, see Sect.~\ref{sec:6.1}.} then the descriptions should be related by

\begin{gather}
\lambda^\pm (a, \Lambda ) : E ^{\pm A} \rightarrow E ^{\pm B} \nonumber \\
(p,c)^A \mapsto \left(  \Lambda p, e^{- i (\Lambda p)_\mu a^\mu} \Phi \left(\Lambda \right) c \right)^B \label{eq:7.8}
\end{gather}
and
\begin{gather}
\lambda_{L} ^\pm (a, \Lambda ) : E_{L} ^{\pm A} \rightarrow E_{L} ^{\pm B} \nonumber \\
(p,c)^A \mapsto \left(  \Lambda p, e^{- i (\Lambda p)_\mu a^\mu} \sigma^\pm \left(W_{L^\pm} ( \Lambda, p) \right) c \right)^B \label{eq:7.9},
\end{gather}
respectively.

\subsubsection{Physical interpretations of the two bundle descriptions}\label{sec:7.2.2}

\hfill

Let's see whether the discussions in Sect.~\ref{sec:3.3} carry over to the two descriptions $E^\pm $ and $E_L ^\pm $ as well.

For that, we need analogues of Eqs.~(\ref{eq:3.13}) and (\ref{eq:3.20}). Given a bispinor $c = \frac{1}{\sqrt{2}} \begin{pmatrix} \chi \\ \pm \chi \end{pmatrix} \in V^\pm$, we form
\begin{equation}\label{eq:7.10}
c c^\dagger =  \frac{1}{4} \begin{pmatrix} \boldsymbol{\tau} \cdot \mathbf{n} + I_2 & \pm (\boldsymbol{\tau} \cdot \mathbf{n} + I_2) \\ \pm (\boldsymbol{\tau} \cdot \mathbf{n} + I_2) & \boldsymbol{\tau} \cdot \mathbf{n} + I_2 \end{pmatrix} = \frac{1}{4} (I_4 \pm \gamma^0 ) \begin{pmatrix} \boldsymbol{\tau} \cdot \mathbf{n} + I_2 & 0 \\ 0 & \boldsymbol{\tau} \cdot \mathbf{n} + I_2 \end{pmatrix} 
\end{equation}
where $\mathbf{n}$ is the spin direction of the qubit $\chi \in \mathbb{C}^2$ (cf. Eq.~(\ref{eq:3.13})). Thus,
\begin{align}\label{eq:7.11}
\Phi( L^\pm (p)) c c^\dagger \Phi( L^\pm (p))^\dagger = \frac{1}{2m} \begin{pmatrix} \tilde{w} \pm \frac{ \tilde{p}}{2} & \frac{  \tilde{p} } {m}  \left( -\utilde{w} \pm \frac{ \utilde{p}}{2} \right)  \\ \frac{  \utilde{p} } {m} \left( \tilde{w} \pm \frac{ \tilde{p}}{2} \right)  &  -\utilde{w} \pm \frac{ \utilde{p}}{2} \end{pmatrix} \nonumber \\
= \frac{1}{2m} (I_4 + \frac{1}{m}  p_\mu \gamma^\mu) \begin{pmatrix} \tilde{w} \pm \frac{\tilde{p}}{2} & 0 \\ 0 & -\utilde{w} \pm \frac{\utilde{p}}{2} \end{pmatrix}
\end{align}
where $w = L^\pm (p) (0, \frac{1}{2} \mathbf{n} )$ is the Pauli-Lubansky four-vector of the qubit $\chi$ (cf. Eq.~(\ref{eq:3.20})).

So, just as for Eqs.~(\ref{eq:3.13}) and (\ref{eq:3.20}), we define the \textit{spin direction of the bispinor $c \in V^\pm$} as the three-vector $\mathbf{n} \in \mathbb{R}^3 \leq \mathbb{R}^4$ which satisfies
\begin{equation}\label{eq:7.12}
\boldsymbol{\tau} \cdot \mathbf{n} = 4 \left( c c^\dagger \right)_{11} - I_2
\end{equation}
(here $M_{11}$ denotes the upper left $2 \times 2$ component of the $4 \times 4$ complex matrix $M$) and the \textit{Pauli-Lubansky four-vector of the bispinor $(p, d ) \in E_p ^\pm$} as the four-vector $w \in \mathbb{R}^4$ which satisfies
\begin{equation}\label{eq:7.13}
\tilde{w} = 2m \left (d d^\dagger \right)_{11} \mp \frac{\tilde{p}}{2}.
\end{equation}

Using these, one can repeat the arguments in Sect.~\ref{sec:3.3} to see that the same interpretations also hold for Dirac particles. I.e., the description $E^\pm$ is related to the Pauli-Lubansky four-vector and $E_L ^\pm$ is related to the Newton-Wigner spin in relation to the former.

In particular, by inserting $(a,\Lambda) = ( 0 , L^\pm (p))$ into Eq.~(\ref{eq:7.9}), we see the following analogue of Eq.~(\ref{eq:3.12}) still holds
\begin{equation}\label{eq:7.14}
\lambda_L ^{\pm} (0, L^\pm (p) ) (p_m ^\pm , c)^A = (p, c)^B
\end{equation}
for $(p_m ^\pm , c) \in (E_L ^\pm )_{p_m ^\pm}$ (cf. Eq.~(\ref{eq:6.4})). Thus, the same remark discussed right below Eq.~(\ref{eq:3.12}) still holds for $E_L ^\pm$. I.e., the elements in $E_L ^\pm$ don't reflect the perception of the fixed inertial observer who is using the bundle $E_L ^\pm$ and hence the description $E_L ^\pm$ inherently depends on frame change considerations.

Also, by inserting $(a, \Lambda) = (0, L^\pm (p))$ into Eq.~(\ref{eq:7.8}) for $(p_m ^\pm , c) \in (E^\pm)_{p_m}$ and checking Eq.~(\ref{eq:7.11}) once more, we see that one can recover the "relativistic perception" $\tilde{w} \pm \frac{ \tilde{p}}{2}$ from the bispinor $(p, \Phi (L (p)) c ) \in E_p ^\pm$ without recourse to frame change considerations and vice versa. In other words, the perception bundle description $E^\pm$ correctly reflects the perception of a fixed inertial observer who is using this bundle for the description of a Dirac particle.

\subsubsection{Theoretical implications on the representations}\label{sec:7.2.3}

\paragraph{The Dirac equation as a manifestation of relativistic perception}

\hfill

We first investigate a striking consequence of the definition of the perception representation $(\mathcal{H}^\pm, U^\pm)$. Fix $\phi \in \mathcal{H}^\pm$. From Table~\ref{tab:4}, we have
\begin{equation}\label{eq:7.15}
[ U^\pm (a,I) \phi ] (p) = \lambda^\pm (a, I) \circ \phi (p) = e^{-i \langle p , a \rangle} \phi(p), \quad \forall p \in X_m ^\pm.
\end{equation}

If we define the four-momentum operators $P^\mu$ on $\mathcal{H}^\pm$ as in Eq.~(\ref{eq:4.10}), then Eq.~(\ref{eq:7.15}) and a computation analogous to Eq.~(\ref{eq:4.10}) would show that by virtue of Eq.~(\ref{eq:7.5}),
\begin{equation}\label{eq:7.16}
[P_\mu \gamma^\mu \phi](p) = p_\mu \gamma^\mu \phi(p) = m \phi (p)
\end{equation}
always holds for ${}^\forall \phi \in \mathcal{H}^\pm$.

Eq.~(\ref{eq:7.16}) is the \textit{Dirac equation}\footnote{In Sect.~\ref{sec:7.4}, we will see how Eq.~(\ref{eq:7.16}) can be converted into the more familiar form of differential equation. Cf. Eq.~(\ref{eq:7.32}).} which is obeyed by massive spin-1/2 particles/antiparticles and is of fundamental importance in QFT (cf. \cite{weinberg}). We have just found that it is automatically satisfied for all wave functions in $\mathcal{H}^\pm$. Given the interpretations of the perception bundles $E^\pm$ presented in Sect.~\ref{sec:7.2.2}, we find that \textit{the Dirac equation is nothing but a manifestation of a fixed inertial observer's perception of the internal quantum states of massive particles/antiparticles with spin-1/2}. This fact is even more clear if we look once again at the definition of the bundles $E^\pm = R^\pm$ given in Eq.~(\ref{eq:7.5}). The Dirac equation is not only satisfied by the wave functions in $\mathcal{H}^\pm$ but also manifests itself in the level of bispinors in the fibers of the perception bundles $E^\pm$ as perceived by a fixed inertial observer (See Remark~\ref{remark:6.2} for more on this point).

\paragraph{The Foldy-Wouthuysen transformation}

\hfill

The boosting representation
\begin{subequations}\label{eq:7.17}
\begin{equation}\label{eq:7.17a}
\mathcal{H}_L ^\pm = L^2 (X_m ^\pm ; \mu^\pm ) \otimes V^\pm
\end{equation}
\begin{equation}\label{eq:7.17b}
[U_L ^\pm (a, \Lambda) \psi] (p) = e^{-i \langle p , a \rangle} \sigma^\pm ( W_{L^\pm} (\Lambda, \Lambda^{-1} p ) \psi (\Lambda^{-1} p )
\end{equation}
\end{subequations}
from Table~\ref{tab:4} has been the standard approach to the description of Dirac particles in the physics literature. However, the perception representation for Dirac particles
\begin{subequations}\label{eq:7.18}
\begin{equation}\label{eq:7.18a}
\mathcal{H} ^\pm = L^2 (X_m ^\pm , E^\pm ; \mu^\pm , h^\pm )
\end{equation}
\begin{equation}\label{eq:7.18b}
[U ^\pm (a, \Lambda) \phi] (p) = e^{-i \langle p , a \rangle} \Phi  (\Lambda ) \phi (\Lambda^{-1} p )
\end{equation}
\end{subequations}
from Table~\ref{tab:4} has also been given some attention. In fact, composing the Hermitian $G$-bundle isomorphism Eq.~(\ref{eq:7.7}) to wave functions in $\mathcal{H}_L ^\pm$, we get a unitary map $FW^\pm := \alpha^{-1} \circ (\cdot) : \mathcal{H}_L ^\pm \rightarrow \mathcal{H} ^\pm$ intertwining the two representations $U_L ^\pm$ and $U ^\pm$. Unwinding the definitions, we see for ${ }^\forall \psi \in \mathcal{H}_L ^\pm$,
\begin{align}\label{eq:7.19}
FW^\pm (\psi) (p) := \alpha^{-1} \circ \psi (p) = \Phi( L^\pm (p) ) \psi (p) = \frac{1}{\sqrt{2m (m \pm p^0 ) }} \left( m I_4 + p_\mu \gamma^\mu \right) \psi(p) \nonumber \\
= \frac{1}{\sqrt{2m (m \pm p^0 ) }} \left( (m \pm p^0 ) I_4 - \mathbf{p} \cdot \boldsymbol{\gamma} \right) \psi(p)
\end{align}
holds. Notice that the last expression is the \textit{Foldy-Wouthuysen transformation} suggested in \cite{foldy1950}. (cf. \cite{deriglazov2020})

This transformation has been widely used to arrange the Dirac Hamiltonian in a mathematically palatable way. Interested readers are referred to \cite{costella1995} for a brief historical account of it and its usefulness in dealing with Dirac particles. We have found that \textit{the Foldy-Wouthuysen transformation is a change of representations from the boosting description into the perception description}.

The same reasoning as in Sect.~\ref{sec:6.2.2} would show that
\begin{equation}\label{eq:7.20}
(U^\pm)_* (\mathbf{S}_{NW} ) = (\alpha^{\pm})^{-1} \circ \left(1 \otimes \Phi_* ( \frac{1}{2} \boldsymbol{\tau} ) \right) \circ \alpha^\pm
\end{equation}
holds. I.e., the $V^\pm$-component of the space $\mathcal{H}_L ^\pm = L^2( X_m ^\pm , \mu^\pm) \otimes V^\pm$ contains information of the Newton-Wigner spin on the representation space $\mathcal{H}^\pm$.

\subsection{The Minkowski space representation of massive particles with spin-1}\label{sec:7.3}

In this subsection, we analyze massive particles with spin-1. The $W$ and $Z$ bosons which are responsible for the weak interaction are of this type. In the context of RQI, this case has been investigated, for example, in \cite{caban2008a, caban2008b}. Again in this subsection, we restrict our attention to the mass shell $X_m = X_m ^ +$ and remove all the $+$-superscripts as we had done in Sects.~\ref{sec:6}--\ref{sec:6.2}.

\paragraph{The perception bundle for massive particles with spin-1}

\hfill

Just as in Sect.~\ref{sec:7.2}, we do not choose $\sigma_1$ from Theorem~\ref{theorem:4.7}. Instead, note that the representation $\sigma:= \kappa|_{SU(2)} : SU(2) \rightarrow SO(3) \hookrightarrow U(3)$ (a restriction of the covering map Eq.~(\ref{eq:2.14})) is an irreducible unitary representation of dimension $3 = 2 \cdot 1 + 1$. So, by Theorem~\ref{theorem:4.7}, it is equivalent to $\sigma_1$. Notice that the representation $\sigma : SU(2) \rightarrow SO(3)$ has a (non-unitary) extension $\kappa : SL(2, \mathbb{C}) \rightarrow SO^\uparrow(1,3) \hookrightarrow GL(4, \mathbb{C})$. So, we can apply Theorem~\ref{theorem:7.1} with the understanding that $\mathcal{H}_\sigma := \mathbb{C}^3 = \{ (z^0 , z^1 , z^2 , z^3 ) \in \mathbb{C}^4: z^0 =0 \}  \leq \mathbb{C}^4 =: \mathcal{K}_\kappa$.

\begin{proposition}\label{proposition:7.3}
The range bundle $R \leq X_m \times \mathbb{C}^4$ of Table~\ref{tab:3} for the Minkowski space representation is given by
\begin{equation}\label{eq:7.21}
R= \left\{ (p, v) \in X_m \times \mathbb{C}^4 : p_\mu v^\mu = 0 \right\},
\end{equation}
which is a rank-3 subbundle.

The Hermitian metric $h$ on $R$ provided by Table~\ref{tab:3} for this case becomes
\begin{equation}\label{eq:7.22}
h_p (v, w ) = v^\dagger \kappa ( \frac{ \utilde{p}} {m} ) w = - v^\dagger \eta w
\end{equation}
for $v,w \in R_p$.
\end{proposition}

\begin{proof}
Let $(p,v) \in X_m \times \mathcal{H}_\sigma $. Then,
\begin{equation*}
p_\mu  \Big(\kappa \big(L(p)\big)  v \Big)^\mu = \langle p , \kappa \big(L(p)\big) v \rangle = \langle \kappa \big(L(p)\big)^{-1} p , v \rangle =\langle p_m , v \rangle = 0
\end{equation*}
since $\kappa \big(L(p)\big) \in SO^\uparrow(1,3)$ and $v \in \mathcal{H}_\sigma = \{ (z^0 , z^1 , z^2 , z^3 ) \in \mathbb{C}^4: z^0 =0 \} \leq \mathbb{C}^4$ while $p_m = (m,0,0,0)$. So, indeed the map Eq.~(\ref{eq:7.1}) maps the boosting bundle into $R$. It is an isomorphism being an injection between two bundles of rank 3.

The first part of Eq.~(\ref{eq:7.22}) is easy and for the second part, notice that the preceding calculation shows that $\kappa( L(p)^{-1} ) v \in \mathcal{H}_\sigma \leq \mathbb{C}^4$ for all $v \in R_p$ and hence
\begin{align*}
v^\dagger \kappa (\frac{ \utilde{p}}{m}) w = \Big( \kappa ( L(p)^{-1} ) v \Big)^\dagger \Big( \kappa ( L(p)^{-1} ) w \Big) = - \left< \kappa (L(p)^{-1}) \overline{v} , \kappa( L(p)^{-1}) w \right>
= - \langle \overline{v} , w \rangle
\end{align*}
where $\langle \cdot, \cdot \rangle$ here denotes the complexified Minkowski metric.
\end{proof}

\paragraph{The description table for massive particles with spin-1}

\hfill

With the help of Proposition~\ref{proposition:7.4}, we apply Theorem~\ref{theorem:7.1} to obtain Table~\ref{tab:5}, the description table for massive particles with spin-1. Note that since $\sigma \cong \sigma_1$, we see $\textup{Ind}_{G_{p_m}} ^G (e^{-i \langle p_m , \cdot \rangle} \sigma) \cong \pi_{m,1}$, which represents particles of mass $m>0$ and spin-1. We call this representation the \textit{Minkowski space representation of massive particle with spin-1}.

\begin{table}[h]
\caption{The perception and boosting bundles for the Minkowski space representation}
\label{tab:5}

\centering
\begin{tabular}{|m{1.2cm}|m{7cm}|m{5.5cm}|}
\hline\noalign{\smallskip}
  &  $E \text{ (The perception bundle)}$ &  $E_{L} \text{ (The boosting bundle)}$ \\
\noalign{\smallskip}\hline\noalign{\smallskip}
Bundle  &  $R= \left\{ (p, v) \in X_m \times \mathbb{C}^4 : p_\mu v^\mu = 0 \right\}$  &  $X_m  \times \mathbb{C}^3 $ \\
\noalign{\smallskip}\hline\noalign{\smallskip}
Metric &   $h \Big( (p , v) , (p, w) \Big)= v^\dagger \kappa(\frac{\utilde{p}}{m}) w = - v^\dagger \eta w$  &   $h_{L}  \Big( (p , v) , (p ,w) \Big)= v^\dagger w $ \\
\noalign{\smallskip}\hline\noalign{\smallskip}
Action &  $\lambda (a,\Lambda) (p , v)=  \Big( \Lambda p  , e^{-i \langle \Lambda p , a \rangle} \kappa (\Lambda) v \Big)$  &    $\lambda_{L }  (a, \Lambda) (p , v)= $ \newline $ \Big(\Lambda p , e^{-i \langle \Lambda p , a \rangle} \kappa ( W_{L} (\Lambda,p)) v \Big)$ \\
\noalign{\smallskip}\hline\noalign{\smallskip}
Space & $\mathcal{H} := L^2 \Big( X_m  , E ; \mu , h \Big)$ & $\mathcal{H}_{L }  := L^2 ( X_m   ; \mu ) \otimes \mathbb{C}^3$ \\
\noalign{\smallskip}\hline\noalign{\smallskip}
$\pi_{m,1}$ & $ U (a, \Lambda) \phi = \lambda (a, \Lambda) \circ \phi \circ \Lambda^{-1}$ & $U_{L} (a, \Lambda)\psi = \lambda_{L}  (a, \Lambda) \circ \psi \circ \Lambda^{-1}$ \\
\noalign{\smallskip}\hline
\end{tabular}
\end{table}

\subsubsection{The vector bundle point of view for massive particles with spin-1}\label{sec:7.3.1}

\hfill

As in Sect.~\ref{sec:7.2.1}, the description table Table~\ref{tab:5} tells us that if two inertial observers Alice and Bob, who are related by a Lorentz transformation $(a, \Lambda) \in G$ as in Eq.~(\ref{eq:2.11}), are using the two bundle descriptions to describe a massive particle with spin-1, then the descriptions should be related by

\begin{gather}
\lambda (a, \Lambda ) : E ^A \rightarrow E ^B \nonumber \\
(p,v)^A \mapsto \left(  \Lambda p, e^{- i (\Lambda p)_\mu a^\mu} \kappa \left(\Lambda \right) v \right)^B \label{eq:7.23}
\end{gather}
and
\begin{gather}
\lambda_{L} (a, \Lambda ) : E_{L} ^A \rightarrow E_{L} ^B \nonumber \\
(p,v)^A \mapsto \left(  \Lambda p, e^{- i (\Lambda p)_\mu a^\mu} \kappa \left(W_{L^\pm} ( \Lambda, p) \right) v \right)^B \label{eq:7.24},
\end{gather}
respectively.

\subsubsection{Physical interpretations of the two bundle descriptions}\label{sec:7.3.2}

\hfill

To find physical interpretations of the two bundle descriptions, we need to examine the three-level quantum system $\mathcal{H}_{\sigma} = \mathbb{C}^3$. Since
\begin{equation}\label{eq:7.25}
\hat{J}^3 := i (\sigma)_* (J^3) = \begin{pmatrix}  0 & -i & 0 \\  i & 0 & 0 \\  0 & 0 & 0 \end{pmatrix}
\end{equation}
in this case, the three vectors
\begin{equation}\label{eq:7.26}
e_1 = \frac{1}{\sqrt{2}} \begin{pmatrix} 1 \\ i \\ 0 \end{pmatrix}, e_0 = \begin{pmatrix} 0 \\ 0 \\ 1\end{pmatrix}, e_{-1} = \frac{1}{\sqrt{2}} \begin{pmatrix} i \\ 1 \\ 0 \end{pmatrix}
\end{equation}
are eigenvectors of the operator $\hat{J}^3$ with eigenvalues $1 , 0 , -1$, respectively.

The vectors $v \in \mathcal{H}_{\sigma}$ have the meaning of \textit{polarization} in QFT (cf. \cite{weinberg}). In this interpretation, $e_1$ gives the right-handed circular polarization, $e_{-1}$ the left-handed circular polarization, and $e_0$ the longitudinal polarization along the $\hat{z}$-axis, respectively. This interpretation becomes clear only when one expands the solutions of the Proca equations (Eq.~(\ref{eq:7.29})) as generalized linear combinations (i.e., integrals) of the plane wave solutions $e^{\pm i p_\mu x^\mu}$. In such an expression, each vector $v$ of $(p, v) \in X_m \times \mathbb{C}^3$ has the meaning of the polarization along the direction of each momentum $p \in X_m$. For more details, see Sect.~5.3 of \cite{weinberg} together with any physics textbook that deals with electromagnetic waves.

The point is that each element $(p_m , v)^A \in (E_{L})_{p_m} ^A = E_{p_m} ^A$ is a genuine (complex) three-vector in Alice's rest frame. So, in the $L (p)$-transformed Bob's frame, the vector $v$ is perceived as the four-vector $\kappa(L (p)) v \in \mathbb{C}^4$, which is precisely captured by the transformation law Eq.~(\ref{eq:7.23}). Therefore, we conclude that the elements in $E$ are "relativistic perception" of a fixed observer who is using $E$ for the description of a massive particle with spin-1.

As usual, Eq.~(\ref{eq:7.24}), evaluated at $(a,\Lambda) = (0, L(p))$ would give
\begin{equation}\label{eq:7.27}
\lambda_{L} (0, L(p)) (p_m , v )^A = (p, v)^B,
\end{equation}
which implies that the fibers of the bundle $E_{L}$ does not reflect the perception of a fixed inertial observer who is using $E_L$ for the description of a massive particle with spin-1, in contrast to the $E$-bundle description. So, we see that the interpretation of Sect.~\ref{sec:3.3} also holds for the Minkowski space representation of massive spin-1 particles.

\subsubsection{Theoretical implications on the representations}\label{sec:7.3.3}

\paragraph{The Proca equations as a manifestation of relativistic perception}

\hfill

Fix $\phi \in \mathcal{H}$. From Table~\ref{tab:5}, we have
\begin{equation}\label{eq:7.28}
[U (a,I) \phi ] (p) = \lambda (a,I) \circ \phi(p) = e^{-i \langle p , a \rangle} \phi(p), \quad \forall p \in X_m.
\end{equation}

If we define the four-momentum operators $P^\mu$ on $\mathcal{H}$ as in Eq.~(\ref{eq:4.10}), then Eq.~(\ref{eq:7.28}) and a computation analogous to Eq.~(\ref{eq:4.10}) would show that, by virtue of Eq.~(\ref{eq:7.21}) and Definition~\ref{definition:4.5},
\begin{equation}\label{eq:7.29}
[P_\mu \phi ^\mu](p) = p_\mu \phi(p) ^\mu = 0, \quad [P_\mu P^\mu \phi](p) = p_\mu p^\mu \phi(p) = m^2 \phi(p)
\end{equation}
always hold for ${ }^\forall \phi \in \mathcal{H}$.

The set of equations Eq.~(\ref{eq:7.29}) are the \textit{Proca equations}\footnote{In Sect.~\ref{sec:7.4}, we will see how Eq.~(\ref{eq:7.29}) can be converted into the more familiar form of differential equation. Cf. Eq.~(\ref{eq:7.36}).} which are obeyed by massive spin-1 particles and become the Maxwell equations with the Lorentz gauge condition when $m \rightarrow 0$ (cf. \cite{weinberg}). As in Sect.~\ref{sec:7.2.2}, we remark that given the interpretations of the bundle $E$ presented in Sect.~\ref{sec:7.3.2}, we find that \textit{the Proca equations are nothing but a manifestation of a fixed inertial observer's perception of the internal quantum states of massive particles with spin-1}. In fact, as one can see from Eq.~(\ref{eq:7.21}), the Proca equations manifest themselves even in the level of elements in the fibers of the perception bundle $E$ as perceived by a fixed inertial observer.

\subsection{A link between the theory of relativistic quantum measurement}\label{sec:7.4}

Even though we have not covered any aspect of measurement in this paper, one must be conversant with the theory of relativistic quantum measurement in order to apply the mathematical framework developed in this paper to actual problems of RQI. Particularly, a relativistic measurement theory based on foliations of space-time and the Schwinger-Tomonaga equation (cf. Ch.~11 of \cite{breuer}) can be applied to the single-particle state spaces analyzed in Sect.~\ref{sec:7}. We want to indicate how in this subsection.

\paragraph{Dirac particles}

\hfill

Let $\phi \in \mathcal{H}^\pm $ (cf. Table~\ref{tab:4}). For each $t \in \mathbb{R}$, we define
\begin{equation}\label{eq:7.30}
\psi(t, \mathbf{x}) := \int_{X_m ^\pm} \exp \left(- i p^0 t + i \mathbf{p} \cdot \mathbf{x} \right) \sqrt{m} \phi( p) \frac{d \mu^\pm (p)}{ (2 \pi)^{\frac{3}{2}} },
 \quad \mathbf{x} \in \mathbb{R}^3.
\end{equation}

Observe that
\begin{equation*}
\int_{\mathbb{R}^3} \left( \frac{\sqrt{m}\phi(p)} {|p^0|} \right)^\dagger \left( \frac{\sqrt{m}\phi(p)} {|p^0|} \right)  d^3 \mathbf{p} = \int_{X_m ^\pm} h_p ^\pm \big(\phi(p), \phi(p) \big) d \mu^\pm (p) = \| \phi \|_{\mathcal{H}^\pm} ^2
\end{equation*}
by the definition of the norm $\| \cdot \|_{\mathcal{H}^\pm}$ of $\mathcal{H}^\pm$ and Eqs.~(\ref{eq:6.2}), (\ref{eq:7.6}). Therefore, we see that for each $t \in \mathbb{R}$, we have $\psi (t, \cdot ) \in L^2 (\mathbb{R}^3 , \mathbb{C}^4 )$ and each map
\begin{equation}\label{eq:7.31}
\mathcal{H}^\pm \ni \phi \mapsto \psi(t, \mathbf{x}) \in L^2 (\mathbb{R}^3 , \mathbb{C}^4)
\end{equation}
is an isometry (by Plancherel's theorem). Note that Eq.~(\ref{eq:7.30}) is the four-dimensional Fourier inversion formula restricted to the mass shells $X_m ^\pm$.

So, we see that the function $\psi$ defined on the spacetime $\mathbb{R}^4$ at least gives rise to an $L^2(\mathbb{R}^3 , \mathbb{C}^4)$-valued functional defined on the set of foliations by spacelike hyperplanes of the spacetime (such as $\big\{ \{t\} \times \mathbb{R}^3 \big\}_{t \in \mathbb{R}}$ or the foliations obtained from it by applying Lorentz transformations). On this functional, we can apply the formalism of \cite{breuer} to test the relativistic measurement schemes developed there on the state $\phi$ of the perception space $\mathcal{H}^\pm$.

As a final note, observe that if we think of $\psi (x) = \psi(t, \mathbf{x})$ as a function defined on the Minkowski space $\mathbb{R}^4$, then formally we have

\begin{equation}\label{eq:7.32}
i \gamma^\mu \frac{\partial}{\partial x^\mu} \psi  = m \psi 
\end{equation}
by Eqs.~(\ref{eq:7.30}) and (\ref{eq:7.16}). I.e., $\psi$ satisfies the Dirac equation.

\paragraph{Massive particles with spin-1}

\hfill

Again, we restrict our attention to the mass shell $X_m ^+$ and suppress all the $+$-signs in the superscripts. Let $\phi \in \mathcal{H}$ (cf. Table~\ref{tab:5}). For each $t \in \mathbb{R}$, we define
\begin{equation}\label{eq:7.33}
\psi(t, \mathbf{x}) := \int_{X_m } \exp \left(- i p^0 t + i \mathbf{p} \cdot \mathbf{x} \right) \sqrt{p^0} \phi( p) \frac{d \mu (p)}{ (2 \pi)^{\frac{3}{2} }},
 \quad \mathbf{x} \in \mathbb{R}^3.
\end{equation}

Then,
\begin{equation*}
\int_{\mathbb{R}^3} \left(\frac{\phi(p)}{ \sqrt{p^0}} \right)^\dagger (-\eta) \left(\frac{\phi(p)}{ \sqrt{p^0}} \right) d^3 \mathbf{p} = \int_{X_m ^\pm} h_p \big(\phi(p), \phi(p) \big) d \mu (p) = \| \phi \|_{\mathcal{H}} ^2
\end{equation*}
by the definition of the norm $\| \cdot \|_\mathcal{H}$ of $\mathcal{H}$ and Eqs.~(\ref{eq:6.2}), (\ref{eq:7.22}). Since, for arbitrary $\mathbb{C}^4$-valued Schwartz class functions $f ,g $ on $\mathbb{R}^3$,
\begin{equation*}
\int_{\mathbb{R}^3} \hat{f} (\mathbf{x})^\dagger (-\eta) \hat{g}(\mathbf{x}) d^3 \mathbf{x} = \int_{\mathbb{R}^3} f (\mathbf{y})^\dagger (-\eta) g ( \mathbf{y}) d^3 \mathbf{y}
\end{equation*}
(here $ \hat{ ( \cdot )}$ denotes the Fourier transform) holds (cf. \cite{rudin3}), we see that for each $t \in \mathbb{R}$, we have
\begin{equation}\label{eq:7.34}
\psi (t, \cdot ) \in \mathcal{K} := \left\{ \varphi \in L^2 (\mathbb{R}^3 , \mathbb{C}^4 ) : 0 \leq \int_{\mathbb{R}^3} \varphi (\mathbf{x}) ^\dagger (-\eta) \varphi(\mathbf{x}) d^3 \mathbf{x} \right\}
\end{equation}
and each map
\begin{equation}\label{eq:7.35}
\mathcal{H} \ni \phi \mapsto \psi(t, \mathbf{x}) \in \mathcal{K}
\end{equation}
is an isometry if we endow $\mathcal{K}$ with the obvious inner product that uses $-\eta$.

As in the case of Dirac particles, we can use the function $\psi$ defined on $\mathbb{R}^4$ to obtain a $\mathcal{K}$-valued functional on the set of foliations by spacelike hyperplanes of the spacetime, thus providing a link between the theory of relativistic measurement.

Finally, regarding $\psi(x) = \psi(t, \mathbf{x})$ as a function defined on the Minkowski space $\mathbb{R}^4$, we obtain
\begin{equation} \label{eq:7.36}
\frac{\partial}{\partial x^\mu} \psi^\mu = 0, \quad \frac{\partial}{\partial x^\mu} \frac{\partial}{\partial x_\mu} \psi = m^2 \psi
\end{equation}
by Eqs.~(\ref{eq:7.33}) and (\ref{eq:7.29}). I.e., $\psi$ satisfies the Proca equations.

\section{Concluding remarks and future research}\label{sec:8}

From the basic principles of SR and QM, we have obtained the definition of single-particle state spaces, which are the smallest possible quantum systems in which one can test relativistic considerations. We briefly surveyed the pioneering works of RQI and observed that the notions of spin state independent of momentum, spin entropy, and spin entanglement, which are important quantum informational resources, are not relativistically meaningful. Rephrasing the definition of single-particle state spaces in terms of the bundle theoretic language which is developed in this paper, we were able to figure out the root of the problem. Namely, the boosting bundle description, which has been used in the RQI literature almost exclusively, does not correctly reflect the perception of a fixed inertial observer and therefore the definitions of the above notions become illegitimate algebraic operations.

We have seen that the perception bundle description is free from this issue and hence can be used as a kind of coordinate system for a moving finite-dimensional quantum system which naturally extends the classical coordinate for a moving classical particle. We have extended the bundle descriptions to the case of massive particle with arbitrary spin, observed that the results for the spin-1/2 case holds in full generality, and defined the Pauli-Lubansky reduced matrix for massive particles with arbitrary spin, which is a Lorentz covariant $(2s +1 ) \times (2s +1)$-matrix containing information of average internal quantum state as perceived by a fixed inertial observer. As an application of the perception bundle description developed in this paper, we have seen that the Dirac equation and the Proca equations, which are fundamental equations of QFT obeyed by massive particles with spin-1/2 and 1, respectively, emerge as manifestations of a fixed inertial observer's perception of the internal quantum states of massive particles with spin-1/2 and 1, respectively. We also briefly indicated a link between the formalism developed in this paper and the theory of relativistic quantum measurement.

While this paper has laid the mathematical foundation for a new framework of single-particle state spaces better suited for RQI investigation and seen some striking theoretical implications of the framework, it has not given any application of the perception bundle description to actual problems of RQI. Given the conceptual advantages of this description over the more standard boosting bundle description as shown in this paper, it is very likely that recasting subtle problems of RQI in terms of the perception bundle description will give profound insight into the problems. An approach closely related to this has appeared only very recently in \cite{caban2018, caban2019}. Interested researchers are invited to pursue this direction of study.

In the sequel to this work, however, the author is planning to investigate \textit{massless particles with helicity} by applying the mathematical theory developed in this paper. More precisely, we are going to give the massless analogues of the boosting and perception bundle descriptions, survey some of the RQI papers that deal with massless particles, see if the same interpretations are possible, and draw some interesting theoretical implications from them.

\appendix

\section{Proof of Theorem~\ref{theorem:5.5}}\label{sec:A}

\def\thesection{\Alph{section}}

In this appendix, we provide a proof for Theorem~\ref{theorem:5.5}. 

\subsection{Preliminaries}\label{sec:A.1}

To prove Theorem~\ref{theorem:5.5}, we need to relate the two induced representation constructions of Definitions~\ref{definition:4.1} and \ref{definition:5.3}. The relation is provided by the language of principal fiber bundle and associated bundles.

\paragraph{Principal fiber bundle and associated bundle}

\hfill

The associated bundle construction of a principal fiber bundle is the primary source of the Hermitian $G$-bundles that will be addressed in this paper. The main reference for this construction is \cite{tu}, Ch.~6.

\begin{definition}\label{definition:A.1}
Let $H$ be a Lie group and $P$ be a right $H$-manifold whose action is smooth and free. A smooth map $P \xrightarrow{\pi} M$ is called a \textit{principal $H$-bundle} if, for every $x \in M$, there is an open set $x \in U \subseteq M$ and an $H$-equivariant fiber preserving diffeomorphism
\begin{equation*}
\begin{tikzcd}[baseline=(current  bounding  box.center)]
U \times H  \arrow[rr, "\phi"] \arrow[dr ]
& & \pi^{-1} (U)  \arrow{dl} \\
& U &
\end{tikzcd}.
\end{equation*} 
\end{definition}

If $P \xrightarrow{\pi} M$ is a principal $H$-bundle, then the right action of $H$ on $P$ is free and proper, and $P/H \cong M$. Conversely, it is easy to show that if an action of $H$ on $P$ is free and proper, then $P \rightarrow P/H$ is a principal $H$-bundle (cf. \cite{lee}, Ch.~21).

In particular, for every Lie group $G$ and a closed subgroup $H$, the right multiplicative action of $H$ on $G$ is free and proper and hence $G \rightarrow G/H$ is a principal $H$-bundle. This particular class of principal bundles will be of paramount importance in what follows.

\begin{proposition}\label{proposition:A.2}
Fix a principal $H$-bundle $P \xrightarrow{\pi} M$. Let $\sigma :H \rightarrow GL(V)$ be a Lie group representation. Then, $P \times V$ becomes a right $H$-space with the action $(p,v) \cdot h = (ph, \sigma(h)^{-1} v )$. The orbit space of this action, denoted by $P \times_{\sigma} V$ becomes a vector bundle over $M$ with fiber $V$, called the \textbf{bundle associated with $(P, \sigma)$}, whose projection map $P \times_\sigma V \xrightarrow{\xi} M$ is induced from the following commutative diagram.
\begin{equation*}
    \begin{tikzcd}[baseline=(current  bounding  box.center)]
P \times V \arrow["\textup{pr}_1"]{r} \arrow{d}[swap]{\textup{proj.}}
& P \arrow["\pi"]{d}  \\
P \times_\sigma V  \arrow[r, dotted, "\xi"]
& M
    \end{tikzcd}
\end{equation*}

We denote the quivalence class of $(p,v) \in P \times V$ by $[p,v] \in P\times_\sigma V$.
\end{proposition}
\begin{proof}
To each local trivialization $U \times H \xrightarrow{\phi} \pi^{-1} (U)$ for $P$ corresponds a local trivialization $ U \times V \rightarrow \xi^{-1}(U)$ for $P \times_\sigma V$ given by $(x,v) \mapsto [\phi(x,e) , v]$.  
\end{proof}

Let $\mathcal{E}:= P \times_{\sigma} V$. For each $x \in M$ and $p \in \pi^{-1} (x)$, the map $[p, \cdot \hspace{0.1cm} ] : V \rightarrow \mathcal{E}_x$ given by $v \mapsto [p, v]$ is a vector space isomorphism. For each $p \in P$ and $ s \in H$, there is a commutative diagram of vector space isomorphisms
\begin{equation}\label{eq:A.1}
    \begin{tikzcd}[baseline=(current  bounding  box.center)]
V \arrow["\sigma(s)"]{rr} \arrow{dr}[swap]{[p s, \cdot \hspace{0.1cm}]}
& & V \arrow["{[p, \cdot \hspace{0.1cm}]}"]{dl}  \\
&\mathcal{E}_x &
    \end{tikzcd}
\end{equation}
by the definition of $\mathcal{E}$.

\begin{proposition}\label{proposition:A.3}
Let $\mathcal{E}:= P \times_\sigma \mathcal{H}_\sigma$ be an associated bundle where $\sigma : H \rightarrow U((\mathcal{H}_\sigma , \langle \cdot , \cdot \rangle_\sigma)$ is a unitary representation. Then, the map $g: \mathcal{E} \otimes \mathcal{E} \rightarrow \mathbb{C}$ defined in each fiber as
\begin{equation}\label{eq:A.2}
[p,v] \otimes [p,w] \mapsto \langle v,w \rangle_\sigma
\end{equation}
is a well-defined (Hermitian) metric on the bundle $\mathcal{E}$, making $\mathcal{E}$ an Hermitian bundle over $M$. For an associated bundle of this form, we always regard it as an Hermitian bundle endowed with this metric.
\end{proposition}
\begin{proof}
The diagram \eqref{eq:A.1} tells us that we can define an inner product on $\mathcal{E}_x$ by transplanting the inner product of $V$ into $\mathcal{E}_x$ via the map $[p, \cdot \hspace{0.1cm}]$ and it does not depend on the choice of $p \in \pi^{-1} (x)$. That is, the map $\mathcal{E}_x \otimes \mathcal{E}_x  \rightarrow \mathbb{C}$ defined by $[(p,v) ] \otimes [(p,w)] \mapsto \langle v , w \rangle_\sigma$ is a well-defined inner product on $\mathcal{E}_x$. The smoothness is easily checked using local trivializations.  
\end{proof}

\begin{proposition}\label{proposition:A.4}
Let $\mathcal{E}:= P \times_\sigma V$ be an associated bundle. Let
\begin{subequations}\label{eq:A.3}
\begin{align}
C_\sigma (P,V) &:=\left \{ \phi \in C (P,V) : \phi(p s) = \sigma(s)^{-1} \phi (p), \hspace{0.1cm} p \in P, s \in H \right\}, \label{eq:A.3a} \\
C_\sigma ^\infty (P,V) &:=\left \{ \phi \in C^\infty (P,V) : \phi(p s) = \sigma(s)^{-1} \phi (p), \hspace{0.1cm} p \in P, s \in H \right\}, \label{eq:A.3b} \\
\mathcal{B}_\sigma (P,V) &:=\left \{ \phi \in \mathcal{B} (P,V) : \phi(p s) = \sigma(s)^{-1} \phi (p), \hspace{0.1cm} p \in P, s \in H \right\}. \label{eq:A.3c}
\end{align}
\end{subequations}

Then, there are linear isomorphisms
\begin{subequations}\label{eq:A.4}
\begin{align}
&\sharp : C(M, \mathcal{E}) \rightarrow C_\sigma (P,V) \label{eq:A.4a} \\
&\sharp : C^\infty (M, \mathcal{E}) \rightarrow C_\sigma ^\infty (P,V) \label{eq:A.4b} \\
&\sharp : \mathcal{B} (M, \mathcal{E}) \rightarrow \mathcal{B}_\sigma (P,V) \label{eq:A.4c}
\end{align}
\end{subequations}
all given by the same formula
\begin{equation}\label{eq:A.5}
\psi(x) = \left[ p, \psi^\sharp (p) \right]
\end{equation}
where $p \in P$ is any element in the fiber $\pi^{-1} (x)$. The inverses will be denoted as $\flat$, i.e., $\phi^\flat (x) = [p, \phi(p)]$.

Let
\begin{subequations}\label{eq:A.6}
\begin{align}\label{eq:A.6a}
\mathcal{F}_0 := \Big\{ \phi \in \mathcal{B} (P, V) : \phi (ps) =& \sigma(s)^{-1} \phi(p), \hspace{0.1cm} p \in P, s \in H \nonumber \\  &\&  \int_M g( \phi^\flat ( x) , \phi^\flat (x) ) d\mu (x) < \infty \Big\}
\end{align}
and $N := \{ \phi \in \mathcal{F}_0 : \phi^\flat = 0 \textup{ $\mu$-almost everywhere} \} \leq \mathcal{F}_0$, and put
\begin{equation}\label{eq:A.6b}
\mathcal{F} := \mathcal{F}_0 /N
\end{equation}
\end{subequations}

Then, the isomorphism Eq.~(\ref{eq:A.4c}) induces a linear isomorphism
\begin{equation}\label{eq:A.7}
\sharp : L^2(M, \mathcal{E};\mu,g) \rightarrow \mathcal{F}. 
\end{equation}
\end{proposition}
\begin{proof}
Given a section $\psi :M \rightarrow \mathcal{E}$, define $\psi^\sharp (p) = \big([p, \cdot \hspace{0.1cm}]\big)^{-1} ( \psi (\pi (p) ) )$. Then, for $p \in P$ and $ s \in G$,
\begin{equation*}
\psi^\sharp (p s ) = \big([ps , \cdot \hspace{0.1cm}]\big)^{-1} (\psi (\pi (ps) )) = \sigma(s)^{-1} \psi(p).
\end{equation*}

Also, given a map $ \phi :P \rightarrow V$ satisfying $\phi(ps) = \sigma(s)^{-1} \phi(p)$ for ${}^\forall p \in P, {}^\forall s \in H$, the expression $\phi^\flat (x) = [p, \phi(p)]$ does not depend on the choice of $p \in \pi^{-1}(x)$ and gives a well-defined section $\phi^\flat : M \rightarrow \mathcal{E}$.

Using smooth local trivializations for $P$ and corresponding trivializations for $\mathcal{E}$ (cf. the proof of Proposition~\ref{proposition:A.2}), it is an easy matter to check that $\sharp$ and $\flat$ preserve continuity, smoothness, and Borel measurability. Also, it is easy to see that the two maps are inverses to each other. The rest is a straightforward calculation.  
\end{proof}

\paragraph{Associated Hermitian \texorpdfstring{$G$}{TEXT}-bundles}

\hfill

Let $G$ be a Lie group and $H \leq G$ be a closed subgroup. Consider the principal $H$-bundle $G \rightarrow G/H$. If $\sigma : H \rightarrow U(\mathcal{H}_\sigma)$ is a unitary representation, then the associated bundle $G \times_\sigma \mathcal{H}_\sigma \rightarrow G/H$ is an Hermitian bundle over $G/H$ with the metric $g$ given by Eq.~\ref{eq:A.2} according to Proposition~\ref{proposition:A.3}.

\begin{proposition}\label{proposition:A.5}
The bundle $G \times_\sigma \mathcal{H}_\sigma \rightarrow G/H$ becomes an Hermitian $G$-bundle over $G/H $ with the action
\begin{equation}\label{eq:A.8}
\Lambda_x ( [y, v] ) = [xy, v]
\end{equation}
covering the left multiplication map $l_x : G/H \rightarrow G/H$ on the base space.
\end{proposition}

\begin{proof}
The well-definedness and Hermiticity of the action are easily checked.  
\end{proof}

The Hermitian $G$-bundle $(G \times_\sigma \mathcal{H}_\sigma , g, \Lambda)$ over $G/H$ will be denoted as $\mathcal{E}_\sigma$ and called the \textit{primitive bundle associated with $\sigma$}.

\paragraph{Induced representations and associated Hermitian \texorpdfstring{$G$}{TEXT}-bundles}

\hfill

We can rephrase the definition of induced representation, Definition~\ref{definition:4.1}, in terms of induced representation associated with Hermitian $G$-bundles (cf. Definition~\ref{definition:5.3}). Notice the similarity between Eqs.~(\ref{eq:4.1}) and Eqs.~(\ref{eq:A.6}).

\begin{theorem}\label{theorem:A.6}
Suppose $G/H$ has a $G$-invariant measure $\mu$ and consider the primitive bundle $\mathcal{E}_\sigma$ defined in the preceding paragraph. Denote the induced representation associated with $\mathcal{E}_\sigma$ as $U$. Then, the isomorphism Eq.~(\ref{eq:A.7}) gives a unitary equivalence between $\textup{Ind}_H ^G \sigma$ and $U$. I.e., for $x  \in G$ and $ \psi \in L^2 ( G/H , \mathcal{E}_\sigma ; \mu , g )$,
\begin{equation}\label{eq:A.9}
\textup{Ind}_H ^G \sigma (x) \circ \sharp (\psi) = \sharp \circ U (x) (\psi).
\end{equation}
\end{theorem}
\begin{proof}
\begin{align*}
 \left( \flat \circ \textup{Ind}_H ^G \sigma (x) \circ \sharp (\psi) \right) (yH) &= \left[y, \left( \textup{Ind}_H ^G \sigma (x) \psi^\sharp \right) ( y) \right] \\
= \left[ y, \psi^\sharp (x^{-1} y) \right] &= L_x (\psi ( x^{-1} yH) ) = \left( U (x) \psi \right) (yH)
\end{align*} 
\end{proof}

Thus, in what follows, when speaking of an induced representation, we always mean $U$ and denote it as $\textup{Ind}_H ^G \sigma$.

\subsection{Perception bundle}\label{sec:A.2}

Consider the principal $H$-bundle $G \rightarrow G/H$.
\begin{theorem}\label{theorem:A.7}
Suppose that a unitary representation $\sigma : H \rightarrow U(\mathcal{H}_\sigma , \langle \cdot , \cdot \rangle_\sigma) $ extends to a (non-unitary) representation $\Phi : G \rightarrow GL (\mathcal{H}_{\sigma})$. Then, the bundle $\mathcal{E}_\sigma = G \times_{\sigma} \mathcal{H}_\sigma$ is trivial. In fact, there is a vector bundle isomorphism
\begin{equation}\label{eq:A.10}
\begin{tikzcd}[baseline=(current  bounding  box.center)]
\mathcal{E}_\sigma \arrow[rr, "{[x, v] \mapsto (xH , \Phi(x) v)}"] \arrow{dr}
&  & G/H \times \mathcal{H}_\sigma \arrow{dl} \\
& G/H &
\end{tikzcd}.
\end{equation}

Via this isomorphism, the Hermitian metric and $G$-action on $\mathcal{E}_\sigma$ are translated into the metric
\begin{equation}\label{eq:A.11}
h_{xH} (v, w ) = \langle v, \Phi (x)^{\dagger -1} \Phi (x)^{-1} w \rangle _\sigma
\end{equation}
(here $\dagger$ is the adjoint operation on the algebra of continuous operators on $\mathcal{H}_\sigma$) and the $G$-action
\begin{equation}\label{eq:A.12}
\lambda_x ( yH , v ) = (xy H, \Phi(x) v )
\end{equation}
on the RHS bundle, with respect to which the isomorphism Eq.~(\ref{eq:A.10}) becomes an Hermitian $G$-bundle isomorphism. The Hermitian $G$-bundle $(G/H \times \mathcal{H}_\sigma , h , \lambda)$ over $G/H$ will be denoted as $E_\sigma$.
\end{theorem}

\begin{proof}
The map is well-defined since $\Phi|_H = \sigma$. It is an isomorphism at each fiber and hence a vector bundle isomorphism.

Also, note that for $ x \in G$ and $k \in H$,
\begin{equation*}
\Phi (xk)^{\dagger -1} \Phi (xk)^{-1} = \Phi (x) ^{\dagger -1} \sigma(k) ^{\dagger -1} \sigma (k)^{-1} \Phi(x)^{-1}= \Phi (x)^{\dagger -1} \Phi (x)^{-1}
\end{equation*}
since $\sigma$ is unitary. So, $h$ is a well-defined sesquilinear form at each fiber and it is easy to check that the map Eq.~(\ref{eq:A.10}) becomes a unitary map at each fiber with respect to these metrics. The statement about the actions is easy.  
\end{proof}

In the case of $\sigma_s$, which extends to $\Phi_s$ (cf. Eqs.~(\ref{eq:4.17})--(\ref{eq:4.18})), the bundle $E_s:=E_{\sigma_s}$ will be called the \textit{perception bundle for massive spin-$s$ particles} since the fibers of it will be shown to be "the moving spin systems as perceived by a fixed inertial observer", as we shall see in Sect.~\ref{sec:6}.

\subsection{Boosting bundle}\label{sec:A.3}

The method of Sect.~\ref{sec:A.2} is not the only way to show that the associated bundle $G\times_\sigma \mathcal{H}_\sigma$ is trivial. Let $P \xrightarrow{\pi} M$ be a principal $G$-bundle. Given a local section $s: U \rightarrow P$, we know that the map $U \times G \rightarrow P|_U$ given by
\begin{equation}\label{eq:A.13}
(x,g) \mapsto s(x) g
\end{equation}
is a local trivialization of $P$ (cf. \cite{tu}).

\begin{lemma}\label{lemma:A.8}
Let $\sigma : G \rightarrow U(V)$ be a representation. Given a (smooth) local section $s : U \rightarrow P$, the map
\begin{equation}\label{eq:A.14}
\begin{tikzcd}[baseline=(current  bounding  box.center), column sep=1.5em]
    U \times V \arrow ["{(x,v) \mapsto [s(x), v]}"]{rrrr} \arrow{drr} 
   & & & & (P \times_\sigma V)|_U \arrow{dll}
\\
   & &U & &
\end{tikzcd}
\end{equation}
is a (smooth) local trivialization of $P \times_\sigma V$.
\end{lemma}
\begin{proof}
Easy. 
\end{proof}

So, in particular, if there is a global section $s : M \rightarrow P$, then the bundles $P$ and $E$ are trivial. Let's specialize this to the case of the principal $H$-bundle $G \rightarrow G/H$.

\begin{theorem}\label{theorem:A.9}
Let $\sigma :H \rightarrow U\big((\mathcal{H}_\sigma , \langle \cdot , \cdot \rangle_\sigma)\big)$ be a unitary representation and $L : G/H \rightarrow G$ be a global section. If the trivial bundle $G/H \times \mathcal{H}_\sigma$ is endowed with the metric
\begin{equation}\label{eq:A.15}
h_{L} ( (xH , v) , (xH, w) ) = \langle v , w \rangle_\sigma
\end{equation}
and the $G$-action
\begin{equation}\label{eq:A.16}
\lambda_{L} (x) ( yH, v) = (xyH, \sigma\left( L(xyH)^{-1} x L(yH) \right) v),
\end{equation}
then it becomes an Hermitian $G$-bundle over $G/H$ and the global trivialization
\begin{equation}\label{eq:A.17}
\begin{tikzcd}[baseline=(current  bounding  box.center), column sep=1.5em]
    G/H \times \mathcal{H}_\sigma \arrow ["{(xH,v) \mapsto [L(xH), v]}"]{rrrr} \arrow{drr} 
    & & & & \hspace{0.2cm} \mathcal{E}_\sigma \arrow{dll}
\\
   & &M & &
\end{tikzcd}
\end{equation}
becomes an Hermitian $G$-bundle isomorphism onto the primitive bundle $\mathcal{E}_\sigma$. The Hermitian $G$-bundle $(G/H \times \mathcal{H}_\sigma, h_L , \lambda_{L,\sigma})$ will be denoted as $E_{L,\sigma}$, signifying its dependence on the choice of section $L$.
\end{theorem}
\begin{proof}
Straightforward. 
\end{proof}

We call $L$ a \textit{choice of boostings} for reasons that will become clear in Sect.~\ref{sec:6}. Note that it has a close relationship with the choice of gauge in Gauge Theory (cf. \cite{bleecker}). $E_{L, \sigma}$ will be called the \textit{boosting bundle associated with $L$}.

The element
\begin{equation}\label{eq:A.18}
W_L (x, yH) = L(xyH)^{-1} x L(yH) \in H
\end{equation}
will be called the \textit{Wigner transformation} since it is preciesely the \textit{Wigner rotation} in the case of $G =SL(2, \mathbb{C})$ and $H= SU(2)$ when $L$ is chosen appropriately (cf. Sect.~\ref{sec:6}). Using this notation, we see that the action Eq.~(\ref{eq:A.16}) becomes
\begin{equation}\label{eq:A.16'}\tag{A.16${}^\prime$}
\lambda_{L} (x) (yH, v) = (xyH, \sigma \big(W_L (x, yH) \big) v).
\end{equation}

\subsection{Semidirect products}\label{sec:A.4}

We apply the preceding constructions to the case of semidirect products.

Let $G$ be a Lie group and $N,H \leq G$ closed subgroups such that $N$ is normal and abelian and $G = N \ltimes H$, i.e., the map $N \times H \rightarrow G$ given by $(n,h) \mapsto nh$ is a diffeomorphism. Since continuous homomorphisms between Lie groups are automatically smooth (cf. \cite{lee}, Ch.~20), $\hat{N}$ consists of Lie group homomorphisms from $N$ to $\mathbb{T}$.

The following lemma shows that the vector bundles associated with the principal bundle $H \rightarrow H/H_\nu$ for $\nu \in \hat{N}$ not only have a (left) $H$-action provided by Eq.~(\ref{eq:A.8}), but also a (left) $G$-action which extends it.

\begin{lemma}\label{lemma:A.10}
Fix $\nu \in \hat{N}$ and a unitary representation $\sigma :H_\nu \rightarrow U(\mathcal{H}_\sigma)$, which induces a unitary representation $\nu \sigma : G_\nu \rightarrow U(\mathcal{H}_\sigma)$ as in Eq.~(\ref{eq:4.3}). Consider the principal $G_\nu$-bundle $G \rightarrow G/G_{\nu}$ and the principal $H_\nu$-bundle $H \rightarrow H/H_\nu$. Then, there is an $H$-equivariant isometric bundle isomorphism
\begin{equation}\label{eq:A.19}
    \begin{tikzcd}[baseline=(current  bounding  box.center)]
\mathcal{E}_\sigma \arrow[rr, "{\jmath :[h,v] \mapsto [h,v]}"] \arrow{d}
& &  \mathcal{E}_{\nu \sigma} \arrow{d} \\
H/H_\nu \arrow{rr}{\imath: hH_\nu \mapsto hG_\nu}[swap]{\cong} & &  G/ G_\nu
    \end{tikzcd}
\end{equation}
whose inverse is given by $[nh,v] \mapsto [h , \nu ( h^{-1} n h ) v]$. By pulling-back the $G$-action on $\mathcal{E}_{\nu \sigma}$ via this map, the $H$-action on the bundle $\mathcal{E}_\sigma$ is extended to a $G$-action which is given by
\begin{equation}\label{eq:A.20}
\Lambda (nh)[k,v]= \left[ hk, \nu \left( (hk)^{-1} n hk \right)v \right],
\end{equation}
with respect to which $\mathcal{E}_\sigma$ becomes an Hermitian $G$-bundle over $H/H_\nu$.
\end{lemma}

\begin{proof}
$\imath$ is well-defined and injective since $h G_\nu = h' G_\nu \Leftrightarrow h^{-1 }h' \in G_\nu \cap H = H_\nu$. Also, given $nh \in G$, we have $h^{-1} nh \in N \leq G_\nu$ and hence $\imath(h H_\nu)= h G_\nu = nh G_\nu$, which implies that $\imath$ is surjective as well. Since it is an $H$-equivariant map from a transitive $H$-space, it is a diffeomorphism (cf. \cite{lee}). So, $\imath:H/H_\nu \cong G/G_\nu$.

$\jmath$ is easily seen to be a well-defined vector bundle homomorphism covering $\imath$, which is an isomorphism at each fiber and hence a vector bundle isomorphism. This also preserves the metrics at each fiber essentially by definition and is trivially $H$-equivariant. The remaining statements are now easily checked. 
\end{proof}

Upon identifying $H/H_\nu \cong G/G_\nu$, we can apply Proposition~\ref{proposition:5.4} and Theorem~\ref{theorem:A.6} to represent $\textup{Ind}_{G_\nu} ^G \nu \sigma$ on the Hermitian $G$-bundle $\mathcal{E}_\sigma$.

\begin{lemma}\label{lemma:A.11}
In addition to the same setting of Lemma~\ref{lemma:A.10}, suppose $H/H_\nu \cong G/ G_\nu$ has an $H$-invariant (and hence $G$-invariant) measure $\mu$. Then, the induced representation $\textup{Ind}_{G_\nu} ^G \nu \sigma$ is equivalent to the induced representation associated with the primitive bundle $\mathcal{E}_\sigma$ (cf. Definition~\ref{definition:5.3}), which is denoted as $U : G \rightarrow U(L^2 (H/ H_\nu , \mathcal{E}_\sigma ; \mu , g))$ and defined as, for $nh \in G$ and $\psi \in L^2(H/H_\nu , \mathcal{E}_\sigma ; \mu , g)$,
\begin{equation}\label{eq:A.21}
U (nh) \psi = \Lambda(nh) \circ \psi \circ (l_{nh})^{-1} = \Lambda(nh) \circ \psi \circ (l_h)^{-1}.
\end{equation}
\end{lemma}
\begin{proof}
This follows from Proposition~\ref{proposition:5.4}, Theorem~\ref{theorem:A.6} and the fact that $l_h = l_{nh}$ on $H/H_\nu$. 
\end{proof}

Now, we are prepared to prove Theorem~\ref{theorem:5.5}.
\begin{proof}[Proof of Theorem~\ref{theorem:5.5}]
All the conclusions are straightforward from Lemmas~\ref{lemma:A.10}--\ref{lemma:A.11} and Theorems~\ref{theorem:A.7}--\ref{theorem:A.9}. Note that since $h_L$ is the trivial metric, we have $\mathcal{H}_{U_L}:= L^2 (H/H_\nu , E_{L, \sigma} ; \mu , h_L) \cong L^2 (H/H_\nu;\mu) \otimes \mathcal{H}_\sigma$. 
\end{proof}

\section*{Acknowledgments}

H. Lee was supported by the Basic Science Research Program through the National Research Foundation of Korea (NRF) Grant NRF-2022R1A2C1092320.

\section*{Data availability statement}

Data sharing is not applicable to this article as no new data were created or analyzed in this study.

\section*{References}
\bibliographystyle{plain}
\bibliography{bibliography}
\end{document}